\documentclass{article}
\pdfoutput=1
\usepackage[numbers,sort&compress]{natbib}
\usepackage[utf8]{inputenc}
\usepackage{amsmath,amssymb,amsthm}
\usepackage{paralist}
\usepackage{subfig}
\usepackage{graphicx}
\usepackage{microtype}
\usepackage{authblk}

\usepackage[ruled,linesnumbered]{algorithm2e}

\usepackage{mathtools}
\newcommand{\pt}{po\-ly\-nom\-i\-al-time}

\usepackage[pdfdisplaydoctitle, colorlinks, linkcolor=red!40!black, citecolor=red!40!black, bookmarks=false]{hyperref}

\makeatletter
\def\NAT@spacechar{~}%
\makeatother

\newtheorem{theorem}{Theorem}
\newtheorem{lemma}{Lemma}
\newtheorem{claim}{Claim}

\newtheorem{corollary}{Corollary}

\theoremstyle{definition}
\newtheorem{property}{Property}
\newtheorem{construction}{Construction}
\newtheorem{definition}{Definition}
\newtheorem{observation}{Observation}

\usepackage[nameinlink]{cleveref}

\crefname{table}{Table}{Tables}
\crefname{figure}{Figure}{Figures}
\crefname{theorem}{Theorem}{Theorems}
\crefname{definition}{Definition}{Definitions}
\crefname{corollary}{Corollary}{Corollaries}
\crefname{lemma}{Lemma}{Lemmas}
\crefname{example}{Example}{Examples}
\crefname{reduction}{Reduction}{Reductions}
\crefname{claim}{Claim}{Claims}
\crefname{proposition}{Proposition}{Propositions}
\crefname{property}{Property}{Properties}
\crefname{observation}{Observation}{Observations}
\crefname{construction}{Construction}{Constructions}

\crefname{subsection}{Subsection}{Subsections}
\crefname{section}{Section}{Sections}
\crefname{algorithm}{Algorithm}{Algorithms}

\usepackage{tikz}
\usetikzlibrary{decorations,arrows,petri,topaths,backgrounds,shapes,positioning,fit,calc,decorations.pathreplacing,patterns,intersections}
\usepackage{tikz-qtree}

\newcommand{\decprob}[3]{%
 {\def\descriptionlabel##1{\hspace\labelsep\quad{}\it{}##1}%
 \par\vspace{\topsep}\noindent%
 \begin{compactdesc}
 \item[\textsc{#1}]
 \item[Input:] #2
 \item[Question:] #3
\end{compactdesc}
}\vspace{\topsep}}

\newcommand{\probStarPart}{\textsc{Star Partition}\xspace}
\newcommand{\probPthreePart}{\textsc{$P_3$-Partition}\xspace}
\newcommand{\probMinPkPart}{\textsc{Min $P_k$-Partition}\xspace}
\newcommand{\probExCover}{\textsc{Exact Cover by $s$-Sets}\xspace}
\newcommand{\probThreeDMatching}{\textsc{3-Dimensional Matching}\xspace}

\newcommand{\TDMR}{\ensuremath{R}}
\newcommand{\TDMB}{\ensuremath{B}}
\newcommand{\TDMY}{\ensuremath{Y}}
\newcommand{\TDMT}{\ensuremath{T}}
\newcommand{\TDMSol}{\ensuremath{M}}
\newcommand{\TDMInst}{\ensuremath{(\TDMR, \TDMB, \TDMY, \TDMT\subseteq \TDMR \times \TDMB \times \TDMY)}}

\newcommand{\sstar}{\ensuremath{s}-star\xspace}
\newcommand{\sstars}{\ensuremath{s}-stars\xspace}
\newcommand{\starpartition}[1]{\ensuremath{#1}-star partition\xspace}

\newcommand{\classNP}{\textrm{NP}\xspace}
\newcommand{\classP}{\textrm{P}\xspace}

\DeclareMathOperator{\argmax}{arg\,max}
\newcommand{\h}[1]{\ensuremath{T_{#1}}}
\newcommand{\A}[1]{\ensuremath{A_{#1}}}
\newcommand{\Psol}[0]{\ensuremath{\mathcal P}}
\newcommand{\Qsol}[0]{\ensuremath{\mathcal Q}}
\newcommand{\lessgood}[0]{\ensuremath{\preccurlyeq}}
\newcommand{\occ}[2]{\ensuremath{\vert #2 \vert_{#1}}}
\newcommand{\qSize}[1]{\ensuremath{\Vert #1\Vert}}
\newcommand{\qAdd}[0]{\oplus}
\newcommand{\qDel}[0]{\ominus}
\DeclareMathOperator{\death}{end}
\DeclareMathOperator{\birth}{start}
\DeclareMathOperator{\rank}{rank}
\DeclareMathOperator{\tokens}{tokens}
\newcommand{\strongorder}{\ensuremath{\prec}}
\newcommand{\myll}{\strongorder}
\newcommand{\strongordereq}{\ensuremath{\preccurlyeq}}
\newcommand{\cross}{\ensuremath{\times}}
\DeclareMathOperator{\scenter}{center}
\DeclareMathOperator{\closure}{scope}
\DeclareMathOperator{\width}{width}
\DeclareMathOperator{\leftm}{lm}
\DeclareMathOperator{\rightm}{rm}

\DeclareMathOperator{\edgecrossingnum}{\#edge-crossings}
\DeclareMathOperator{\dist}{d}
\DeclareMathOperator{\union}{\oplus}
\DeclareMathOperator{\join}{\otimes}
\tikzstyle{birth}=[circle, inner sep=1.0pt, draw=black!70,
fill=black!70] 
\tikzstyle{death}=[circle, inner sep=1.0pt,
draw=black!70, fill=white] 
\tikzstyle{single}=[circle, inner
sep=1.5pt, draw=black!70, fill=black!20]
\tikzstyle{table}=[draw=none] 
\tikzstyle{thickedge}=[draw=black, line width=1.2pt]
\tikzstyle{intervalline}=[draw=black, line width=0.8pt]
\tikzstyle{vertex} = [align=center, circle, inner sep=1.5pt, 
fill=black!30, draw=black]
\tikzstyle{center} = [circle,draw=black, fill=black, inner
sep=1.5pt]
\newcommand{\mytitle}{Partitioning Perfect Graphs into Stars\thanks{An
    extended abstract of this work appeared under the title ``Star
    Partitions of Perfect Graphs'' in Proceedings of the 41st
    International Colloquium on Automata, Languages, and Programming
    (ICALP'14), Part~I, LNCS 8572, pp.~174--185, Springer, 2014.}}

\title{\mytitle}

\date{}
\author[1,2]{René van Bevern}
\author[2]{Robert Bredereck}
\author[3]{Laurent Bulteau}
\author[2]{Jiehua~Chen}
\author[2]{Vincent Froese}
\author[2]{Rolf Niedermeier}
\author[4]{Gerhard~J.~Woeginger}

\affil[1]{Novosibirsk State University, Novosibirsk, Russian Federation, \texttt{rvb@nsu.ru}}

\affil[2]{Institut f\"ur Softwaretechnik und Theoretische Informatik,
  TU Berlin, Germany, \texttt{\{robert.bredereck,jiehua.chen,vincent.froese,rolf.niedermeier\}@tu-berlin.de}}

\affil[3]{Institut Gaspard-Monge,
  Université Paris-Est Marne-la-Vallée, 
  France
\texttt{l.bulteau@gmail.com}}
\affil[4]{Department of Mathematics and Computer Science,
TU Eindhoven, The~Netherlands, \texttt{gwoegi@win.tue.nl}}
\begin{document}
\maketitle

\begin{abstract}
  \noindent The partition of graphs into ``nice'' subgraphs is a central algorithmic
  problem with strong ties to matching theory.  We study the
  partitioning of undirected graphs into same-size stars, 
  a problem known to be
  NP-complete even for the case of stars on three vertices.  We perform
  a thorough computational complexity study of the problem on subclasses
  of perfect graphs and identify several \pt{} solvable
  cases, for example, on interval graphs and bipartite permutation graphs, and
  also NP-complete cases, for example, on grid graphs and chordal graphs.
\end{abstract}

\section{Introduction}
We study the computational complexity (tractable versus intractable
cases) of the following basic graph problem.

\decprob{\probStarPart}
{An undirected $n$-vertex graph~$G=(V,E)$ and an integer~$s\in\mathbb N$.}
{Can the vertex set~$V$ be partitioned into $k\coloneqq\lceil
  n/(s+1)\rceil$ mutually disjoint vertex subsets~$V_1,$ $V_2,\ldots,V_k$, such that each
  subgraph~$G[V_i]$ contains an $s$-star~(a~$K_{1,s}$)?}

\noindent Two prominent special cases of \probStarPart are the case $s=1$
(finding a perfect matching) and the case $s=2$ (finding a partition
into connected triples).  Perfect matchings ($s=1$), of course, can be
found in polynomial time.  Partitions into connected triples (the
case~$s=2$), however, are hard to find; this problem, denoted
\probPthreePart, was proven to be NP-complete by \citet{KH83}.

\makeatletter \newcommand{\gettikzxy}[3]{%
  \tikz@scan@one@point\pgfutil@firstofone#1\relax \edef#2{\the\pgf@x}%
  \edef#3{\the\pgf@y}%
} \makeatother

\begin{figure}[t]
  \centering \def\dist{1pt} \def\ldist{4pt} \def\lldist{8pt}
  \resizebox{\textwidth}{!}{
    \begin{tikzpicture}[->,>=stealth', shorten >=2pt, edge from parent
      path={(\tikzparentnode) -- (\tikzchildnode)}]
      \tikzset{level distance = 1.25cm}

      \tikzstyle{treenode} = [align=center, inner sep=2pt, text
      centered] \tikzstyle{hardnode} = [treenode, rectangle, minimum
      height=.6cm, rounded corners, draw, fill=white, very thick]

      \tikzstyle{unknownnode} = [rounded corners, treenode, minimum
      height = .6cm, rectangle,draw] \tikzstyle{normalnode} =
      [treenode, rectangle, minimum height = .6cm, draw, rounded
      corners, very thick, fill=white] \tikzstyle{easynode} =
      [treenode, rectangle, minimum height = .6cm, rounded corners,
      draw, fill=white, very thick] \tikzstyle{half_easynode_extern} =
      [treenode, rectangle split, rectangle split parts=2, minimum
      height = .6cm, rounded corners, draw, rectangle split part
      fill={none, gray!20}, rectangle split draw splits=false]
      \tikzstyle{half_easynode} = [treenode, rectangle split,
      rectangle split parts=2, minimum height = .6cm, rounded corners,
      draw, rectangle split part fill={white, white}, rectangle split
      draw splits=false, very thick] \tikzstyle{easynode_extern} =
      [treenode, rectangle, minimum height = .6cm, rounded corners,
      draw, fill=white] \tikzstyle{hardnode_extern} = [treenode,
      rectangle, minimum height = .6cm, rounded corners, draw,
      fill=white] \tikzstyle{hardback} = [-,red!28, line
      width=0.4pt,rounded corners,top color=red!20,bottom
      color=red!20] \tikzstyle{easyback} = [-,green!28, line
      width=0.4pt,rounded corners,top color=green!20,bottom
      color=green!20] \Tree [ .\node[hardnode_extern] (perfect)
      {perfect}; [ .\node[hardnode_extern] (comparability)
      {comparability}; [ .\node[hardnode_extern] (bipartite)
      {bipartite~\cite{CP14}}; [ .\node[hardnode_extern](planar_bip) {subcubic
        planar\\
        bipartite~\cite{MZ05,MT07}}; [
      .\node[hardnode](grid){subcubic\\grid}; ] ] [ .\node[text width
      = 10ex, unknownnode] (chordal_bip) {chordal bipartite}; [
      .\node[easynode, text width = 13ex] (bip_perm) {bipartite
        permutation}; ] ] ] [ .\node[unknownnode,xshift=-1em] (perm)
      {permutation}; [ .\node[easynode] (cograph) {cograph}; [
      .\node[easynode] (triv_perfect) {trivially\\perfect}; [
      .\node[easynode] (threshold) {threshold\\{\small ($P_3$
          known~\cite{YCC96})}}; ] ] ] ] ] [ .\node[hardnode]
      (chordal) {chordal}; [ .\node[normalnode] (inter) {interval}; [
      .\node[easynode, text width = 8ex,xshift=1em] (unit_inter) {unit
        interval}; ] ] \edge[draw=none]; {\phantom{Claw-free}} [
      .\node[normalnode] (split) {split}; ] ] ] \node[easynode_extern,
      below = 6pt of grid, xshift=-5mm] (ser_para)
      {series-parallel~\cite{TNS82}}; \node[easynode_extern, below =
      9pt of bip_perm] (tree) {tree};

      \path[->] (ser_para) edge (tree); \path[->] (bip_perm) edge
      (tree); \path[->] (inter) edge (triv_perfect.north east); \path[->] (perm)
      edge (bip_perm.35); \draw[->] (split)--(threshold.east);

      \begin{pgfonlayer}{background}
        \gettikzxy{(perfect.north)}{\perfectnorthx}{\perfectnorthy};
        \gettikzxy{(ser_para.west)}{\serparawestx}{\serparawesty};
        \gettikzxy{(split.east)}{\spliteastx}{\spliteasty};
        \gettikzxy{(threshold.south)}{\thresholdsouthx}{\thresholdsouthy};
        \gettikzxy{(chordal.south)}{\chordalsouthx}{\chordalsouthy};
        \gettikzxy{(chordal.east)}{\chordaleastx}{\chordaleasty};
        \gettikzxy{(unit_inter.north)}{\uinorthx}{\uinorthy};
        \gettikzxy{(inter.west)}{\iwestx}{\iwesty};
        \gettikzxy{(comparability.south
          east)}{\compsoutheastx}{\compsoutheasty};
        \gettikzxy{(cograph.north
          west)}{\conorthwestx}{\conorthwesty};
        \gettikzxy{(tree.north)}{\treenorthx}{\treenorthy};
        \gettikzxy{(chordal_bip.south
          east)}{\chordbipsoutheastx}{\chordbipsoutheasty};
        \gettikzxy{(ser_para.north east)}{\serparanex}{\serparaney};
        \gettikzxy{(bipartite.east)}{\bipex}{\bipey};
        \gettikzxy{(bipartite.south)}{\bipsx}{\bipsy};
        \gettikzxy{(planar_bip.east)}{\planarbipex}{\planarbipey};
        \gettikzxy{(planar_bip.south
          east)}{\planarbipsex}{\planarbipsey};
        \gettikzxy{(grid.east)}{\gridex}{\gridey};
        \gettikzxy{(bip_perm.north)}{\bippermnx}{\bippermny};

        \draw[hardback,black]
        ($(\serparawestx-\ldist,\perfectnorthy+\ldist)$) --
        ($(\spliteastx+\ldist, \perfectnorthy+\ldist)$) --
        ($(\spliteastx+\ldist, \iwesty+\dist)$) --
        ($(\chordaleastx+\ldist, \iwesty+\dist)$) --
        ($(\chordaleastx+\ldist, \chordalsouthy-\ldist)$) --
        ($(\compsoutheastx+\ldist, \chordalsouthy-\ldist)$) --
        ($(\bipex+\lldist, \chordalsouthy-\ldist)$) --
        ($(\bipex+\lldist, \bipsy-\ldist)$) -- ($(\planarbipex+\ldist,
        \bipsy-\ldist)$) -- ($(\planarbipex+\ldist,
        \planarbipsey-\ldist)$) --
        ($(\gridex+\ldist, \planarbipsey-\ldist)$) -- ($(grid.east) +
        (\ldist, 0)$) -- ($(\gridex+\ldist, \serparaney+\ldist)$) --
        ($(ser_para.north west) + (-\ldist, \ldist)$) -- cycle;
     
        \draw[easyback,black] ($(ser_para.north west) + (-\ldist,
        2*\dist)$) -- ($(ser_para.north east) + (\ldist, 2*\dist)$) --
        ($(bip_perm.west) + (-\ldist, 0)$) -- ($(bip_perm.north west)
        + (-\ldist, \lldist)$) -- ($(bip_perm.north) + (0, \lldist)$)
        -- ($(\conorthwestx-\ldist, \bippermny+\lldist)$) --
        ($(\conorthwestx-\ldist, \uinorthy+\lldist)$) --
        ($(\iwestx-\ldist, \uinorthy+\lldist)$) -- ($(inter.west) +
        (-\ldist, -\dist)$) -- ($(\spliteastx+\ldist, \iwesty-\dist)$)
        -- ($(\spliteastx+\ldist, \thresholdsouthy-\ldist)$) --
        ($(threshold.south) + (0, -\ldist)$) --
        ($(\serparawestx-\ldist, \thresholdsouthy-\ldist)$) -- cycle;

        \node[] at ($(\serparawestx+\lldist*6,
        \perfectnorthy-\lldist)$) (nph) {\textbf{\classNP-complete}};
        \node[] at ($(\serparawestx+\lldist*6,
        \thresholdsouthy+\ldist)$) (p) {\textbf{\classP}};
      \end{pgfonlayer}
    \end{tikzpicture}
  }
  \caption{Complexity classification of \probStarPart{}. Bold borders
    indicate results of this paper. An arrow from a class~$A$ to a
    class~$B$ indicates that $A$~contains~$B$. In most classes,
    \classNP-completeness results hold for~$s=2$ (that is, for
    \probPthreePart). However, on split graphs, \probStarPart{} is
    \pt{} solvable for~$s\le 2$, while it is
    \classNP-complete for~$s\ge 3$. \probPthreePart{} is solvable on
    interval graphs in quasilinear time. We are not aware of any
    result for permutation graphs, chordal bipartite graphs,
    interval graphs (for $s\ge 3$), or grid graphs (for~$s=3$).}
  \label{fig-overview}

\end{figure}
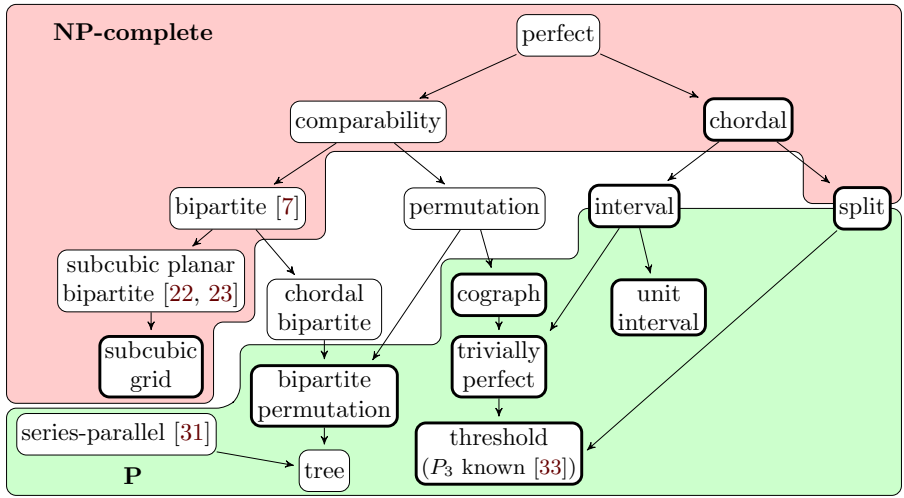

Our goal in this paper is to achieve a better understanding of star
partitions for certain subclasses of perfect graphs.  We provide a fairly
complete classification in terms of \pt{} solvability versus
NP-completeness on the most prominent subclasses of perfect graphs,
leaving a few potentially challenging cases open; see
\cref{fig-overview} for an overview of our results.

\paragraph{Motivation.}
\looseness=-1 The literature in algorithmic graph theory is full of packing and
partitioning problems (packing is an optimization variant of partitioning,
where one tries to maximize the number of disjoint vertex subsets).
Concerning practical relevance, note that
\textsc{$P_3$-Packing} and \textsc{$P_3$-Partition} find applications
in dividing distributed systems into subsystems~\cite{KMZ06} as well
as in the \textsc{Test Cover} problem arising in
bioinformatics~\cite{BHHHLRS03}.  In particular, the application in
distributed systems explicitly motivates the consideration of very
restricted (perfect) graph classes such as grid-like structures.
\probStarPart{} on grid graphs naturally occurs in political
redistricting problems~\cite{BBCFNW15}.  We show that \probStarPart{}
remains NP-hard on subcubic grid graphs.

Interval graphs are a further famous class of perfect graphs.  Here,
\probStarPart can be considered a team formation problem: Assume that
we have a number of agents, each being active during a certain time
interval. Our goal is to form teams, all of the same size, such that each
team contains at least one agent sharing time with every other team
member. This specific team member becomes the team leader, since he or she
can act as an information hub. Forming such teams is nothing else than
solving \probStarPart on interval graphs.  We present efficient
algorithms for \probStarPart on unit interval graphs (that is, for the
case when all agents are active for the same amount of time) and for
\probPthreePart{} on general interval graphs.

\paragraph{Related work.}
\looseness=-1 Packing and partitioning problems are central problems in algorithmic graph theory with many applications and with close connections to matching theory~\cite{Yus07}.  In the case of packing, one wants to maximize the number of graph vertices that are ``covered'' by vertex-disjoint copies of some fixed pattern graph~$H$.  In the case of partitioning, one wants to cover \emph{all} vertices in the graph.  We focus on the partitioning problem, which is also called \textsc{$H$-Factor} in the literature. In this work, we always refer to it as \textsc{$H$-Partition}.  Since \citet{KH83} established the NP-completeness of \textsc{$H$-Partition} on general graphs for every connected pattern~$H$ with at least three vertices, one branch of research has turned to the investigation of classes of specially structured graphs.  For instance, on the upside, \textsc{$H$-Partition} has been shown to be \pt{} solvable on trees and series-parallel graphs~\cite{TNS82} and on graphs of maximum degree two~\cite{MT07}.  On the downside, \textsc{$P_k$-Partition} (for each fixed~$k\ge3$) remains NP-complete on planar bipartite graphs~\cite{DF85}; this hardness result is generalized to \textsc{$H$-Partition} on planar graphs for any outerplanar pattern~$H$ with at least three vertices~\cite{BJLSS90}.  For every fixed~$s\geq 2$, \probStarPart{} is NP-complete on bipartite graphs~\cite{CP14}.  Partitioning into triangles~($K_3$), that is, \textsc{$K_3$-Partition}, is \pt{} solvable on chordal graphs~\cite{DK98} and linear-time solvable on graphs of maximum degree three~\cite{RNB13}.

An optimization version of \textsc{$P_k$-Partition}, called \probMinPkPart, has also received considerable interest in the literature.  
This version asks for a partition of a given graph into a minimum number of paths, each of length \emph{at most}~$k$.  
Clearly, all hardness results for \textsc{$P_k$-Partition} carry over to this minimization version.  
If~$k$ is \emph{part of} the input, then \probMinPkPart is hard for cographs~\cite{Stei00} and chordal bipartite graphs~\cite{Stei03}.  In fact, \probMinPkPart is \classNP-complete even on convex graphs and trivially perfect graphs (also known as quasi-threshold graphs), and hence on interval and chordal graphs \cite{AN07}.  \probMinPkPart is solvable in polynomial time on trees \cite{YCHH97}, threshold graphs, cographs (for fixed~$k$)~\cite{Stei00} and bipartite permutation graphs~\cite{Stei03}.

While in this work we study the \textsc{$H$-Partition} problem,
which partitions the vertex set of a graph into mutually vertex-disjoint copies of some fixed pattern graph~$H$, 
the literature also studies the \textsc{$H$-Decomposition} problem,
which partitions the \emph{edge set} of a graph into mutually \emph{edge}-disjoint copies of a pattern~$H$.  
In general, \textsc{$H$-Decomposition} is NP-hard~\citep{CT91},
yet easy to solve on highly-connected graphs if $H$~is a $k$-star: \citet{Tho12} shows that every $(k^2+k)$-edge-connected graph has a $k$-star decomposition provided its number of edges is a multiple of~$k$. \citet{LTWZ13} strengthen this result to $(3k-3)$-edge-connected graphs for odd~$k\geq 3$.  However, since a graph may have a $k$-star decomposition without having a $k$-star partition and vice versa, the results on \textsc{$H$-Decomposition} are not applicable to the \probStarPart{} problem considered in our work.

\paragraph{Our contributions.}
So far, surprisingly little was known about the complexity of
\probStarPart for subclasses of perfect graphs.  We provide a detailed
picture of the corresponding complexity landscape for classes of perfect graphs; see
\cref{fig-overview} for an overview.
Let us briefly summarize our major findings.
(Note that all problem variants we consider are clearly contained in
NP, which means that our NP-hardness results in fact imply NP-completeness.)

As a central result, we provide a quasilinear-time algorithm for
\probPthreePart (which is \probStarPart with $s=2$) on interval
graphs; the complexity of \probStarPart
for $s\geq 3$ remains open.
But if we restrict the input graphs to be unit interval graphs or trivially perfect graphs, 
we can solve \probStarPart even in linear time.
Furthermore, we develop a \pt{}
algorithm for \probStarPart on cographs and on bipartite permutation graphs.  
Most of our \pt{} algorithms are simple to describe: 
they are based on dynamic
programming or even on greedy approaches, and hence should work well
in implementations.
Their correctness proofs, however, are intricate.

\looseness=-1 On the boundary of NP-completeness, we strengthen a result
of \citet{MZ05} and \citet{MT07} by showing that \probPthreePart is
NP-hard on grid graphs with maximum degree three.  Note that in
strong contrast to this, \textsc{$K_3$-Partition} is linear-time
solvable on graphs with maximum degree three~\cite{RNB13}.
Furthermore, we show \probPthreePart to be NP-hard on chordal
graphs, while \textsc{$K_3$-Partition} is known to be \pt{}
solvable in this case~\cite{DK98}.
Note that NP-hardness for~$s=2$ does not directly imply NP-hardness
for all values~$s\ge 2$ (for example, the case~$s=5$ is trivially
solvable on grid graphs since they have maximum degree four).
We observe that \probPthreePart is
typically not easier than \probStarPart for~$s\ge3$.  An exception to
this rule is the class of split graphs (which are chordal), where
\probPthreePart is \pt{} solvable but \probStarPart is
NP-hard for any constant value~$s\ge3$.  

\paragraph{Preliminaries.}
We assume basic familiarity with standard graph
classes~\cite{BLS99,Gol04}.  Definitions of the graph classes are
provided when first studied in this paper.  We call the
complete bipartite graph~$K_{1,s}$ an \emph{\sstar}.  For a graph $G=(V,E)$, an
\emph{\starpartition{s}} is a set of~$k\coloneqq|V|/(s+1)$ pairwise disjoint
vertex subsets~$V_1, V_2, \ldots,\allowbreak V_k \subseteq V$ with
$\bigcup_{1 \leq i\leq k}V_i=V$ such that each induced subgraph~$G[V_i]$
contains an \sstar as a (not necessarily induced) subgraph.  We refer
to the vertex sets~$V_i$ as \emph{stars}, even though the correct
description of a star would be an \emph{arbitrary~$K_{1,s}$-subgraph
  of~$G[V_i]$}.  \probPthreePart is the special case of \probStarPart
with $s=2$.  Without loss of generality, we assume throughout the
paper that the input graph~$G$ is connected (otherwise, we can solve
the partition problem separately for each connected component of~$G$).
We denote by $n\coloneqq|V|$ the number of vertices and
by $m\coloneqq|E|$ the number of edges in a graph~$G=(V,E)$.
For a vertex $v\in V$, we denote by~$N[v]\coloneqq\{u\in V\mid \{u,v\}\in
E\}\cup\{v\}$ the \emph{closed neighborhood} of~$v$.

\paragraph{Article outline.}
The article is structured into one section per graph class.  Herein,
we first present the results on graph classes with \pt{}
algorithms and then head over to the graph classes with NP-hardness
results.  Each section gives a formal definition of the graph class it
considers.
\cref{sec:intervals} considers interval graphs and their subclasses
unit interval graphs and trivially perfect graphs.
\cref{sec:cographs} provides a \pt{} algorithm for cographs,
\cref{sec:biperm} for bipartite permutation graphs.
\cref{sec:splitgraphs} marks the boundary between tractability and
NP-hardness: it shows that \probPthreePart{} is \pt{}
solvable on split graphs, while \probStarPart{} is NP-hard.
\cref{sec:gridgraphs} shows that \probPthreePart{} is NP-hard on grid
graphs and, finally, \cref{sec:chordgraphs} shows it for chordal
graphs.

\section{Interval graphs}\label{sec:intervals}

\looseness=-1 In this section, we present algorithms that solve \probStarPart{} on
unit interval graphs and on trivially perfect graphs in linear time, 
and a simple greedy algorithm that solves \probPthreePart{} on interval
graphs in quasilinear time.

An \emph{interval graph} is a graph whose vertices one-to-one
correspond to intervals on the real line such that there is an edge
between two vertices if and only if their representing intervals
intersect.  Interval graphs naturally occur in many scheduling
applications~\cite{KLPS07,BMNW14}.  In a \emph{unit interval graph},
all representing intervals are open and have the same length, while in
a \emph{trivially perfect graph}, any two representing intervals are
either disjoint or one is properly contained in the other.

\subsection{Star Partition on unit interval graphs}

The restricted structure of unit interval graphs allows us to solve
\probStarPart using a simple greedy approach, which yields the
following result.

\begin{theorem}
  \label{thm:unitlinear}
  \probStarPart is solvable in $O(n+m)$~time on unit interval
  graphs.
\end{theorem}

\noindent The general idea behind the algorithm for \cref{thm:unitlinear} is to
order the vertices in such a way that we can repeatedly select the
$s+1$~leftmost vertices to form an $s$-star and then delete them. If,
at some point, the $s+1$~leftmost vertices do not contain an $s$-star,
then it can be shown that the graph cannot be partitioned
into~$s$-stars.  We order the vertices according to a so-called
bicompatible elimination order:

\newcommand{\rightmost}{\ensuremath{\mathrm{r}}}

\begin{definition}[\cite{PD03}]
  For a graph~$G=(V,E)$, a \emph{bicompatible elimination order} is an
  ordering~$\sigma\colon V \rightarrow \{1,\dotsc,n\}$
  such that, for each vertex~$v \in V$,
  \begin{align*}
    \text{the set }N_l[v] &\coloneqq \{u \in N[v] \mid \sigma(u) \le
    \sigma(v)\}\text{ of its left neighbors and}\\
    \text{the set }N_r[v] &\coloneqq \{u \in N[v] \mid \sigma(u) \ge
    \sigma(v)\}\text{ of its right neighbors}
  \end{align*}
  each form a clique in~$G$.
\end{definition}

\noindent A graph is a unit interval graph if and only if it allows for a
bicompatible elimination order~\cite{PD03}.  Our algorithm will
exploit the following property of bicompatible elimination orders:

\begin{lemma}[\cite{BKMN10}]\label[lemma]{lem:range of an edge induces a clique}
  Let $G=(V,E)$ be a connected unit interval graph and~$\sigma$ be a bicompatible elimination order for~$G$.  Then, for all~$\{u,v\}
  \in E$ with~$\sigma(u) < \sigma(v)$, the set~$\{w \in V \mid
  \sigma(u) \le \sigma(w) \le \sigma(v)\}$ induces a clique in~$G$.
\end{lemma} 

\noindent We are now ready to prove \cref{thm:unitlinear}.

\begin{proof}[Proof of \cref{thm:unitlinear}]
  Given a unit interval graph~$G=(V,E)$ with $n\coloneqq|V|$ and $m\coloneqq|E|$, we
  can compute in linear time a bicompatible elimination
  order~$\sigma$~\cite{PD03}.  Moreover, we can assume $G$~to be
  connected, thus making \cref{lem:range of an edge induces a clique}
  applicable.  For a subset~$V' \subseteq V$ let~$\rightmost(V') \coloneqq 
  \argmax_{v\in V'}\sigma(v)$ denote the rightmost vertex in~$V'$ with
  respect to~$\sigma$.

  Now, we greedily partition~$G$ into $s$-stars starting with the
  first (with respect to~$\sigma$) $s+1$ vertices~$v_1,\dotsc,v_{s+1}$ with $\sigma(v_1) < \ldots < \sigma (v_{s+1})$.
  If~$G[\{v_1,\dotsc,v_{s+1}\}]$ does not contain an~$s$-star, then we
  answer ``no''.
  Otherwise, we delete~$v_1,\dotsc,v_{s+1}$ from~$G$ and
  continue on the remaining graph.  
  If we end up with the empty graph, then
  we have found a partition of~$G$ into~$s$-stars and answer~``yes''.

  Obviously, the algorithm requires $O(n+m)$ time since checking whether an induced subgraph~$G[V']$ with $|V'|=s$ vertices contains an $s$-star runs in $O(|G[V']|)$~time and after the check we delete the set~$V'$ from the graph. 
  In this way, we touch each vertex at most once and each edge at most twice.

  \begin{figure}[t]
  \centering
  \begin{tikzpicture}
    \tikzstyle{neuron}=[circle, draw=black, inner sep=1.5pt]

    \foreach \x in {1,2,3,4} 
    \node[single, label=below:{$v_\x$},
    fill=gray] (v\x) at (\x,2) {}; 

    \foreach \x in {5,6,...,9}
    \node[single] (v\x) at (\x,2) {}; 
    
    \node[single, label=below:{$x$}] (v10) at (10,2) {}; 
    \node[single, label=below:{$u$}, fill=black] (u) at (11,2) {}; 
    \node[single] (v12) at (12,2) {}; 
    \node[single] (v13) at (13,2) {};

    \draw (v1) -- (v2); \draw (v2) to[out=45, in=135] (v5); \draw (v2)
    to[out=45, in=135] (v6);

    \draw (v3) to[out=315,in=225] (v7); \draw (v7) -- (v8); \draw (v7)
    to[out=315, in=225] (v10);

    \draw (v4) to[out=45, in=135] (v9); \draw (v9) to[out=45, in=135]
    (v12); \draw (v9) to[out=45, in=135] (v13);

    \foreach \x in {1,2,3,4} 
    \node[single, label=below:{$v_\x$},
    fill=gray] (v\x) at (\x,0) {}; 

    \foreach \x in {5,6,...,9}
    \node[single] (v\x) at (\x,0) {}; 

    \node[single, label=below:{$x$}] (v10) at (10,0) {}; 

    \node[single, label=below:{$u$}, fill=black] (u) at (11,0) {}; 

    \node[single] (v12) at (12,0) {}; 

    \node[single] (v13) at (13,0) {};

    \draw (v1) -- (v2); 
    \draw (v2) -- (v3); 
    \draw (v2) to[out=45, in=135] (v4);

    \draw (v5) -- (v6); 
    \draw (v5) to[out=45, in=135] (v7); 
    \draw (v5) to[out=45, in=135] (v8);

    \draw (v10) -- (v9); 
    \draw (v10) to[out=45, in=135] (v12); 
    \draw (v10) to[out=45, in=135] (v13);
  \end{tikzpicture}
  \caption{Example of a $3$-star partition of a unit interval graph with
    vertices ordered according to a bicompatible elimination order from left
    to right. Only the edges and vertices of the first three stars as well as
    the rightmost neighbor~$u\coloneqq \rightmost(N(v_4))$ of $v_4$ (black) are shown.
    Top: $v_1,\dotsc,v_4$ are not grouped together into a star in the
    partition.  Bottom: A possible rearrangement of the~$3$-stars as described
    in the proof of \cref{thm:unitlinear}. It is always possible to
    group~$v_1,\dotsc,v_4$ into a $3$-star.}
  \label{fig:unit}
\end{figure}
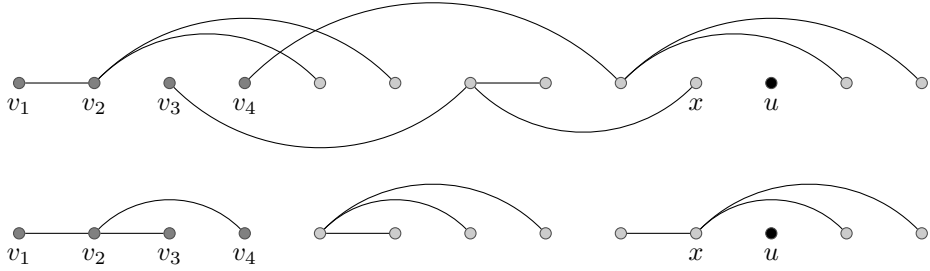
  
  It remains to show
  that this procedure is correct. To this end, we show that 
  if $G$ admits an $s$-star partition, 
  then $G$ also admits an $s$-star partition~$P'$ with $S \coloneqq \{v_1,\dots,v_{s+1}\}\in P'$
  (note that $v_1,\ldots, v_{s+1}$ are the first $s+1$ vertices).
  Let~$P$~be a partition of~$G$ into~$s$-stars such that
  $\{v_1,\dots,v_{s+1}\}\notin P$, that is, the first $s+1$ vertices
  are not grouped into one star but distributed among several
  stars. Then, let~$S_1,\dotsc,S_\ell\in P$, $2\le \ell \le s+1$,
  be the stars that contain at least one
  vertex from $S$, that is,  
  $S\subseteq \bigcup_{i=1}^{\ell}S_i$ and
  $S_i \cap \{v_1,\dotsc,v_{s+1}\} \neq \emptyset$ for~$1 \le i \le \ell$,
  and assume that~$v_1 \in S_1$.  Further,
  let~$c_i$ denote the center vertex of~$S_i$ for $1\le i \le \ell$. 
  Note that~$\sigma(\rightmost(S_1))
  > s+1$, which implies~$S \subseteq N[c_1]$.
  Since~$N_l[c_1]$ and~$N_r[c_1]$ are cliques, it follows
  that~$G[S]$ contains an~$s$-star that could
  participate in an $s$-star partition if the remaining vertices
  in~$S' \coloneqq  \bigcup_{i=1}^\ell S_i \setminus S$
  can also be partitioned into $s$-stars. To verify that this is possible,
  observe first that the number~$|S'| = (\ell-1)(s+1)$ of the remaining
  vertices is again divisible by~$s+1$.

  We now show that we can greedily partition $S'$ into stars, because $S'$
  consists of two cliques such that there is a vertex of the first
  clique that is adjacent to all vertices of the second clique. 
  To show this, we utilize the following claim, which describes the relative
  position of the center~$c_i$ of star~$S_i$ and the rightmost neighbor of $v_{s+1}$:  
  \begin{claim}\label[claim]{claim:star center and the rightmost neighbor}
    For all $1\le i \le \ell$, the center~$c_i$ of 
    star~$S_i$ satisfies that $\sigma(c_i)\le \sigma(\rightmost(N[v_{s+1}]))$.
  \end{claim}
  
  \begin{proof}[Proof of \cref{claim:star center and the rightmost neighbor}]
  \renewcommand{\qedsymbol}{(of
      \cref{claim:star center and the rightmost neighbor})~$\diamond$}
    Suppose towards a contradiction that $\sigma(c_i) > \sigma(\rightmost(N[v_{s+1}]))$.
    Then $c_i \neq \rightmost(N[v_{s+1}])$ and thus, $\{c_i,
    v_{s+1}\}\in E$ since $c_i$ is adjacent to at least one vertex
    from~$v_1,\allowbreak\dotsc,\allowbreak v_{s+1}$ and \cref{lem:range of an edge induces a
      clique} holds.  Hence, $c_i\in N_r[v_{s+1}]$, which
    contradicts~$\sigma(c_i) > \sigma(\rightmost(N[v_{s+1}]))$.
  \end{proof}
    
  Now, let $u \coloneqq  \rightmost(N[v_{s+1}])$ denote the rightmost neighbor
  of~$v_{s+1}$.
  It holds that~$S' \subseteq N[u]$. This can be seen as follows:
  For a vertex~$v' \in S'$, either $s+1 < \sigma(v') \le \sigma(u)$ or
  $\sigma(u) < \sigma(v')$ holds.
  In the first case, \cref{lem:range of an edge induces a clique}
  implies that~$\{v',u\}\in E$ since $\{v_{s+1},u\}\in E$.
  For the second case, let~$S_i$, $1\le i\le \ell$ be the star containing $v'$.
  Then, by \cref{claim:star center and the rightmost neighbor}, 
  it follows that $S_i$'s center~$c_i$ satisfies $\sigma(c_i)\le \sigma(u)$.
  If~$\sigma(u) < \sigma(v')$, then \cref{lem:range of an edge induces a clique}
  implies $\{u,v'\}\in E$ since~$\{c_i,v'\}\in E$.

  Now, consider the vertex~$x\coloneqq \rightmost(S' \cap N_l[u])$, that is,
  the rightmost vertex in~$S'$ that is a left neighbor of~$u$.
  Clearly, from \cref{claim:star center and the rightmost neighbor}
  it follows that $\sigma(c_i)\le \sigma(x)$ holds for every star center~$c_i$,
  $1\le i \le \ell$, 
  since otherwise $c_i$ were to be ordered between $x$ and~$u$, and is hence, a left neighbor of $u$---a contradiction to $x$ being the rightmost left neighbor of $u$ in $S'$.
  Thus, $x$ is adjacent to all vertices in~$S' \cap N_r[u]$ due to
  \cref{lem:range of an edge induces a clique}.
  The vertices in~$S' \cap N_l[u]$ are also adjacent to~$x$ as they
  induce a clique which includes~$x$.  Moreover, $S'\cap N_r[u]$ also induces a
  clique.
  Therefore, we simply partition the vertices in $S'$ from right
  to left (with respect to~$\sigma$) into $s$-stars.
  This is always possible since~$x$ is connected to all vertices in
  both cliques $S'\cap N_r[u]$ and $S'\cap N_l[u]$.
  \cref{fig:unit} depicts an example of the rearranged partition. 
\end{proof}

\subsection{Star Partition on trivially perfect graphs}

Recall that an \emph{interval graph} is a graph whose vertices
correspond directly to intervals on the real line, and there is an
edge between two vertices if their intervals intersect. A
\emph{trivially perfect} (also known as \emph{quasi-threshold}) graph
is an interval graph representable such that any two intervals are
either disjoint or one is properly contained in the other. %

In order to solve \probStarPart{} in linear time on trivially perfect
graphs, we will make use of the linear-time computable \emph{(rooted)
  tree representation} of connected trivially perfect
graphs~\cite{YCC96}: %

\begin{definition}[Rooted tree representation]\label[definition]{def:rtrep}
  Let $G=(V,E)$~be a connected trivially perfect graph. Let $T(G)$ be
  the directed graph on the vertex set~$V$ that contains an
  arc~$(v,w)$ if and only if 
  a) the interval representing~$v$ contains the interval representing~$w$, and
  b) there is no other vertex~$u$ such that 
  its representing interval contains the interval representing~$w$ and 
  is contained in the interval representing~$v$. 

  By definition of trivially perfect graphs, $T(G)$ is a directed tree
  having a unique vertex, the \emph{root}, with in-degree zero. 
  We call $T(G)$ the \emph{rooted tree representation} of~$G$.

  If, in $T(G)$,
  a vertex~$u$ lies on the directed path from the root to a vertex~$v$, 
  or equivalently, 
  if there is a directed path from $u$ to $v$, 
  we call $u$ \emph{ancestor of $v$} and $v$ \emph{descendant of $u$}.
  The \emph{depth} of a vertex is the length of the path from the root to this vertex. 
\end{definition}

\noindent 
\cref{def:rtrep} is illustrated in \cref{fig:trivial_perfect_graphs}. 
It is crucial to observe the equivalence of the adjacency of two vertices and their ancestor-descendant relation:

\begin{observation}\label[observation]{obs:adjacency<=>ancestors or descendants}
  The graph~$G$ contains an edge~$\{p,q\}$ if and only if
  $p$~is either an ancestor or a descendant of~$q$ in~$T(G)$.
\end{observation}

\begin{proof}
  Since $G$ is a trivially perfect graph,
  $G$ contains an edge $\{p,q\}$ if and only if either $p'\subset q'$ or $q' \subset p'$ where $p'$ and $q'$ are the representing intervals of $p$ and $q$, respectively.
  If $p'\subset q'$, 
  then there is a directed path from $q$ to $p$ in~$T(G)$.
  Conversely, $q'\subset p'$ implies that there is a directed path from $p$ to $q$ in~$T(G)$.
  By the definition of ancestors and descendants,
  $p$ is either an ancestor or a descendant of $q$. 
\end{proof}

\noindent Also the following is easy to observe:

\begin{observation}\label[observation]{obs:ancestor and edge relation}
  If there are three vertices~$p,q,r$ such that 
  $G$ contains an edge~$\{q, r\}$ and $p$ is an ancestor of $q$ in $T(G)$,
  then $G$ also contains an edge~$\{p,r\}$. 
\end{observation}

\begin{proof}
  By \cref{obs:adjacency<=>ancestors or descendants}, $\{q,r\}$ being an edge in~$G$ implies that $q$ is either an ancestor or a descendant of $r$.
  If $q$ is an ancestor of $r$, 
  then $p$ is also an ancestor of $r$, 
  implying that $G$~contains the edge~$\{p,r\}$.
  Otherwise, the interval~$q'$ representing~$q$ is properly contained in the interval~$r'$ representing~$r$.
  Since $q'$~is also properly contained in the interval~$p'$ representing~$p$ ($p$ is an ancestor of $q$),
  we obtain that $p'$ and~$r'$ are not disjoint.
  By the definition of trivially perfect graphs, 
  $\{p,r\}$~is contained in~$G$.
\end{proof}

\noindent Before presenting our algorithm, 
we show that we may assume that star partitions of~$G$ have a very restricted structure with respect to~$T(G)$. 
First of all, we can assume that the center of a star is an ancestor of all its leaves:

\begin{figure}
    \centering
    \def\gap{.2}
    \def\dist{.2}
    \def\lone{.6}
    \def\ltwo{.9}
    \def\sep{0.6} 
    \begin{tikzpicture}[>=stealth',  edge from parent path={(\tikzparentnode) -- (\tikzchildnode)}]    
      \tikzset{level distance = 17pt, sibling distance=15pt}
      \begin{scope}[xshift=-3cm]
        \node[vertex] (S-a) at (0,0) {};

        \node[vertex] (S-b) at (0,-\sep) {};

        \node[vertex] (S-c) at (\lone+2*\dist,-\sep) {};

        \node[vertex] (S-d) at (0,-2*\sep) {};

        \node[vertex] (S-e) at (2*\lone+\dist,-2*\sep) {};

        \node[vertex] (S-f) at (-\ltwo-\dist,-3*\sep) {}
        node[right = 0pt of S-f] {$f$};

        \node[vertex] (S-g) at (0,-3*\sep) {}
        node[] at ($(S-g.north east) + (0.1, 0.1)$) {$g$};

        \node[vertex] (S-h) at (\ltwo+\dist,-3*\sep) {}
        node[right = 0pt of S-h] {$h$};

        \foreach \x in {a, b, c, d, e} {
          \node[] at ($(S-\x.north east) + (0.1, 0.1)$) {$\x$};
        }

        \draw[-] (S-d) -- (S-f);
        \draw[-] (S-d) -- (S-h);

        \draw[-, draw=white, %
        line width=1.5pt]
        (S-a) to[bend right] (S-d);
        \draw[-, draw=white, %
        line width=1.5pt]
        (S-a) to[bend right=40] (S-g);
        \draw[-, draw=white, %
        line width=1.5pt]
        (S-a) to[bend left=10] (S-h);
        \draw[-] (S-a) to[bend right] (S-d);
        \draw[-] (S-a) to[bend right=40] (S-g);
        \draw[-] (S-a) to[bend left=10] (S-h);

        \draw[] (S-c) -- (S-e);
        \draw[-] (S-a) -- (S-b);
        \draw[-] (S-d) -- (S-g);

        \draw[thickedge] (S-a) -- (S-c);
        \draw[thickedge] (S-a) to[] (S-e);
        \draw[thickedge] (S-a) to[bend right=10] (S-f);

        \draw[thickedge] (S-b) -- (S-d);
        \draw[thickedge] (S-b) -- (S-h);
        \draw[-] (S-b) -- (S-f);

        \draw[-, draw=white, %
        line width=1.5pt]
        (S-b) to[bend left=85] (S-g);
        \draw[thickedge] (S-b) to[bend left=95] (S-g);  

      \end{scope}

      \node[birth] (B-f) at (0,-3*\sep) {};
      \node[death] (D-f) at (\lone,-3*\sep) {};
      \node[birth] (B-g) at (\lone+\gap,-3*\sep) {};
      \node[death] (D-g) at (2*\lone+\gap,-3*\sep) {};
      \node[birth] (B-h) at (2*\lone+2*\gap,-3*\sep) {};
      \node[death] (D-h) at (3*\lone+2*\gap,-3*\sep) {};

      \node[birth] (B-d) at (0-\dist, -2*\sep) {};
      \node[death] (D-d) at (3*\lone+3*\gap, -2*\sep) {};
      \node[birth] (B-e) at (4*\lone+7*\gap, -2*\sep) {};
      \node[death] (D-e) at (4.5*\lone+7*\gap, -2*\sep) {};

      \node[birth] (B-b) at (-2*\dist, -\sep) {};
      \node[death] (D-b) at (4*\lone+4*\gap+\dist, -\sep) {};
      \node[birth] (B-c) at (4*\lone+5*\gap+\dist, -\sep) {};
      \node[death] (D-c) at (5*\lone+5*\gap+\dist, -\sep) {};

      \node[birth] (B-a) at (-3*\dist, 0) {};
      \node[death] (D-a) at (5*\lone+5*\gap+2*\dist, 0) {};

      \foreach \x in {a,b,c,d,e,f,g,h} {
        \draw[intervalline] (B-\x) -- (D-\x) {}
        node [pos=0.5, above=-0.1] {$\x$};
      }

      \begin{scope}[shift={(5.5,0)}]
      
      \node[vertex] (S-a) at (0,0) {};

        \node[vertex] (S-b) at (0,-\sep) {};

        \node[vertex] (S-c) at (\lone+2*\dist,-\sep) {};

        \node[vertex] (S-d) at (0,-2*\sep) {};

        \node[vertex] (S-e) at (2*\lone+\dist,-2*\sep) {};

        \node[vertex] (S-f) at (-\ltwo-\dist,-3*\sep) {}
        node[right = 0pt of S-f] {$f$};

        \node[vertex] (S-g) at (0,-3*\sep) {}
        node[] at ($(S-g.north east) + (0.1, 0.1)$) {$g$};

        \node[vertex] (S-h) at (\ltwo+\dist,-3*\sep) {}
        node[right = 0pt of S-h] {$h$};

        \foreach \x in {a, b, c, d, e} {
          \node[] at ($(S-\x.north east) + (0.1, 0.1)$) {$\x$};
        }

        \draw[->] (S-a) to (S-c);
        \draw[->] (S-c) to (S-e);
        \draw[->] (S-a) -- (S-b);
        \draw[->] (S-b) -- (S-d);
        \draw[->] (S-d) -- (S-g);
        \draw[->] (S-d) -- (S-f);
        \draw[->] (S-d) -- (S-h);

      \begin{pgfonlayer}{background}
        \def\lw{8pt}
        \def\sep{1pt}
        \draw [fill=gray!30, draw=gray!30, -, line cap=round, line width=\lw]  
        ($(S-b)+(0,\sep)$) -- ($(S-g)+(0,-\sep)$);
         \draw [fill=gray!30, draw=gray!30, -, line cap=round, line width=\lw]
        (S-d) -- ($(S-h.south east)+(-\sep,\sep)$);

        \draw [fill=gray!30, draw=gray!30, -, line cap=round, line width=\lw]
        ($(S-a)+(0,\sep)$) -- ($(S-c)+(\sep, -\sep)$);
        \draw [fill=gray!30, draw=gray!30, -, line cap=round, line width=\lw]
        ($(S-c)+(\sep, -\sep)$) -- ($(S-e.south east)+(-\sep,\sep)$);
        \draw [fill=gray!30, draw=gray!30, -, line cap=round, line width=\lw]
        ($(S-a)+(0,\sep)$) -- ($(S-f.south west)+(\sep,\sep)$);
      \end{pgfonlayer}
      \end{scope}

    \end{tikzpicture}
    \caption{An example of a trivially perfect graph and its partition into
      stars~$K_{1,3}$.
      Left: The trivially perfect graph with eight vertices partitioned into stars~$K_{1,3}$ (bold).
      Middle: The interval representation.
      Right: The rooted tree representation with the corresponding partition in shaded gray.}
    \label{fig:trivial_perfect_graphs}
\end{figure}
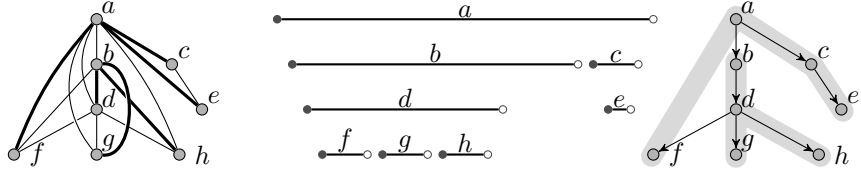

\begin{observation}\label[observation]{lem:triv_perfect:star-center-is-ancestor}
  Let $G$ be a trivially perfect graph with $n$~vertices. 
  If $G$~allows for an $s$-star
  partition~$\{V_1, V_2, \dots, V_{n/(s+1)}\}$, then each $G[V_i]$, $1\le i \le n/(s+1)$,
  contains an $s$-star whose center vertex~$c_i$ is an ancestor of all
  vertices $V_i\setminus \{c_i\}$ in~$T(G)$.
\end{observation}

\begin{proof}
  Let $c_i$ be the center of an $s$-star in $G[V_i]$.  By
  \cref{obs:adjacency<=>ancestors or descendants}, for each
  vertex~$u\in V_i\setminus \{c_i\}$, $c_i$ is either an ancestor or a
  descendant of $u$.  If $c_i$ is not an ancestor of all vertices
  in~$V_i \setminus \{c_i\}$, then let $a$~be an ancestor of~$c_i$
  in~$V_i$ with smallest depth.  Clearly, since $a$~is an ancestor of
  the center~$c_i$, $a$~is adjacent to all vertices of~$V_i$ in~$G$ by
  \cref{obs:ancestor and edge relation}.  It remains to show that
  $a$~is also an ancestor of all vertices in~$V_i\setminus \{a\}$.
  Suppose, towards a contradiction, that there is a vertex~$u \in
  V_i\setminus \{a\}$ such that $a$~is not an ancestor of~$u$.  By
  \cref{obs:adjacency<=>ancestors or descendants}, $u$~is an ancestor
  of~$a$ and is, hence, an ancestor of~$c_i$ with a smaller depth
  than~$a$---a contradiction.
\end{proof}

\noindent 
Our next observation is that we can assume that no star center is contained in a subtree~$T'$ of a rooted tree representation~$T(G)$ if $T'$~contains ``too few'' vertices.
Therefore, of special interest to use are subtrees of~$T(G)$ that can contain $s$-stars but of which no subtree can:

\begin{definition}[Center barrier]
  A subtree $X$ of a rooted tree representation~$T(G)$ is a \emph{center barrier for $s$-stars $K_{1,s}$} if $X$ has at least $s+1$~vertices and each proper subtree of~$X$ has at most $s$~vertices.
\end{definition}

\noindent 
The term ``center barrier'' is chosen since we can assume that no subtree of the center barrier contains an $s$-star center. 
Note that any (connected) rooted tree representation with at least $s+1$~vertices contains a center barrier.

\begin{observation}\label{center-barrier}
  Let $X$ be a center barrier for $s$-stars in a rooted tree representation~$T(G)$, 
  and $P$ be an $s$-star partition of~$G$. 
  Then, for any $V_i\in P$ that shares a vertex with~$X$, 
  the graph~$G[V_i]$ contains a star whose center is the root of~$X$ or an ancestor of that root.
\end{observation}

\begin{proof}
  By the definition of center barriers, 
  $X$~is a subtree of~$T(G)$.
  Let $x$~be its root.
  By \cref{lem:triv_perfect:star-center-is-ancestor}, 
  $G[V_i]$~contains a star whose center~$c$ is the ancestor of all vertices in~$V_i\setminus\{c\}$. 
  If $x\neq c$, then let $w \in V_i\cap X$, which exists by assumption.
  Observe that $x$~is an ancestor of~$w$. 
  Since $G$~contains the edge~$\{w,c\}$ ($c$~is a center for~$V_i \ni w$), 
  by \cref{obs:ancestor and edge relation} and $\{c, w\}$~being an edge in~$G$,
  $G$~also contains an edge~$\{x,c\}$.
  Then, by \cref{obs:adjacency<=>ancestors or descendants},
  $c$ is either an ancestor or a descendant of $x$. 
  If $c$~is a descendant of $x$,
  then the subtree of $T(G)$ rooted at $c$ (which contains $V_i$) is a proper subtree of~$X$.
  This is impossible since the proper subtrees of~$X$ have at most $s$~vertices.
  Hence, $x\neq c$ implies that $c$ is an ancestor of~$x$.
\end{proof}

\noindent 
Finally, we show that there exists a feasible star partition where each star
consists only of vertices from center barriers.

\begin{lemma}\label[lemma]{lem:triv_perfect:a-star-as-subtree}
  Let $G$ be a trivially perfect graph allowing for an $s$-star
  partition, let $X$ be a center barrier for $s$-stars in the rooted
  tree representation~$T(G)$, and let $x$ be the root of $X$. Then,
  $G$~admits a star partition~$\Psol$ with $S\cup \{x\}\in \Psol$,
  where $S$~consists of $s$~arbitrary vertices of~$X\setminus
  \{x\}$.
\end{lemma}

\begin{proof}
  Let $\Qsol$ be a star partition of $G$ and let $S$ consist of $s$~arbitrary vertices of~$X\setminus\{x\}$.
  By \cref{lem:triv_perfect:star-center-is-ancestor}, 
  we can assume the center~$c$ of a vertex subset~$V_i\in \Qsol$ 
  to be the one being the ancestor of all other vertices in $V_i \setminus \{c\}$.
  If $S\cup \{x\} \in \Qsol$, 
  then the partition we are searching for is $\Psol\coloneqq \Qsol$.
  Otherwise, 
  we show how to transform the partition~$\Qsol$ into a new partition~$\Psol$ containing~$S\cup\{x\}$.
  We repeatedly exchange the vertices of two vertex sets in~$\Qsol$ until 
  \begin{align}\label{P'contains an element containing all S}
    \text{ the modified partition~} \Qsol' \text{ contains a set}~V_w\text{ such that } S\subsetneq V_w. 
    \end{align}
  Finally, we set $\Psol\coloneqq  (\Qsol' \setminus \{V_w, V_u\}) \cup \{(V_u \setminus \{x\}) \cup \{w\}, (V_w\setminus \{w\}) \cup \{x\}\}$,
  where $w$~is the center of~$V_w$ and $V_u$~is the vertex set with center~$u$ in $\Qsol'$ such that $x\in V_u$.
  One can verify that both 
  $(V_w\setminus \{w\}) \cup \{x\}$ and $(V_u \setminus \{x\}) \cup \{w\}$~contain an $s$-star, 
  implying that $\Psol$~is indeed an $s$-star partition for $G$:
  on the one hand, 
  it is easy to see that $G[(V_w\setminus \{w\}) \cup \{x\}]$ contains an $s$-star with center~$x$ since $x$ is the ancestor of all vertices in $V_w \setminus \{w\}=S$.
  On the other hand, 
  the fact that $w$~and~$x$ are both ancestors of all vertices in~$S$ (for~$w$, this holds since $S\subsetneq V_w$ and by \cref{lem:triv_perfect:star-center-is-ancestor})
  implies that $w$~and~$x$ are adjacent in~$G$ (\cref{obs:adjacency<=>ancestors or descendants,obs:ancestor and edge relation}).
  Since $u$~is an ancestor of~$x$, 
  by \cref{obs:ancestor and edge relation}, 
  we have that $u$~and~$w$ are adjacent in~$G$.
  This implies that $u$~is either an ancestor or a descendant of~$w$ (\cref{obs:adjacency<=>ancestors or descendants,obs:ancestor and edge relation}).
  In any case, $G[(V_u \setminus \{x\}) \cup \{w\}]$ contains an $s$-star with star center either $u$~or~$w$. 

  Now, in the remainder of the proof, 
  we aim at transforming the partition~$\Qsol$ into a new partition~$\Qsol'$ fulfilling Property (\ref{P'contains an element containing all S}).
  To this end, 
  among all vertex subsets $V_i\in \Qsol$ with~$V_i\cap S\neq \emptyset$,
  we let $V_y$ be the one with the center~$y$ closest to~$x$ with respect to~$T(G)$ (possibly, $y=x$).
  By assumption, $|S \setminus V_y| \ge 1$.
  Thus, let $V_z\in \Qsol$ be another subset with center~$z$ that contains at least one vertex from~$S$.
  By \cref{center-barrier},
  $z$ is an ancestor of $x$.
  By the selection of $y$, $z$ is also an ancestor of $y$.
  Thus, in the graph~$G$, $z$~is adjacent to every vertex in~$V_y$ 
  and $y$~is adjacent to all vertices in~$V_z\cap S$ since $y$~is either~$x$ or an ancestor of~$x$ (\cref{center-barrier,obs:adjacency<=>ancestors or descendants}).
  Thus, by setting $V'_y \coloneqq  (V_y \setminus Y) \cup (V_z \cap S)$ and
  $V'_z \coloneqq  (V_z \setminus S) \cup Y$, where $Y\subsetneq V_y\setminus S$
  is an arbitrary size-$(|V_z\cap S|)$ subset,
  we obtain a new valid partition~$(\Qsol\setminus \{V_y, V_z\}) \cup \{V'_y, V'_z\}$ such that $V'_y$ shares more vertices with $S$ than $V_y$ does.
  Note that $Y$ exists since $|S|+1=|V_y|=|V_z|$.
  Repeating the above procedure at most $s-1$ times results in a partition satisfying~(\ref{P'contains an element containing all S}).
\end{proof}

\noindent  Based on \cref{lem:triv_perfect:a-star-as-subtree}, we now give a
 linear-time algorithm computing an $s$-star partition (if existent) of a given
 trivially perfect graph.
 
\begin{theorem}\label{thm:triv_perfect-poly}
  \probStarPart can be solved in $O(n+m)$ time on trivially perfect graphs.
\end{theorem}

\begin{proof}
\looseness=-1  Let $G$ be a connected trivially perfect graph.
  Construct a tree representation~$T(G)$ of $G$ in linear time~\cite{YCC96}.
  Furthermore, construct in linear time a directed acyclic graph~$D(G)$ from $G$ which has the same vertex set as $G$,
  and for each edge~$\{u,v\}\in E(G)$, there is an arc~$(u,v)$ in~$D(G)$ if and only if the degree of $u$ is larger than the degree of $v$ in $G$.  

  Due to \cref{lem:triv_perfect:a-star-as-subtree}, if $G$ admits an
  $s$-star partition, then $G$ also admits an $s$-star
  partition~$\Psol\coloneqq \{V_1, V_2, \ldots, V_{|V|/(s+1)}\}$ such that for
  each $i\in \{1,2,\ldots, |V|/(s+1)\}$, $V_i$ is contained in a
  center barrier of the rooted tree representation for the
  graph~$G[V\setminus (\bigcup_{j=1}^{i-1}V_j)]$ resulting from~$G$ by
  deleting the vertices in~$\bigcup_{j=1}^{i-1}V_j$.  Hence, it is
  sufficient to recursively search for a center barrier~$X$ for
  $s$-stars, and delete the root of $X$ and $s$~arbitrary remaining
  vertices from~$X$ (these deleted vertices form a subset in the
  $s$-star partition).  If, at some point, there is no center barrier,
  then there are less than~$s+1$ remaining vertices; hence, $G$ cannot
  allow for an $s$-star partition.

  To realize the above algorithm in linear time,  
  we traverse $T(G)$ in a depth-first post-order way.
  If, in $D(G)$,
  the current vertex~$u$ has at least~$s$ non-marked (out-going) neighbors, 
  then we \emph{mark} this vertex~$u$ and $s$ (arbitrary) non-marked (out-going) neighbors of $u$.
  Otherwise, we do nothing.
  We answer yes if all vertices in~$D(G)$ are marked after traversing the whole tree,
  and no otherwise.
  Since we mark, in $D(G)$, each vertex and each of its out-going neighbor at most once, and since we traverse each vertex in $T(G)$ at most once,
  by the construction of $D(G)$ and $T(G)$, 
  the total running time is $O(|V(G)|+|E(G)|)$.
\end{proof}

\subsection{\boldmath$P_3$-Partition on interval graphs}
While it might not come as a surprise that \probStarPart can be solved
efficiently on unit interval graphs using a greedy strategy, this is
far from obvious for general interval graphs even when $s=2$. 
The obstacle here is that two intervals arbitrarily far apart from each other may eventually be required to form a~$P_3$ in the solution.  Indeed, the
greedy strategy we propose to overcome this obstacle is naive in the
sense of allowing wrong choices that can be corrected later.  Note
that, while we can solve the more general \probStarPart{} in
polynomial time on subclasses of interval graphs like unit interval
graphs and trivially perfect graphs (see previous subsections), we
are not aware of a \pt{} algorithm for \probStarPart{} with
$s\geq3$ on interval graphs.

\paragraph{Overview of the algorithm.}
\looseness=-1
The algorithm is based on the following analysis of a $P_3$-partition
of an interval graph.  Each $P_3$ contains a \emph{center} and two
\emph{leaves} connected to the center via their incident edges called \emph{links}.
We associate with each interval two so-called \emph{tokens}. We require
that the link between a leaf and a center \emph{consumes} both of the leaf's
tokens (such that a leaf can be associated to only one link) and one token of
the center (which can thus be linked to two leaves).

The algorithm examines the \emph{event points} (start and end points
of intervals) of an interval representation in increasing order.  We
consider that a link~$\{x,y\}$ consumes three tokens of~$x$ and~$y$\ as soon
as one of the two intervals ends.
Intuitively, a graph is a no-instance if, at some point, an interval
with one or two remaining tokens ends, but there are not enough 
tokens of other adjacent intervals to \emph{create} a link.
Note that a link consumes three tokens. 
A graph is a yes-instance if
the number of tokens is always sufficient.

The algorithm works according to the following two rules: when an
interval starts, its two tokens are added to a list;
when an interval with remaining tokens ends, then three tokens are
deleted from this list. 
Only tokens of the earliest-ending intervals will be deleted
(this choice may not directly translate into a ``sane'' solution, 
with each link consuming tokens from only two intervals, but it turns out not 
to be a problem). The algorithm is sketched in \cref{algo:P3-interval}.
\cref{fig:intervals} shows an example instance and the list of
tokens maintained by the algorithm.  Note that a token of an interval
$x$ is simply represented by a copy of interval $x$ itself.  We now
introduce the necessary formal definitions.

\SetAlgoSkip{skipamount}
\begin{algorithm}[t]
  \caption{$P_3$-partition of an interval graph}
  \KwIn{An interval representation of an interval graph with pairwise
    distinct event points in~$\{1,\dots,2n\}$.}  \KwOut{\texttt{true}
    if the graph allows for a $P_3$-partition, otherwise
    \texttt{false}.}
  \label{algo:P3-interval}
  $\A 0\gets$ empty token list $\emptyset$\;
  \For(\label{algoLoop}){$t\gets 1$ to $2n$}{
    \lIf{$t=\birth(x)$}
    {$\A{t}\gets\A {t-1} \qAdd (x,x)$ \label{alg:insert step}}
    \If{$t=\death(x)$}
    {\lIf{$x\notin \A{t-1}$}{$\A t\gets \A {t-1}$}
      \lElseIf{$\qSize{\A {t-1}}<3$}
      {\Return{\texttt{false}}} 
      \Else{$(x,y,z)\gets$ top three elements of $\A{t-1}$ (intervals
        ending first)\; $\A{t}\gets\A {t-1} \qDel
        (x,y,z)$\label{delAtline}\; } } }
    
  \Return{\texttt{true}}\;
\end{algorithm}

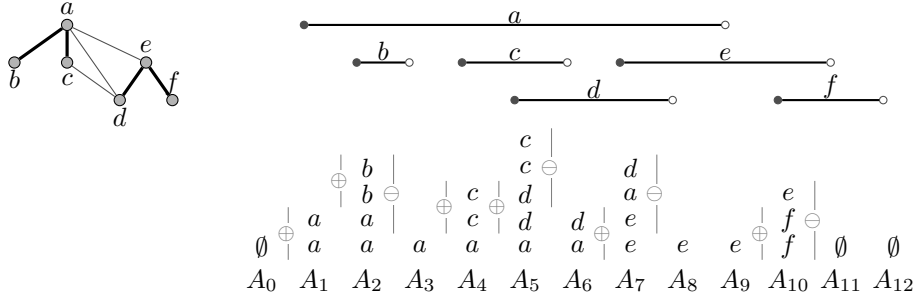
\begin{figure}[t]
  \centering \def\sep{0.5}
  \begin{tikzpicture}[draw=black!70, xscale=0.7]

    \begin{scope}[xshift=-6cm]
      \node[vertex] (S-a) at (2.5,2*\sep) {}; \node[vertex] (S-b) at
      (1.5, \sep) {}; \node[vertex] (S-c) at (2.5, \sep) {};
      \node[vertex] (S-d) at (3.5, 0) {}; \node[vertex] (S-e) at (4,
      \sep) {}; \node[vertex] (S-f) at (4.5, 0) {};

      \foreach \x in {a,e,f} \node at ([shift={(90:0.22)}]S-\x) {$\x$}
      ;

      \foreach \x in {b,c,d} \node at ([shift={(-90:0.22)}]S-\x)
      {$\x$} ;

      \foreach \x/\y in {a/d,a/e,d/c} \draw (S-\x) -- (S-\y);

      \foreach \x/\y in {a/b,a/c,d/e,e/f} \draw[thickedge] (S-\x) --
      (S-\y);
    
    \end{scope}

    \node[birth] (B-a) at (1,2*\sep) {}; \node[death] (D-a) at
    (9,2*\sep) {}; \node[birth] (B-b) at (2,\sep) {}; \node[death]
    (D-b) at (3,\sep) {}; \node[birth] (B-c) at (4,\sep) {};
    \node[death] (D-c) at (6,\sep) {}; \node[birth] (B-d) at (5,0) {};
    \node[death] (D-d) at (8,0) {}; \node[birth] (B-e) at (7,\sep) {};
    \node[death] (D-e) at (11,\sep) {}; \node[birth] (B-f) at (10,0)
    {}; \node[death] (D-f) at (12,0) {};

    \foreach \x in {a,b,c,d,e,f} \draw[intervalline] (B-\x) -- (D-\x) {} node
    [pos=0.5, above=-0.1] {$\x$} ;

    \begin{scope}[yshift=-0.6cm]

      \newcommand{\listdown}[2] { \node[table] at (#1+0.2,0.1)
        {$\A{#1}$}; \foreach \y/\l in {#2} \node[table] at
        (#1+0.2,-0.7*\y*\sep) {$\l$}; } 
      \newcommand{\add}[2] {
        \draw[draw=black!50] (#1-0.3, - 0.7*#2*\sep+0.35*\sep) --
        (#1-0.3, -0.7*#2*\sep-1.05*\sep) ; \node[fill=white,inner
        sep=1pt] at (#1-0.34, -0.7*#2*\sep-0.35*\sep)
        {\color{black!50}\footnotesize $\qAdd$}; }
      \newcommand{\del}[2] { \draw[draw=black!50] (#1-0.3,
        -0.7*#2*\sep+0.35*\sep) -- (#1-0.3, -0.7*#2*\sep-1.75*\sep) ;
        \node[fill=white,inner sep=1pt] at (#1-0.34,
        -0.7*#2*\sep-0.7*\sep) {{\color{black!50}\footnotesize
            $\qDel$}}; }
    
	\begin{scope}[yscale=-1, yshift=1.7cm]

      \listdown0{1/\emptyset} \listdown1{1/a,2/a} \add 11
      \listdown2{1/a,2/a,3/b,4/b} \add 23 \listdown3{1/a} \del 32
      \listdown4{1/a,2/c,3/c} \add 42 \listdown5{1/a,2/d,3/d,4/c,5/c}
      \add 52 \listdown6{1/a,2/d} \del 63 \listdown7{1/e,2/e,3/a,4/d}
      \add 71 \listdown8{1/e} \del 82 \listdown9{1/e}
      \listdown{10}{1/f,2/f,3/e} \add {10}1 \listdown{11}{1/\emptyset}
      \del {11}1 \listdown{12}{1/\emptyset}
      \end{scope}
    \end{scope}
  
   \end{tikzpicture}
  \caption{Left: An interval graph with six vertices and a
    $P_3$-partition $\Psol$ (bold).  Right: Interval representation of
    this graph and successive token lists $\A0,\ldots,\A{12}$
    computed by \cref{algo:P3-interval} (additions and deletions
    are marked with $\qAdd$ and $\qDel$). }
  \label{fig:intervals}
\end{figure}

\begin{definition}
Let $G=(V,E)$ be a fixed interval graph. 
We assume that any vertex~$u\in V$ represents a right-open interval $u=[\birth (u),\death(u)[$ with integer end points $\birth(u) < \death(u)$. 
Moreover, without loss of generality, each position in~$\{1,\ldots,2n\}$
corresponds to exactly one event.

Let $\Psol$ be a $P_3$-partition and $P=\{x,y,z\}\in \Psol$ with 
$\death(x)<\death(y) < \death(z)$, 
we write $\rank_\Psol(x)=1$, $\rank_\Psol(y)=2$, and $\rank_\Psol(z)=3$
(we omit the subscript when there is no ambiguity). 
Moreover, we call the element among $\{y,z\}$
having the earliest start point the \emph{center} of~$P$.
The other two elements of~$P$ are called
\emph{leaves}.  Note that the center of~$P$ intersects both leaves.

A \emph{token list}~$Q$ is a list of intervals $(q_1,\ldots,q_k)$
sorted in decreasing order of their end points ($\death(q_i)\geq\death(q_{j})$ for
$1\leq i\leq j\leq k$). 
To avoid confusion with the left-to-right sequence of event points, 
we consider the list to be written vertically, with the earliest-ending interval on top. %
We write
$\qSize{Q}$ for the length of~$Q$, $\emptyset$ for the empty token
list, and $x\in Q$ if interval $x$ appears in~$Q$.  We now define
\emph{insertion}~$\qAdd$, \emph{deletion}~$\qDel$,
and \emph{comparison}~$\lessgood$ of token lists: $Q\qAdd
(x_1,\ldots, x_l)$ is the token list obtained from $Q$ by inserting
intervals $x_1\dots,x_l$ so that the list remains sorted.  For $x\in
Q$, the list~$Q\qDel x$ is obtained by deleting one copy of $x$ from
$Q$ (otherwise, $Q\qDel x=Q$); and $Q\qDel (x_1,\ldots, x_l)=Q\qDel
x_1\qDel\ldots\qDel x_l$.
We write $(q_1,\ldots,q_k)\lessgood (q'_1,\ldots,q_{k'}')$ if $k\leq
k'$ and $\forall i\in \{1,\ldots k\}: \death(q_i) \le \death(q_i').$

Let $\Psol$ be a $P_3$-partition. We define $\tokens(\Psol)$ as a
tuple $(\h0,\h1,\ldots,\h{2n})$ of $2n+1$ token lists such that
$\h0\coloneqq \emptyset$ and for $t\in \{1,\ldots, 2n+1\}$,
\begin{itemize}
\item if $t=\birth(x)$, then $\h{t}\coloneqq \h{t-1} \qAdd (x,x)$,
\item if $t=\death(x)$, then let $P\coloneqq \{x,y,z\}$ be the $P_3$ in
  $\Psol{}$ containing $x$ and
  \begin{itemize}
  \item if $\rank(x)=1$, then $\h{t}\coloneqq \h{t-1}\qDel (x,x,c)$ where $c$
    is the center of $P$,
  \item if $\rank(x)=2$, then $\h{t}\coloneqq \h{t-1}\qDel (x,x,y,y,z,z)$,
  \item if $\rank(x)=3$, then $\h{t}\coloneqq \h{t-1}$.
  \end{itemize}
\end{itemize}
Note that in \cref{fig:intervals}, each token list
$\h{t}$ for $\Psol$ is equal to the respective $\A{t}$, except for $\h
6=(d,d)$ and $\h 7=(e,e,d,d)$.
\end{definition}

\noindent To %
compare the token lists generated by
\cref{algo:P3-interval} to those induced by a $P_3$-partition, we
show a few properties for both types of lists.
 
\begin{property}\label{property:token:delete_size=3}
  Let $\Psol$ be a $P_3$-partition with
  $\tokens(\Psol)=(\h0,\h1,\ldots,\h{2n})$ and let~$x$ be an interval
  with $t\coloneqq \death(x)$. Then, one of the following is true:
  \begin{enumerate}[i)]
  \item $x\in\h {t-1}$, $\qSize{\h{t-1}}\ge 3$, and
    $\qSize{\h{t}}=\qSize{\h{t-1}}-3$ or
  \item $x\notin\h {t-1}$ and $\h{t}=\h {t-1}$.
  \end{enumerate}
  Moreover, in both cases, $x\notin \h {t}$.
\end{property}

\begin{proof}
  Let $P\in\Psol$ be the $P_3$ containing $x$.  Depending on the rank
  of~$x$, we prove that either case~(i) or~(ii) applies.

  If $\rank(x)=1$, then $x$ is not the center of $P$. Let $c$ be the
  center of $P$. Since $c$ is adjacent to $x$, it follows that
  $\birth(c) < t$. Since $x$ ranks first, $\h{t-1}$ contains twice
  both elements~$x$ and~$c$.  Hence, $\qSize{\h{t-1}}\ge 4$, and from
  the definition of $\h{t}=\h{t-1}\qDel (x,x,c)$ it follows that
  $\qSize{\h{t}}=\qSize{\h{t-1}}-3$, we are thus in case~(i).
  Moreover, only one copy of~$c$ remains in~$\h {t}$.

  If $\rank(x)=2$ and $x$ is not the center, then let~$c$ be the center of
  $P$ and $y$ be the interval of the first rank in $P$. 
  From the reasoning above, 
  it follows that $\h{t-1}$ contains once $c$ and twice $x$ but
  no~$y$, implying that $\qSize{\h{t-1}}\ge 3$.
  As~$\h{t}=\h{t-1}\qDel (x,x,c,c,y,y)$, this implies that
  $\qSize{\h{t}}=\qSize{\h{t-1}}-3$: we are in case~(i).

  If $\rank(x)=2$ and $x$ is the center, then let $P=\{x, y, z\}$ such
  that $y$ ranks first and $z$ ranks third. Using the same reasoning
  as before it follows that $\h{t-1}$ contains~$x$ once and $z$ twice,
  but not $y$, implying that $\qSize{\h{t-1}}\ge 3$.
  As~$\h{t}=\h{t-1}\qDel (x,x,y,y,z,z)$, this implies that
  $\qSize{\h{t}}=\qSize{\h{t-1}}-3$: we are in case~(i).
  
  Finally, if $\rank(x)=3$, then $\h{t-1}$ does not contain $x$ (the
  last copies have been removed when the rank-2-interval ended), and
  $\h{t}=\h {t-1}$: we are in case~(ii).
  
  The fact that $x\notin \h{t}$ is clear in each case (all copies are
  removed when $x\in \h {t-1}$, none is added).
\end{proof}

\begin{property}\label{property:in_token_is_alive}
  For any $\A{t}$ defined by \cref{algo:P3-interval} and $x\in
  \A{t}$, it holds that $\birth(x)\leq t<\death (x)$.
  For any $P_3$-partition $\Psol$ with $\tokens(\Psol) =
  (\h0,\h1,\dots,\h{2n})$ and $x\in \h{t}$, it holds that
  $\birth(x)\leq t<\death (x)$.
\end{property}
\begin{proof}
  An element $x$ is only added to a token list $\A t$ or $\h t$ when
  $t=\birth(x)$, so the inequality $\birth(x)\leq t$ is trivial in
  both cases.  Consider now an interval~$x$ and~$t\coloneqq \death(x)$. We
  show that neither $\A t$ nor $\h t$ contain~$x$, which suffices to
  complete the proof.

  The fact that $x\notin \h t$ is already proven in
  \cref{property:token:delete_size=3}. Moreover, if $x\notin
  \A{t-1}$, then $x\notin \A t$ follows obviously.

  Now, assume that $x\in \A{t-1}$. We inductively apply
  \cref{property:in_token_is_alive} to obtain that for any
  $y\in\A{t-1}$, we have $t-1<\death(y)$ (note that the property is
  trivial for~$A_0$). Hence, $x$ is the interval with the earliest end
  point in $\A{t-1}$ (i.\,e., the interval on top) and all of its copies%
  (at most two) are removed from $\A{t-1}$ to obtain $\A t$ in
  line~\ref{delAtline} of \cref{algo:P3-interval}. It follows that
  $x\notin \A t$.  
\end{proof}

\begin{property}\label{property:keep_list_order}
  Let $Q=(q_1,\ldots,q_k)$ and $Q'=(q'_1,\ldots,q'_{k'})$ be two
  token lists such that $Q\lessgood Q'$.  Then for any $q_i\in Q$,
  $Q\qDel q_i\lessgood Q'\qDel q'_{k'}$ and for any interval~$x$,
  $Q\qAdd x\lessgood Q'\qAdd x$.
\end{property}
\begin{proof}
  For both insertion and deletion, the size constraint is clearly
  maintained (both list lengths respectively increase or decrease by
  1).  It remains to compare pairs of elements with the same index in
  both lists (such pairs are said to be \emph{aligned}).

  For the deletion case, $q_i$ is removed from $Q$. For any $j\neq i$,
  $q_j$ is now aligned with either~$q'_j$ (if $j<i$) or $q'_{j-1}$ (if
  $j>i$).  We have $\death(q_j)\leq \death(q'_j)$ since $Q\lessgood
  Q'$ and $\death(q'_j)\leq \death(q'_{j-1})$ since $Q'$ is
  sorted. Hence, $q_j$~is aligned with an interval in $Q'$ ending no later
  than $q_j$ itself.
  
  We now prove the property for the insertion of $x$ in both $Q$ and
  $Q'$.  An element $q$ of $Q$ or $Q'$ is said to be \emph{shifted} if it is
  higher than the insertion point of~$x$ (assuming that already-present%
  copies of $x$ in $Q$ or $Q'$ are not shifted), this is equivalent to
  $\death(q)<\death(x)$.  Note that if some~$q'_i$ is shifted but
  $q_i$ is not, then $\death(q'_i)<\death(x)\leq \death(q_i)$, a
  contradiction to $\death(q'_i)\geq \death(q_i)$.  This implies that
  the insertion point of $x$ in $Q'$ is not lower than the insertion%
  point in $Q$.
  
  Let $q$ be an interval of $Q\qAdd x$, now aligned with some $q'$ in
  $Q'\qAdd x$. We prove that $\death(q')\geq \death(q)$.  Assume first
  that $q=x$, then either $q'=x$, in which case trivially
  $\death(q')\geq \death(q)$, or $q'\neq x$. Then, $q'$ cannot be
  shifted (since $x$'s insertion point is not lower in $Q'$ than%
  in~$Q$), and $\death(q')\geq \death(x)= \death(q)$.

  Assume now that $q\neq x$. Then, $q=q_i$ for some $i$. With
  $Q\lessgood Q'$, we have $\death(q)\leq\death(q'_i)$. If $q'=q'_i$
  then we directly have $\death(q)\leq\death(q')$. Otherwise, exactly
  one of $q_i$ and $q'_i$ must be shifted. It cannot be $q'_i$ (the
  insertion point of $x$ is not lower in $Q'$ than in~$Q$), hence%
  $q_i$ is shifted and $q'_i$ is not. 
  In~$Q$ we have $\death(q_i)<\death(x)$, and in $Q'$ 
  interval $q'_i$ must be placed directly below $q'$ and %
  both cannot occur higher than $x$ (note that $q'=x$ is%
  possible), thus we have $\death(q'_i)\geq \death(q')\geq
  \death(x)$. Overall, we indeed have $\death(q)\leq\death(q')$.
\end{proof}

\noindent Using the proven properties, we can put the token lists
defined by a $P_3$-partition into relation with the token lists
generated by \cref{algo:P3-interval}.

The following two lemmas state that, on the one hand, if there is a
$P_3$-partition, then each token list created by \cref{algo:P3-interval}
is comparable with the corresponding~$\h{t}$, hence it 
always contains enough tokens to create the next list, up to $\A{2n}$, and
answer ``\texttt{true}'' in the end.  
On the other hand, if the algorithm
returns ``\texttt{true}'', then it is indeed possible to construct a
$P_3$-partition using (indirectly) the triples of intervals removed from the 
token list to create the links.

\begin{lemma}\label[lemma]{lem:P3-partition->h_t<=A_t} 
  If an interval graph $G$ has a $P_3$-partition $\Psol$, then, for all
  $0\le t \le 2n$, \cref{algo:P3-interval} defines list $\A t$ with
  $\h{t}\lessgood \A{t}$ and $\qSize{\h{t}}-\qSize{\A{t}}\equiv 0\pmod
  3$, where $\tokens(\Psol) = (\h0,\h1,\ldots,\h{2n})$.
\end{lemma}

\begin{proof}
  We show by induction that for any position $t$, $0\le t \le 2n$, the
  algorithm defines a list $\A{t}$ with $\h{t} \lessgood \A{t}$ and
  $\qSize{\h{t}}-\qSize{\A{t}}\equiv 0\pmod 3$.

  For $t=0$, \cref{algo:P3-interval} defines list $\A 0=\emptyset$,
  and $\h{0} = \emptyset \lessgood \A{0}$.  Consider now some $0<t\le
  2n$, and assume that the induction property is proven for $t-1$.
  
  If an interval~$x$ starts at position~$t$, then $x\notin \h{t-1}$,
  $x\notin \A{t-1}$, $\h{t}=\h{t-1}\qAdd(x,x)$, and
  \cref{algo:P3-interval} defines $\A{t}\coloneqq \A{t-1}\qAdd(x,x)$.  Then
  property $\qSize{\h {t}}-\qSize{\A{t}}\equiv 0\pmod 3$ is trivially
  preserved, and \cref{property:keep_list_order} implies
  $\h{t}\lessgood \A{t}$.

  If an interval~$x$ ends at~$t$, then we first show that
  \cref{algo:P3-interval} defines $\A{t}$.  Towards a
  contradiction, suppose that $\A{t}$ is not defined.  This means that
  $x\in\A{t-1}$ and $\qSize{\A{t-1}}\le 2$.  Then, $\qSize{\h{t-1}}\le
  2$ since $\h{t-1}\lessgood \A{t-1}$, which implies $\qSize{\h{t-1}}=
  \qSize{\A{t-1}}$ (since $\qSize{\h{t-1}}-\qSize{\A{t-1}}\equiv
  0\pmod 3$).  By \cref{property:token:delete_size=3}, we must have
  $x \notin \h{t-1}$ and $\h {t}=\h{t-1} $ (the second case).  
  Also, with
  $\h{t-1}\lessgood \A{t-1}$, the top element~$x'$ in $\h{t-1}$
  must have $\death(x')\leq \death(x)=t$.  By
  \cref{property:in_token_is_alive} and due to $x' \in T_{t-1}$, 
  $\death(x')>t-1$,
  i.\,e., $\death(x')=\death(x)$, and $x'=x$: a contradiction since
  $x\notin\h{t-1}=A_{t-1}$.
   
  We have shown that \cref{algo:P3-interval} defines $\A{t}$.
  Moreover, we have $\qSize{\h{t}}-\qSize{\h{t-1}}\in \{0,-3\}$
  (\cref{property:token:delete_size=3}), and
  $\qSize{\A{t}}-\qSize{\A{t-1}}\in \{0,-3\}$
  (\cref{algo:P3-interval}), hence $\qSize{\h
    {t}}-\qSize{\A{t}}\equiv 0\pmod 3$.  Note also that
  $\h{t}\lessgood\h{t-1}$ and $\A{t}\lessgood\A{t-1}$.
  
  If $x\notin \A{t-1}$, then $\A{t}=\A{t-1}$, and
  $\h{t}\lessgood\h{t-1}\lessgood \A{t-1} = \A{t}$.

  If $x\in \A{t-1}$ and $x\notin \h{t-1}$, then $\h{t}=\h{t-1}$
  (\cref{property:token:delete_size=3}). Observe that
  $\h{t-1}\lessgood\A{t-1}$ and, therefore, $\qSize{\h{t-1}} =
  \qSize{\A{t-1}}$ would imply the existence of an interval~$u\in
  \h{t-1}$ with $\death(u) \le \death(x)$. This is impossible since
  for all $u\in \h{t}=\h{t-1}$, by \cref{property:in_token_is_alive},
  $\death (u)>t=\death(x)$. Thus, one has
  $\qSize{\h{t-1}}<\qSize{\A{t-1}}$, which implies
  $\qSize{\h{t-1}}\leq \qSize{\A{t-1}}-3$.  Hence, since the
  \emph{top} three elements of $\A{t-1}$ are removed to obtain%
  $\A{t}$, from $\h{t-1}\lessgood\A{t-1}$ we conclude
  $\h{t-1}\lessgood \A{t}$ and, in turn, $\h{t}\lessgood \A{t}$.
     
  If $x\in \A{t-1}$ and $x\in \h{t-1}$, then let $(u,v,w)$ be the
  three deleted intervals from~$\h {t-1}$, and $(u',v',w')$ the top%
  three elements of $\A{t-1}$ (which are removed to obtain
  $\A{t}$). Then, applying \cref{property:keep_list_order} three
  times, we obtain $\h{t-1}\qDel (u,v,w)\lessgood \A{t-1}\qDel
  (u',v',w')$, i.\,e., $\h{t}\lessgood \A{t}$. 
\end{proof}

\noindent Before we move on to our last lemma for the interval graph algorithm, 
we introduce further notions necessary for constructing a $P_3$-partition 
from the list~$A_{2n}$ that our algorithm produces at step~$2n$.

\begin{definition}[Partial partition]
  Given an interval $x$ and a
  token list $Q$, we write $\occ{x}{Q}$ for the number of occurrences
  of interval $x$ in $Q$.  For $0\le t\le 2n$, let $\Psol =
  \{V_1,\ldots, V_k\}$ be a partition of~$\{u\in V\mid\birth (u)\le
  t\}$.  Then $\Psol$ is called a \emph{partial partition at $t$} if
  each~$V_j$ is either 
  \begin{itemize}
  \item a singleton $\{x\}$, in which case $\death(x)> t$, 
  \item an edge $\{x,y\}$, in which case $\max\{\death(x),\death(y)\} > t$, 
  \item a triple $\{x,y,z\}$ containing a $P_3$. 
  \end{itemize} 
  Note that a $P_3$-partition of an interval
  graph corresponds to a partial partition at~$t=2n$.
  A partial solution~$\Psol$ at~$t$ \emph{satisfies}~$\A{t}$ if
  \begin{itemize}
  \item for any singleton $\{x\}\in \Psol$ we have $\occ{x}{\A{t}}=2$, 
  \item for any edge $\{x,y\}\in \Psol$ with $\death(x)<\death(y)$ we have $\occ
    x{\A t}=0$ and $\occ y{\A t}=1$, and
  \item for any triple~$\{x,y,z\}\in
  \Psol$ we have $\occ x{\A t}=\occ y{\A t}=\occ z{\A t}=0$.
  \end{itemize}
\end{definition}

\noindent Note that, for any $x\in \A t$, since $\birth(x)\leq
t<\death(x)$~(\cref{property:in_token_is_alive}), it follows that
$x$ must be in a singleton or in an edge of any partial solution
satisfying $\A t$.
Moreover, for any $t$ and $x,y\in \A t$ with $x\neq y$, intervals $x$ and $y$
intersect (there is an edge between them in the interval graph).

\begin{lemma}\label[lemma]{lem:A_t-defined->partial solution}
  Let $G$ be an interval graph such that \cref{algo:P3-interval}
  returns \texttt{true} on~$G$.  Then $G$ admits a $P_3$-partition.
\end{lemma}

\begin{proof}
  We prove by induction that for any $t$ such that
  \cref{algo:P3-interval} defines $\A t$, there exists a partial
  solution at $t$ satisfying $\A t$.

  For $t=0$, the partial solution $\emptyset$ satisfies $\A 0$. Assume
  now that for some~$t\leq 2n$, \cref{algo:P3-interval}
  defines~$\A{t}$, and that there exists a partial solution~$\Psol$
  at~$t-1$ satisfying~$\A{t-1}$.

  First, if $t=\birth(x)$ for some interval $x$, then let
  $\Psol'\coloneqq \Psol\cup\{\{x\}\}$. Thus, $\Psol'$~is now a partial
  solution at $t$ (it partitions every interval with earlier starting
  point into singletons, edges and $P_3$s) which satisfies $\A{t}$
  since by construction of~$\A{t}$ by \cref{algo:P3-interval},
  $\occ{x}{\A{t}}=2$.

  Now assume that $t=\death(x)$ with $x\notin \A{t-1}$. Then, in
  $\Psol$, either $x$ is part of an edge $\{x,y\}$ with $\death(y)>t$,
  or $x$ is part of a $P_3$. In both cases, $\Psol'\coloneqq \Psol$ is a
  partial solution at $t$ which satisfies $\A{t}=\A{t-1}$.

  We now explore the case where $t=\death(x)$ with $x\in \A{t-1}$.
  Then, the top element of $\A{t-1}$ must be $x$ (no other interval 
  $u\in\A{t-1}$ can have $t-1<\death(u)\leq \death(x)$).  Let $y$ and
  $z$ be the two elements below~$x$ in $\A{t-1}$.  Then, by 
  construction, $\A{t}=\A{t-1}\qDel (x,y,z)$ and $\death(x)\leq
  \death(y)\leq \death(z)\leq \death(u)$ for all $u\in \A t$. We
  create a partial solution $\Psol'$ at~$t$ depending on the number of
  occurrences of $x$, $y$, and $z$ in~$\A{t-1}$.

  \looseness=-1 If $x= y$ (hence, $\occ{x}{\A{t-1}}=2$) and $\occ{z}{\A{t-1}}=2$,
  then $\Psol$ contains two singletons~$\{x\}$ and $\{z\}$. Let
  $\Psol'\coloneqq (\Psol\setminus\{\{x\},\{z\}\})\cup \{\{x,z\}\}$. Then,
  $\Psol'$ is indeed a partial solution at~$t$ (since $\{x,z\}$ is an
  edge with $\death(z)>t$) that satisfies $\A{t}$, since
  $\occ{x}{\A{t}}=0$ and $\occ{z}{\A{t}}=1$.

  If $x= y$ (hence, $\occ{x}{\A{t-1}}=2$) and $\occ{z}{\A{t-1}}=1$,
  then $\Psol$ contains a singleton~$\{x\}$ and an edge~$\{z,u\}$.
  Also, note that $\occ{u}{\A{t-1}}=0$, that is, $u\notin\A{t-1}$.
  Because there is an edge $\{x,z\}$, the triple $\{x,z,u\}$ contains
  a~$P_3$.  Let $\Psol'\coloneqq (\Psol\setminus\{\{x\},\{z,u\}\})\cup
  \{\{x,z,u\}\}$. Then $\Psol'$ is a partial solution at~$t$
  that satisfies $\A{t}$, since
  $\occ{x}{\A{t}}=\occ{z}{\A{t}}=\occ{u}{\A{t}}=0$.

  If $z= y$ (hence, $ \occ{x}{\A{t-1}}=1$ and $\occ{z}{\A{t-1}}=2$),
  then similarly $\Psol$ contains an edge~$\{x,u\}$ and a singleton
  $\{z\}$: $\Psol'\coloneqq (\Psol\setminus\{\{x,u\},\{z\}\})\cup
  \{\{x,z,u\}\}$ is a partial solution at $t$ that satisfies $\A{t}$.

  If $y\neq x$ and $y\neq z$ (hence, $\occ{x}{\A{t-1}}=1$ and
  $\occ{y}{\A{t-1}}=1$), and $\occ{z}{\A{t-1}}=2$, then $\Psol$
  contains two edges~$\{x,u\}$ and~$\{y,v\}$ and a singleton
  $\{z\}$. Recall that $v,u\notin\A{t-1}$. Assume first that
  $\birth(y)<\birth(x)$, then interval $u$ intersects $y$,
  and~$\{y,u,v\}$ contains a $P_3$. Also, $\{x,z\}$ forms an edge with
  $\occ{z}{\A{t}}=1$: define
  $\Psol'\coloneqq (\Psol\setminus\{\{x,u\},\{y,v\},\{z\}\})\cup
  \{\{y,u,v\},\{x,z\}\}$.  In the case where $\birth(x)<\birth(y)$,
  $\{x,u,v\}$ contains a $P_3$ and
  $\Psol'\coloneqq (\Psol\setminus\{\{x,u\},\{y,v\},\{z\}\})\cup
  \{\{x,u,v\},\{y,z\}\}$ is a partial solution at $t$ that satisfies
  $\A{t}$.

  Finally, we have a similar situation when $y\neq x$, $y\neq z$ and
  $\occ{z}{\A{t-1}}=1$: then, $\Psol$ contains three edges $\{x,u\}$,
  $\{y,v\}$ and $\{z,w\}$.  If $\birth(y)<\birth(x)$, then
  both~$\{y,u,v\}$ and $\{x,z,w\}$ contain~$P_3$s. Otherwise,
  $\{x,u,v\}$ and $\{y,z,w\}$ contain~$P_3$s.  Thus, we define
  $\Psol'\coloneqq (\Psol\setminus\{\{x,u\},\{y,v\},\{z,w\}\})\cup
  \{\{y,u,v\}, \{x,z,w\}\}$ and
  $\Psol'\coloneqq (\Psol\setminus\{\{x,u\},\{y,v\},\{z,w\}\})\cup
  \{\{x,u,v\}, \{y,z,w\}\}$ respectively. In both cases, $\Psol'$ is a
  partial solution at $t$ that satisfies~$\A{t}$.

\looseness=-1  Overall, if \cref{algo:P3-interval} returns \texttt{true}, then
  it defines $\A {2n}$. According to the property we have proven,
  there exists a partial solution at $t=2n$, hence $G$ has a
  $P_3$-partition.  
\end{proof}

\noindent The above lemmas allow us to conclude the correctness of
\cref{algo:P3-interval}.

\begin{theorem} \label{thm:P3-interval} \probPthreePart{} on interval
  graphs is solvable in $O(n\log n+m)$~time.
\end{theorem}

\begin{proof}
  Let $G$ be an interval graph.  To prove the theorem, we show that
  \cref{algo:P3-interval} returns \texttt{true} on $G$ if and only if
  $G$ has a $P_3$-partition.  The ``only if'' part is the statement of
  \cref{lem:A_t-defined->partial solution}.  For the ``if'' part,
  suppose that $G$ has a $P_3$-partition~$\Psol$. Then
  \cref{lem:P3-partition->h_t<=A_t} implies that
  \cref{algo:P3-interval} defines list $\A{t}$ at position~$t=2n$,
  which means it returns \texttt{true}.

  It remains to prove the running time bound. We first preprocess the
  input as follows: in $O(n+m)$~time, we can get an interval
  representation of an interval graph with $n$~intervals that use
  start and end points in~$\{1,\dots,n\}$~\cite[Section~8]{COS09}.  We
  modify this representation so that each position is the start or end
  point of at most one interval: first, for each interval, we add its
  start point to the beginning of a list~$L$ and its end point to the
  end of~$L$.  We sort~$L$ using a stable sorting algorithm like
  counting sort in $O(n)$~time.  The result is a sorted list~$L$ that,
  for each position, contains first the start points and then the end
  points.  Now, in $O(n)$~time, we iterate over~$L$ and reassign each
  event points to its own position in $\{1,\dots,2n\}$ in the order of
  its appearance in~$L$. At the same time, we build an $2n$-element
  array~$B$ such that $B[i]$ holds a pointer to the interval starting
  or ending at event point~$i$ (there is at most one such
  interval). It follows that all preprocessing works in $O(n+m)$~time.

\looseness=-1  After this preprocessing, each of the $O(n)$~iterations for some
  $t\in\{1,\dots,2n\}$ of the loop in line~\ref{algoLoop} of
  \cref{algo:P3-interval} is executed in $O(\log n)$~time: in
  constant time, we get the interval $B[t]$ starting or ending at~$t$
  and each operation on the token list can be executed in $O(\log n)$
  time if it is implemented as a balanced binary tree (note that only
  the current value of $\A t$ need to be kept at each point, hence it
  is never necessary for the algorithm to make a copy of the whole
  token list).
\end{proof}

\section{Cographs}\label{sec:cographs}
A cograph is a graph that does not contain a
$P_4$ (path on four vertices) as an induced subgraph. Cographs allow for
a so-called \emph{cotree} to be computed in linear time~\cite{CPS85}. 

\begin{definition}
\looseness=-1  A cotree~$\cot(G)$ of a cograph~$G=(V,E)$ is a rooted binary
  tree~$T=(V_T,E_T,r), r \in V_T$, 
  where each internal node is assigned a label in $\{\union,\join\}$
  and the set of leaves
  corresponds to the original set~$V$ of vertices such that:
  \begin{itemize}
  \item A subtree consisting of a single \emph{leaf node} corresponds
    to an induced subgraph with a single vertex.
  \item A subtree rooted at a \emph{union node},
    labeled~\textnormal{``$\union$''}, corresponds to the disjoint
    union of the subgraphs defined by the two children of that node.
  \item A subtree rooted at a \emph{join node},
    labeled~\textnormal{``$\join$''}, corresponds to the \emph{join}
    of the subgraphs defined by the two children of that node; that
    is, the union of the two subgraphs with additional edges between
    every two vertices corresponding to leaves in different subtrees.
  \end{itemize}
  Consequently, the subtree rooted at the root~$r$ of $\cot(G)$
  corresponds to~$G$.
\end{definition}

\noindent Using a
dynamic programming approach on the cotree representation of the
cograph, we can solve \probStarPart in polynomial time.

\begin{theorem}\label{thm:cographs}
  \probStarPart can be solved in $O(kn^2)$ time on cographs.
\end{theorem}

\begin{proof}%
  \looseness=-1 Let $(G=(V,E),s)$ be a \probStarPart instance with~$G$ being a
  cograph.  Let $T=(V_T,E_T,r)=\cot(G)$ denote the cotree of $G$.
  Furthermore, for any node~$x \in V_T$, let $T[x]$ denote the subgraph
  of~$G$ that corresponds to the subtree of $T$ rooted at~$x$.

  We define a dynamic programming table~$L$ as follows.  For every
  node $x \in V_T$ and every non-negative integer $c \le k$, the table
  entry $L[x,c]$ denotes the maximum number of leaves in~$T[x]$ that
  are covered by a center in $T[x]$ when $c$ vertices in~$T[x]$ are
  centers.  Consequently, $(G,s)$ is a yes-instance if and only
  if~$L[r,k]= ks$.
  Now, let us describe how to compute~$L$ processing the cotree~$T$
  bottom up.

  \paragraph{Leaf nodes.}  For a leaf node~$x$, either the only vertex~$v$
  from~$T[x]$ is a center or not.  In both cases no leaf in $T[x]$ is
  covered by~$v$.  Thus, $L[x,0]=L[x,1]=0$ and $\forall c>1:
  L[x,c]=-\infty$.

  \paragraph{Union nodes.}  Let~$x$ be a node labeled with ``$\union$'' and
  let $x_1$ and $x_2$ be its children.  Note that there is no edge
  between a vertex from~$T[x_1]$ and a vertex from~$T[x_2]$, neither
  in~$T[x]$ nor in any other subgraph of~$G$ corresponding to
  any~$T[x']$, $x' \in V_T$.  Thus, for every leaf~$v$ in~$T[x]$ that is
  covered by a center~$v'$ from~$T[x]$, it holds that either both $v$
  and $v'$ are in $T[x_1]$ or both are in $T[x_2]$.  Hence, it follows
  $L[x,c]=\max_{c_1+c_2=c}(L[x_1,c_1] + L[x_2,c_2])$.

  \paragraph{Join nodes.}  Let~$x$ be a node labeled with ``$\join$'' and
  let $x_1$ and $x_2$ be its children.  Join nodes are more
  complicated than leaf or union nodes for computing the table
  entries, because these nodes actually introduce the edges.  However,
  they always introduce \emph{all possible} edges between vertices
  from $T[x_1]$ and $T[x_2]$ which has some nice consequences.  The
  idea is that the maximum number of leaves in~$T[x]$ that are covered
  by centers in~$T[x]$ is achieved by maximizing the number of leaves
  from~$T[x_2]$ that are covered by centers from~$T[x_1]$ and vice
  versa.
   
  To compute~$L[x,c]$, we introduce an auxiliary table~$A$ as follows.
  For every pair~$c_1, c_2$ of non-negative integers with~$c_1+c_2=c$,
  the table entry $A[c_1,c_2]$ denotes the maximum number of leaves
  in~$T[x]$ that are covered by a center in~$T[x]$ when $c_1$ vertices
  in~$T[x_1]$ are centers and $c_2$ vertices in $T[x_2]$ are centers.
  To this end, let $\ell_i$, $i \in \{1,2\}$, be the number of leaves in
  the desired \starpartition{s} being in~$T[x_i]$.  (Note that in
  every solution every vertex is either a center or a leaf and a leaf
  is not necessarily already covered within the current~$T[x]$.
  That is, $\ell_i$ can be larger than the number of leaves
  covered by a center in~$T[x]$.)
  Moreover, $\ell_i=|V(T[x_i])|-c_i$, where $V(T[x_i])$ is the set of
  vertices in $T[x_i]$.  To compute the auxiliary table~$A$, we
  consider three cases:

  \def\noCdot{}

  \paragraph{Case 1: \boldmath$(c_1 \noCdot s > \ell_2) \wedge (c_2 \noCdot s >
    \ell_1)$.}  In this case, we can cover all leaves in $T[x]$ by
    covering the leaves from $T[x_1]$ with centers from $T[x_2]$ and
    vice versa.  Thus, $A[c_1,c_2]=\ell_1+\ell_2$.
    \paragraph{Case 2: \boldmath$(c_1 \noCdot s \le \ell_2) \wedge (c_2 \noCdot s
      \le \ell_1)$.} In this case, we can cover $c \noCdot s$ leaves
    in~$T[x]$ by covering $c_1 \noCdot s$ leaves from~$T[x_2]$ by
    centers from~$T[x_1]$ and $c_2 \noCdot s$ leaves from~$T[x_1]$ by
    centers from~$T[x_2]$.  This is obviously the best one can do.
    Thus, $A[c_1,c_2]=c \noCdot s$.
  \paragraph{Case 3: \boldmath$(c_1 \noCdot s > \ell_2) \wedge (c_2 \noCdot s \le \ell_1)$ or
    $(c_1 \noCdot s \le \ell_2) \wedge (c_2 \noCdot s > \ell_1)$.} 
    In this case it is also optimal to greedily maximize the number of
    leaves from $T[x_2]$ that are covered by centers from $T[x_1]$ and
    vice versa. To see this, let $y_i, i \in \{1,2\}$, denote the number
    of leaves from $T[x_i]$ that are covered by a center
    from~$T[x_i]$.
    More precisely, assume that $y_1$ and $y_2$ are both
    greater than zero.  Then, repeatedly take one center from~$T[x_1]$
    covering a leaf in $T[x_2]$ and one center from $T[x_2]$ covering
    a leaf in $T[x_1]$ and exchange their leaves until either $y_1$ or
    $y_2$ is zero (if both become zero, we would be in Case~2). 

    Without loss of generality, let $y_1>0$ and $y_2=0$.  
    Note that this corresponds
    to the first subcase, i.\,e., $(c_1 \noCdot s > \ell_2) \wedge (c_2 \noCdot s \le
    \ell_1)$---the other subcase works analogously.  As $y_2=0$ and $c_2
    \noCdot s \le \ell_1$, we can assume that $c_2$ centers from
    $T[x_2]$ cover altogether $c_2 \noCdot s$ leaves from~$T[x_1]$.
    Furthermore, all~$\ell_2$ leaves from~$T[x_2]$ are covered by
    centers in~$T[x_1]$.  Since $c_1 \noCdot s > \ell_2$, the
    centers in~$T[x_1]$ might additionally cover some number~$\ell'$
    of leaves from~$T[x_1]$. We thus have $A[c_1,c_2]=c_2 \noCdot  s + \ell_2 + \ell'$. 
    We now compute the maximum possible value of $\ell'$.
    Clearly:
    \begin{itemize}
    \item $\ell'$ is at most $c_1  \noCdot  s - \ell_2$, the maximum
      number of leaves that can be covered by $c_1$ centers after 
      $\ell_2$ leaves are covered in $T[x_2]$,
    \item $\ell'$ is at most $\ell_1 - c_2 \noCdot  s$, the maximum
      number of leaves that are not already covered by centers from
      $T[x_2]$, and
    \item $\ell'$ is at most $L[x_1,c_1]$, the maximum number of
      leaves from $T[x_1]$ that can be covered by centers from
      $T[x_1]$.  
    \end{itemize}
    Hence, $\ell' \leq  \min(c_1 \noCdot  s - \ell_2, \ell_1 - c_2 \noCdot  s, L[x_1,c_1])$  
    
    Conversely, for $\ell''= \min(c_1 \noCdot  s - \ell_2, \ell_1 - c_2 \noCdot  s, L[x_1,c_1])$, it is possible
    for $c_1$ centers in~$T[x_1]$ to cover $\ell''$ leaves in $T[x_1]$ and $\ell_2$ leaves in $T[x_2]$,
    and for $c_2$ centers in~$T[x_2]$ to cover $c_2  s$~leaves in $T[x_1]$.
    Here, the property that a join node introduces all
    possible edges between the two subgraphs is crucial, because we
    can therefore simply cover leaves from~$T[x_1]$ by centers
    from~$T[x_1]$ in an optimal way.  (\emph{Each} center from
    $T[x_1]$ can cover \emph{each} leaf from $T[x_2]$ and vice
    versa.) So $\ell'\geq \ell'' =\min(c_2 \noCdot  s - \ell_1, \ell_2 - c_1 \noCdot  s, L[x_2,c_2])$. Overall,
    \begin{align*}
A[c_1&,c_2]=\\&
      \begin{cases}
       \ell_1+\ell_2  	&\text{if $(c_1 \noCdot s > \ell_2) \wedge (c_2 \noCdot s > \ell_1)$ } \\
       c \noCdot s 	&\text{if $(c_1 \noCdot s \le \ell_2) \wedge (c_2 \noCdot s \le \ell_1)$} \\
       c_2 \noCdot  s + \ell_2 + \min(c_1 \noCdot  s - \ell_2, \ell_1 - c_2  \noCdot s,
         L[x_1,c_1])	&\text{if $(c_1  \noCdot s > \ell_2) \wedge (c_2  \noCdot s \le \ell_1)$}\\
       c_1  \noCdot s + \ell_1 + \min(c_2 \noCdot  s - \ell_1, \ell_2 - c_1 \noCdot  s,
         L[x_2,c_2])	&\text{if $(c_1 \noCdot  s \le \ell_2) \wedge (c_2  \noCdot s > \ell_1)$}.
     \end{cases}
    \end{align*}
    Finally, we compute~$L[x,c]$ by considering the auxiliary table,
    that is,
    \[L[x,c]=\max_{c_1+c_2=c}(A[c_1,c_2]).\]
 The $O(kn^2)$ running time of this algorithm can be seen as follows:
 Computing the cotree representation runs in linear time~\cite{CPS85}.
 The table size of the dynamic program is bounded by $O(kn)$---there are $O(n)$~nodes
 in the cotree and $c\leq k$.  Since $V(T[x_i])$ corresponds to the
 set of leaf nodes of the subtree of~$T$ rooted in~$x_i$, the sizes
 $|V(T[x_i])|$ can be precomputed in linear time for each node~$x_i$
 of the cotree.  Hence, computing a table entry costs at most
 $O(n)$.
\end{proof}

\section{Bipartite permutation graphs}\label{sec:biperm}

In this section, we show that \probStarPart can be solved in
$O(n^2)$~time on bipartite permutation graphs.  The class of bipartite
permutation
graphs is the intersection of the class of bipartite graphs and
the class of permutation graphs.  An alternative characterization of
 bipartite permutation graphs
can be given using \emph{strong orderings} of the vertices of
a bipartite graph:

\begin{definition}[\citet{SBS87}]\label[definition]{def:strongorder}
  A \emph{strong ordering}~$\strongorder$ of the vertices of a
  bipartite graph $G=(U,W,E)$ is the union of a total order~$\strongorder_{U}$ of $U$ and a total order~$\strongorder_{W}$ of
  $W$, such that, for all edges $\{u,w\}$, $\{u',w'\}$ in $E$ with $u,u'
  \in U$ and $w,w' \in W$, $u \strongorder u'$ and $w' \strongorder w$
  implies that there are edges $\{u,w'\}$ and $\{u',w\}$ in $E$.
\end{definition}

A graph is a bipartite permutation graph if and only if it is
bipartite and there is a strong ordering of its vertices; a strong ordering can be
computed in linear time~\cite{SBS87}.

In a bipartite graph~$G$ with vertex set~$U\cup W$, if the subgraph
induced by a size-$(s+1)$ vertex subset $X\subseteq U \cup W$ contains
an \sstar, then this induced subgraph \emph{is} a star---there is only
one way to choose the star center.  Thus, we refer to $G[X]$ as a
star.  We denote by $\scenter(X)$ the center of the star~$G[X]$.
Observe that the number $k_U$ of star centers in $U$ and the number
$k_W$ of star centers in $W$ are uniquely determined by the sizes
$|U|$ and $|W|$ of the two independent vertex sets and by the number $s$
of leaves in a star, since
\begin{alignat*}{4}
  |U|&=k_U+s\cdot k_W &\quad\text{ and }\quad |W| &= k_W+s\cdot k_U\\
  \intertext{and therefore}
  k_U&=\frac{|U| - |W| \cdot s}{1 - s^2}  &\quad\text{ and }\quad k_W&=\frac{|W| - |U| \cdot s}{1 - s^2}.
\end{alignat*}

If these numbers are not positive integers, then $G$ does not have an
\starpartition{s}. Thus, we assume throughout this section that $k_U$
and $k_W$ are positive integers.

Our key to obtain star partitions on bipartite permutation graphs
is a structural result that only a certain ``normal form'' of star
partitions has to be searched for. This paves the way to developing a
dynamic programming algorithm exploiting these normal forms.  We define
these structural properties of an \starpartition{s} of bipartite
permutation graphs in the following.

  Let $(G,s)$ be a \probStarPart instance, where $G=(U,W,E)$ is a
  bipartite permutation graph, $\strongorder$ is a strong ordering of
  the vertices, and $\strongordereq$ is the reflexive closure
  of~$\strongorder$.  For two vertex sets~$A,B$, we also write $A\myll
  B$ if for all vertices $v\in A$ and $w\in B$, we have $v\strongorder
  w$.

  \looseness=-1 Assume that $G$ admits an \starpartition{s}~$\Psol$.  Let $X\in
  \Psol$ form a star. By $\leftm(X)$ (respectively by $\rightm(X)$),
  we denote the leftmost (that is, the minimum), respectively the
  rightmost (that is, the maximum) leaf of~$X$ with respect
  to~$\strongorder$.  The \emph{scope} of star~$X$ is the
  set~$\closure(X)\coloneqq \{v\mid x_{l}\strongordereq v\strongordereq
  x_{r}\}$ containing all vertices from~$x_l=\leftm(X)$ to~$x_r=\rightm(X)$.  The \emph{width} of star~$X$ is the cardinality
  of its scope, that is, $\width(X)\coloneqq |\closure(X)|$.  The \emph{width}
  of $\Psol$, $\width(\Psol)$, is the sum of $\width(X)$ over all
  $X\in \Psol$.

  Let $e=\{u, w\}$ and $e'=\{u', w'\}$ be two edges.  We say that $e$
  and $e'$ \emph{cross} each other if it holds that $u \strongorder u'$ and
  $w'\strongorder w$ or if it holds that $u' \strongorder u$ and $w \strongorder
  w'$.  The \emph{edge-crossing number} of two stars~$X,Y\in\Psol$ is
  the number of pairs of crossing edges~$e, e'$ with respect to the given
  strong order~$\strongorder$ where $e$ is an edge
  of~$X$ and~$e'$ is an edge of~$Y$.  The edge-crossing
  number~$\edgecrossingnum(\Psol)$ of~$\Psol$ is the sum of the
  edge-crossing numbers over all pairs of stars $X\neq Y \in \Psol$.

  We identify the possible configurations of two stars, depending on
  the relative positions of their leaves and centers, see
  \cref{fig:crossing}. Among those, the following two configurations
  are favorable: Given $X,Y\in\Psol$, we say that $X$ and $Y$ are
  \begin{itemize}
  \item \emph{non-crossing} if their edge-crossing number is zero;
  \item \emph{interleaving} if $\scenter(X)\in \closure(Y)$ and
    $\scenter(Y)\in \closure(X)$;
  \end{itemize}

  We say that $\Psol$ is \emph{good} if any two stars $X\neq Y \in
  \Psol$ are either non-crossing or interleaving.  We define the \emph{score}
  of $\Psol$ as the tuple~$(\width(\Psol), \edgecrossingnum(\Psol))$.
  We use the lexicographical order %
  to compare scores.

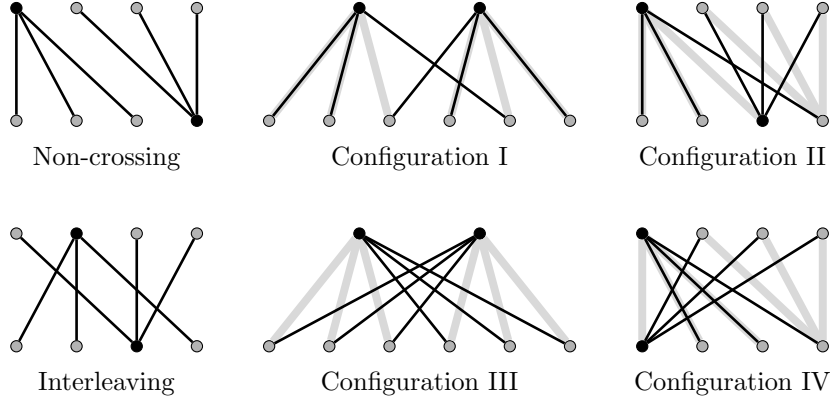
\begin{figure}[t!]
  \centering \def\sep{1.5} \def\xgrid{5.2cm} \def\ygrid{-3cm}
  \begin{tikzpicture}[draw=black!100, xscale=0.8]
    \tikzstyle{newedge}=[draw=gray!30, line width=3pt]
    \tikzstyle{oldedge}=[draw=black,line width=1pt]

    \begin{scope}[xshift=-\xgrid]
      \node[center] (u-1) at (1,\sep) {}; \foreach \x in {2,3,4}
      \node[vertex] (u-\x) at (\x,\sep) {};
      \node[center] (v-4) at (4,0) {}; \foreach \x in {1,2,3}
      \node[vertex] (v-\x) at (\x,0) {};
      \foreach \x in {1,2,3} \draw[oldedge] (u-1) -- (v-\x); \foreach
      \x in {2,3,4} \draw[oldedge] (u-\x) -- (v-4); \node at (2.5,-.5)
      {Non-crossing};
    \end{scope}
  
    \begin{scope}[xshift=-\xgrid,yshift=\ygrid]
      \node[center] (u-2) at (2,\sep) {}; \foreach \x in {1,3,4}
      \node[vertex] (u-\x) at (\x,\sep) {};
      \node[center] (v-3) at (3,0) {}; \foreach \x in {1,2,4}
      \node[vertex] (v-\x) at (\x,0) {};
      \foreach \x in {1,2,4} \draw[oldedge] (u-2) -- (v-\x); \foreach
      \x in {1,3,4} \draw[oldedge] (u-\x) -- (v-3); \node at (2.5,-.5)
      {Interleaving};
    \end{scope}

    \begin{scope}[]
      \node[center] (u-2) at (1.5,\sep) {}; \node[center] (u-4) at
      (3.5,\sep) {};
      \foreach \x in {0,1,2,3,4,5} \node[vertex] (v-\x) at (\x,0) {};

      \foreach \x in {0,1,2} \draw[newedge] (u-2) -- (v-\x); \foreach
      \x in {3,4,5} \draw[newedge] (u-4) -- (v-\x);
      \foreach \x in {0,1,4} \draw[oldedge] (u-2) -- (v-\x); \foreach
      \x in {2,3,5} \draw[oldedge] (u-4) -- (v-\x); \node at (2.5,-.5)
      {Configuration~I};
    \end{scope}
  
    \begin{scope}[yshift=\ygrid]
      \node[center] (u-2) at (1.5,\sep) {}; \node[center] (u-4) at
      (3.5,\sep) {};
      \foreach \x in {0,1,2,3,4,5} \node[vertex] (v-\x) at (\x,0) {};
      \foreach \x in {0,1,2} \draw[newedge] (u-2) -- (v-\x); \foreach
      \x in {3,4,5} \draw[newedge] (u-4) -- (v-\x);
      \foreach \x in {3,4,5} \draw[oldedge] (u-2) -- (v-\x); \foreach
      \x in {0,1,2} \draw[oldedge] (u-4) -- (v-\x); \node at (2.5,-.5)
      {Configuration~III};
    \end{scope}

    \begin{scope}[xshift=\xgrid]
      \node[center] (u-1) at (1,\sep) {}; \foreach \x in {2,3,4}
      \node[vertex] (u-\x) at (\x,\sep) {};
      \node[center] (v-3) at (3,0) {}; \foreach \x in {1,2,4}
      \node[vertex] (v-\x) at (\x,0) {};
      \foreach \x in {1,2,3} \draw[newedge] (u-1) -- (v-\x); \foreach
      \x in {2,3,4} \draw[newedge] (u-\x) -- (v-4);
      \foreach \x in {1,2,4} \draw[oldedge] (u-1) -- (v-\x); \foreach
      \x in {2,3,4} \draw[oldedge] (u-\x) -- (v-3);
      
      \node at (2.5,-.5) {Configuration~II};
    \end{scope}
  
    \begin{scope}[xshift=\xgrid,yshift=\ygrid]
      \node[center] (u-1) at (1,\sep) {}; \foreach \x in {2,3,4}
      \node[vertex] (u-\x) at (\x,\sep) {};
      \node[center] (v-1) at (1,0) {}; \foreach \x in {2,3,4}
      \node[vertex] (v-\x) at (\x,0) {};
      \foreach \x in {1,2,3} \draw[newedge] (u-1) -- (v-\x); \foreach
      \x in {2,3,4} \draw[newedge] (u-\x) -- (v-4);
      \foreach \x in {2,3,4} \draw[oldedge] (u-1) -- (v-\x); \foreach
      \x in {2,3,4} \draw[oldedge] (u-\x) -- (v-1);
      
      \node at (2.5,-.5) {Configuration~IV};
    \end{scope}

  \end{tikzpicture}
  \caption{Possible interactions between two stars of a
    partition. Centers are drew black.
    The four possible configurations of star centers and scopes
    that are neither non-crossing nor interleaving are labeled
    I to IV.
    By \cref{lem:only_noncrossing+interleaving}, any partition
    containing one of the configurations I to IV can be edited to
    reduce the score (see the thick gray edges). }
  \label{fig:crossing}
\end{figure}

These definitions allow us to observe the following property 
and show a normal form of star
partitions in bipartite permutation graphs.

\begin{property} \label{prop:multiple-crossings} Let $u_0\strongorder
  u_1$ and $w_0\strongorder w_1$ be four vertices such that edges
  $\{u_0,w_1\}$ and $\{u_1,w_0\}$ are in $G$.  Then, $G$ has edges
  $\{u_0,w_0\}$ and $\{u_1,w_1\}$ and, for any edge~$e$ crossing one
  (respectively both) edge(s) in $\{\{u_0,w_0\},\{u_1,w_1\}\}$, $e$ crosses
  one (respectively both) edge(s) in $\{\{u_0,w_1\},\{u_1,w_0\}\}$.
\end{property}
\begin{proof}
  The existence of the edges $\{u_0,w_0\}$ and $\{u_1,w_1\}$ is a
  direct consequence of \cref{def:strongorder}.  Let $e=\{u,w\}$ be an
  edge crossing $\{u_0,w_0\}$ and/or $\{u_1,w_1\}$.  We consider the
  cases where $u\strongorder u_0$ and where $u_0\strongorder u
  \strongorder u_1$ (the case $u_1\strongorder u$ being symmetrical
  to $u\strongorder u_0$).
 
  \looseness=-1 If $u\strongorder u_0$, then $w_0\strongorder w$, and $e$ crosses
  both $\{u_0,w_0\}$ and $\{u_1,w_0\}$. Also, if $e$~crosses
  $\{u_1,w_1\}$, then $e$ also crosses $\{u_0,w_1\}$, which proves the
  property for this case.
 
  \looseness=-1 If $u_0\strongorder u\strongorder u_1$, then if $e$ crosses~$\{u_0,w_0\}$, then $e$ also crosses $\{u_0,w_1\}$. If $e$ crosses~$\{u_1,w_1\}$, then $e$ also crosses $\{u_1,w_0\}$.  Overall, the
  property is thus proven for all cases. 
\end{proof}

\noindent Our main structural lemma now is the following.

\begin{lemma}\label[lemma]{lem:only_noncrossing+interleaving}
  Any \starpartition{s} of a bipartite permutation graph $G$ with
  minimum score is a good \starpartition{s}, that is, any two stars
  are either non-crossing or interleaving.
\end{lemma}

\begin{proof}
  Let $\Psol$ be an \starpartition{s} for $G$.  First, we show that any two stars $X\neq Y\in \Psol$ are 
  non-crossing, interleaving, or in one of the
  following four configurations (possibly after exchanging the roles
  of $X$ and $Y$, see \cref{fig:crossing} for an illustration):
  \begin{enumerate}[{Configuration}~I.]
  \item $\closure(X)\cap\closure(Y)\neq \emptyset$;
  \item $\scenter(Y)\in\closure(X)$ and
    $\scenter(X) \not\in \closure(Y)$;
  \item $\scenter(X) \strongorder \scenter(Y)$ and
    $\closure(Y) \myll \closure(X)$;
  \item $\scenter(X)\myll \closure(Y)$ and
    $\scenter(Y) \myll \closure(X)$ or, symmetrically, $\closure(Y) \myll
    \scenter(X)$ and $\closure(X) \myll \scenter(Y)$.
  \end{enumerate}
  
  First, assume that $\scenter(X)$ and $\scenter(Y)$ are both either
  in~$U$ or in~$W$. Furthermore, assume, without loss of generality,
  that $\scenter(X)\strongorder \scenter(Y)$.  If $X$ and~$Y$ are not
  in Configuration~I, then either $\closure(X)\myll \closure(Y)$ or
  $\closure(Y) \myll \closure(X)$.  If~$\closure(X)\myll \closure(Y)$,
  then $G[X]$ and $G[Y]$ are non-crossing. Otherwise, $\closure(Y)
  \myll \closure(X)$ and, hence, Configuration~III holds.

  If $\scenter(X)$ and $\scenter(Y)$ are in different vertex sets and
  if $X$ and $Y$ are not in Configuration~IV, then $\scenter(X) \in
  \closure(Y)$ and/or $\scenter(Y) \in \closure(X)$.  If~$\scenter(X)
  \in \closure(Y)$ and $G[X]$ and $G[Y]$ are not interleaving, then
  $\scenter(Y) \notin \closure(X)$ and we are in Configuration~II.
  Otherwise, if $\scenter(Y)\in\closure(X)$ and, again, $G[X]$ and
  $G[Y]$ are not interleaving, then $\scenter(X) \notin \closure(Y)$
  and we are again in Configuration~II.
  
  We now prove that a minimum-score \starpartition{s} $\Psol$ does not
  contain any pair of stars~$X\neq Y \in \Psol$ in Configurations~I,
  II, III or~IV (see \cref{fig:crossing}). For each such configuration, we construct an
  \starpartition{s}~$\Psol'$ with a score strictly smaller than~$\Psol$.
  
  \paragraph{Configuration~I.} Let $X,Y$ be two stars of $\Psol$ in
  Configuration~I, that is, $\closure(X)\cap\closure(Y)\neq
  \emptyset$.  Write $x_c=\scenter(X)$ and $y_c=\scenter(Y)$.  Then,
  $x_c$ and~$y_c$ are either both in~$U$ or both in~$W$. Without loss
  of generality, assume $x_c\strongorder y_c$.  Write
  $\{z_1,z_2,\ldots,z_{2s}\}$ for the union of the leaves of~$X$
  and~$Y$, with indices taken such that $z_i\strongorder z_j$ for
  $1\leq i<j\leq2s$.  Let~$Z_l=\{z_1,\ldots, z_s\}$
  and~$Z_r=\{z_{s+1},\ldots, z_{2s}\}$. We first show that both vertex
  sets $Z_l\cup\{x_c\}$ and $Z_r\cup\{y_c\}$ form a star in $G$.
   
  Let $k$ be the index such that $z_k=\leftm(Y)$.  Then, since the
  scopes of~$X$ and~$Y$ intersect, $z_k$~cannot be to the right of all
  the leaves of $G[X]$, hence we have~$k\leq s$, and $z_k\myll
  Z_r$. Consider now any $z\in Z_r$. If $z\in Y$, then there exists an
  edge $\{z, y_c\}$ in~$G$. If $z\in X$, then there exists an edge
  $\{z, x_c\}$ in $G$ that crosses~$\{z_k, y_c\}$
  (since~$z_k\strongorder z$ and $x_c\strongorder y_c$).  Thus, there
  also exists an edge $\{z, y_c\}$ in~$G$ by \cref{def:strongorder}.
  With a symmetrical argument, $G$ has an edge~$\{z, x_c\}$ for
  all~$z\in Z_l$.  It follows that the vertex sets $X'=Z_l\cup\{x_c\}$ and
  $Y'=Z_r\cup\{y_c\}$ both form stars in $G$.
   
  We now compare the widths of $G[X']$ and $G[Y']$ to the widths of
  the original stars~$G[X]$ and $G[Y]$. Let $w$ be the total number of
  elements between $z_1$ and $z_{2s}$, that is, the cardinality of the vertex
  set $\{u\mid z_1\strongordereq u\strongordereq
  z_{2s}\}=\closure(X)\cup\closure(Y)$.  Then, using the fact that the
  scopes of $X'$ and $Y'$ are disjoint and included in a size-$w$ set,
  we have
  \begin{align*}
    \width(X')+\width(Y')&\leq w \\
    &=|\closure(X)|  + |\closure(Y)| - |\closure(X) \cap \closure(Y)|\\
    &< \width(X) + \width (Y).
  \end{align*}
  We can thus construct an \starpartition{s}
  $\Psol'=(\Psol\setminus\{X,Y\})\cup\{X',Y'\}$ such that
  $\width(\Psol')<\width(\Psol)$, that is,  with strictly smaller
  score. Thus, no pair of stars in the minimum-score \starpartition{s} $\Psol$ may
  be in Configuration~I.
   
  \paragraph{Configuration~II.} Let $X,Y$ be two stars of
  $\Psol$ in Configuration~II, i.\,e., $\scenter(Y)\in\closure(X)$
  and $\scenter(X) \not\in \closure(Y)$.  Write $x_c=\scenter(X)$
  and~$y_c=\scenter(Y)$.  Then $y_c \strongorder \rightm(X)$ and
  either $x_c\myll \closure(Y)$ or $\closure(Y) \myll x_c$.  We only
  consider the case $x_c\myll \closure(Y)$;  the case $\closure(Y)
  \myll x_c$ works analogously.
  
  Let $v=\rightm(X)$ be the rightmost vertex of the leaves of the
  star~$G[X]$.  First, $G$~contains the edge~$\{x_c, y_c\}$ since the star~$G[Y]$
  has at least one leaf~$u$ with $x_c \strongorder u$
  and~$y_c\strongorder v$, and $G$ contains the edges~$\{x_c, v\}$ and
  $\{y_c, u\}$.  Now, consider any vertex~$u \in Y\setminus
  \{\scenter(Y)\}$.  Then, the edge~$\{x_c, v\}$ crosses the edge~$\{u,
  y_c\}$, since $x_c \strongorder u$ and $y_c\strongorder \rightm(X)$.
  The graph~$G$ contains the edges~$\{x_c, y_c\}$ and $\{v, u\}$.  Thus,
  the vertex sets~$X'=(X\setminus \{v\}) \cup \{y_c\}$ and $Y'=(Y
  \setminus \{y_c\}) \cup \{v\}$ both form stars in~$G$.

  We now compare the widths of $G[X']$ and $G[Y']$ to the widths of
  the original stars $G[X]$ and $G[Y]$.

  Since $y_c \strongorder v$, one has $\width(X')\le \width(X)-1$.  Obviously,
  $\width(Y) = \width(Y')$.  We can thus construct an
  \starpartition{s} $\Psol'=(\Psol\setminus \{X,Y\})\cup\{X',Y'\}$
 with $\width(\Psol')<\width(\Psol)$, that is,  with strictly smaller
  score. Therefore, no pair of stars in the \starpartition{s} $\Psol$ may
  be in Configuration~II.

  \paragraph{Configuration~III.} Let $X,Y$ be two stars of
  $\Psol$ in Configuration~III.  Let~$x_c\coloneqq \scenter(X)$ and
  $y_c\coloneqq \scenter(Y)$ and assume, without loss of generality,
  that~$x_c\strongorder y_c$. Then, $\closure(Y) \myll \closure(X)$.
  Thus, all edges of $G[X]$ cross all edges of $G[Y]$. Hence,
  there exists an edge $\{x_c,y\}$ for each leaf~$y$ of $G[Y]$, and an
  edge $\{y_c,x\}$ for each leaf~$x$ of $G[X]$.  Defining
  $X'=(X\setminus\{x_c\})\cup\{y_c\}$
  and~$Y'=(Y\setminus\{x_c\})\cup\{y_c\}$, we thus have two stars
  $G[X']$ and $G[Y']$ with the same width as $G[X]$ and $G[Y]$,
  respectively.  Hence, the \starpartition{s}
  $\Psol'=(\Psol\setminus\{X,Y\})\cup\{X',Y'\}$ has the same width as
  $\Psol$.
   
  We now show that $\edgecrossingnum(\Psol')<\edgecrossingnum(\Psol)$.
  We write $B_X$ (respectively $B_Y$, $B_{X'}$, and $B_{Y'}$) for the
  \emph{branches} of the corresponding star, that is, for the set of
  edges of $G[X]$ (respectively of $G[Y]$, $G[X']$, and $G[Y']$), and $R_{X,Y}$
  (respectively $R_{X',Y'}$) for the set of edges in $G[Z]$ for any $Z\in
  \Psol\setminus\{X,Y\}$ (respectively for any $Z\in \Psol'\setminus\{X',Y'\}$). Note
  that, by definition of $\Psol'$, $R_{X',Y'}=R_{X,Y}$.  We thus simply
  denote this set by $R$.  We write~$\cross_{b,b}$
  (respectively~$\cross_{b',b'}$) for the number of crossings between
  branches of~$B_X$ and~$B_Y$ (respectively of $B_{X'}$ and $B_{Y'}$),
  $\cross_{b,r}$ (respectively $\cross_{b',r}$) for the number of crossings
  between a branch in~$B_X\cup B_Y$ (respectively in $B_{X'}\cup B_{Y'}$) and an
  edge in $R$, and $\cross_{r,r}$ for the number of crossings between
  two edges in $R$.  Note that $\edgecrossingnum(\Psol) =
  \cross_{b,b}+\cross_{b,r}+\cross_{r,r}$ and that
  $\edgecrossingnum(\Psol') =
  \cross_{b',b'}+\cross_{b',r}+\cross_{r,r}$.
   
  It is easy to see that $\cross_{b',b'}=0$ ($X'$ and $Y'$ form
  non-crossing stars), and~$\cross_{b,b}>0$%
  .  Let~$x_i$ (respectively $y_i$) be the~$i$-th leaf
  of~$X$ (respectively of $Y$) in the order $\strongorder$.  Then, by
  \cref{prop:multiple-crossings}, any edge in $R$ crossing one or two
  edges among~$\{\{y_c, x_i\},\{x_c, y_i\}\}$ also crosses at least as
  many edges among~$\{\{x_c, x_i\},\{y_c, y_i\}\}$. Summing over all
  branches and all crossing edges, we obtain
  $\cross_{b',r}\leq\cross_{b,r}$.  Thus, overall, we indeed have
  $\edgecrossingnum(\Psol')<\edgecrossingnum(\Psol)$.
   
  Finally, we have constructed an \starpartition{s} with the same
  width but fewer crossings, that is, with strictly smaller score.
  Thus, no pair of stars in the~\starpartition{s} $\Psol$ may be in
  Configuration~III.

  \paragraph{Configuration~IV.}  Let $X,Y$ be two stars of
  $\Psol$ in Configuration~IV.  Without loss of generality, we assume
  that $\scenter(X)\myll\closure(Y)$ and $\scenter(Y)\myll\closure(X)$.
  We moreover assume that $X$ and $Y$ are chosen so that the number of
  elements between $\scenter(X)$ and $\rightm(Y)$, written
  $\dist(X,Y)$, is minimal among all pairs in Configuration~IV.  The
  configuration is depicted in more detail in \cref{fig:configIV}
  (left).

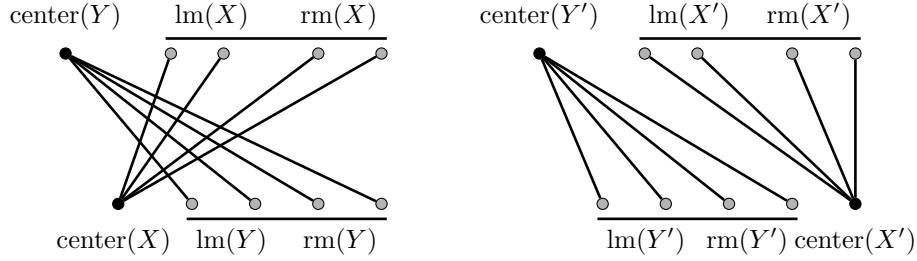
\begin{figure}[t!]
  \centering \def\sep{2} \def\ylabels{0.5}
  \begin{tikzpicture}[draw=black!100, xscale=0.7]
    \tikzstyle{newcenter}=[center] \tikzstyle{thickedge}=[draw=black,
    line width=1pt]
    \tikzstyle{oldedge}=[draw=black, line width=1pt]
    \tikzstyle{newedge}=[oldedge]

    \begin{scope}[xshift=-9cm]
      \node[center] (u-1) at (0,\sep) {}; \foreach \x in {2,3}
      \node[vertex] (u-\x) at (\x*1,\sep) {}; \foreach \x in {4,5}
      \node[vertex] (u-\x) at (\x*1.2,\sep) {};
    
      \node[] (c) at ([shift={(0,\ylabels)}]u-1) {$\scenter(Y)$};
      \node[anchor=west] (lm) at ([shift={(-0.1,\ylabels)}]u-2)
      {$\leftm(X)$}; \node[anchor=east] (rm) at
      ([shift={(0.1,\ylabels)}]u-5) {$\rightm(X)$}; \draw[oldedge]
      (lm.south west) -- (rm.south east);

      \node[center] (v-1) at (1,0) {}; \foreach \x in {2,3,4,5}
      \node[vertex] (v-\x) at (\x*1.2,0) {};
    
      \node[] (c) at ([shift={(0,-\ylabels)}]v-1) {$\scenter(X)\ $};
      \node[anchor=west] (lm) at ([shift={(-0.1,-\ylabels)}]v-2)
      {$\leftm(Y)$}; \node[anchor=east] (rm) at
      ([shift={(0.1,-\ylabels)}]v-5) {$\rightm(Y)$}; \draw[oldedge]
      (lm.north west) -- (rm.north east);

      \foreach \x in {2,3,4,5} \draw[oldedge] (u-1) -- (v-\x);
      \foreach \x in {2,3,4,5} \draw[oldedge] (u-\x) -- (v-1);
    \end{scope}

    \node[newcenter] (u-1) at (0,\sep) {}; \foreach \x in {2,3}
    \node[vertex] (u-\x) at (\x*1,\sep) {}; \foreach \x in {4,5}
    \node[vertex] (u-\x) at (\x*1.2,\sep) {};

    \node [] (cp) at ([shift={(0,\ylabels )}]u-1) {$\scenter(Y')$};
    \node[anchor=west] (lmp) at ([shift={(-0.1,\ylabels )}]u-2)
    {$\leftm(X')$}; \node[anchor=east] (rmp) at ([shift={(0.1,\ylabels
      )}]u-5) {$\rightm(X')$}; \draw[newedge] (lmp.south west)--
    (rmp.south east);
    \node[newcenter] (v-5) at (6,0) {}; \foreach \x in {1,2,3,4}
    \node[vertex] (v-\x) at (\x*1.2,0) {};
      
    \node [] (cp) at ([shift={(0,-\ylabels )}]v-5) {$\scenter(X')$};
    \node[anchor=west] (lmp) at ([shift={(-0.1,-\ylabels )}]v-1)
    {$\leftm(Y')$}; \node[anchor=east] (rmp) at
    ([shift={(0.1,-\ylabels )}]v-4) {$\rightm(Y')$}; \draw[newedge]
    (lmp.north west)-- (rmp.north east);
    \foreach \x in {1,2,3,4} \draw[newedge] (u-1) -- (v-\x); \foreach
    \x in {2,3,4,5} \draw[newedge] (u-\x) -- (v-5);

  \end{tikzpicture}
 \caption{Left: Two stars $X$ and $Y$ in Configuration~IV such that $\dist(X,Y)$ is minimal.
   Right: Two stars $X'$ and $Y'$ obtained from $X$ and $Y$, with
   equal width and fewer crossings.  }
 \label{fig:configIV}
\end{figure}

We first show that $\dist(X, Y)=s-1$, which means that no vertex
exists between $\scenter(X)$ and~$\rightm(Y)$, except for the $s-1$
other leaves of $Y$.  Suppose, towards a contradiction, that there is a
vertex~$z\notin \{\scenter(X)\}\cup \closure(Y)$ such that
$\scenter(X)\strongorder z \strongorder \rightm(Y)$.
   
Assume first that $z$ is the center of a star $G[X']$ with~$X'\in
\Psol$.  Then, $\closure(X) \myll \closure(X')$, since $X$ and $X'$
cannot be in Configuration~I or~III. Moreover, we have
$\scenter(X)\strongorder z \myll \closure(Y)$ since, otherwise, $z\in
\closure(Y)$ and $Y$ and $X'$ would be in Configuration~II.  Hence,
$X'$ and $Y$ are in configuration~IV (with
$\scenter(X')\myll\closure(Y)$, $\scenter(Y)\myll\closure(X')$), and
$\dist(X', Y) < \dist(X, Y)$, which is a contradiction.

Now assume that $z$ is a leaf of a star $G[Y']$ with $Y' \in \Psol$.
First compare $Y$ and $Y'$: $\closure(Y') \cap \closure(Y) =
\emptyset$ since, otherwise, $Y$ and $Y'$ would be in Configuration~I.
Using $z \strongorder \rightm(Y)$, it follows that $\closure(Y') \myll
\closure(Y)$.  This implies that $\scenter(Y')\strongorder
\scenter(Y)\myll \closure(X)$ since, otherwise, $Y'$ and $Y$ would be in
Configuration~III.  We now compare $X$ and $Y'$. We have already seen that
$\scenter(Y')\myll \closure(X)$.  Also, $\scenter(X) \notin
\closure(Y')$ since, otherwise, $Y'$ and~$X$ would be in Configuration~II.
Using $\scenter(X) \strongorder z$, we thus have $\scenter(X) \myll
\closure(Y')$, which implies that $X$ and $Y'$ are in Configuration~IV
with $\dist(X, Y') < \dist(X, Y)$, which is a contradiction.
We conclude that no vertex other than the leaves of $Y$ may exist between
$\scenter(X)$ and $\rightm(Y)$.

We now construct an \starpartition{s} with score strictly less than
$\mathcal P$.  To this end, let $X_0=X\setminus\{\scenter(X)\}$ and
$Y_0=Y\setminus\{\rightm(Y)\}$.  First observe that~$G$ contains the
edge~$\{\scenter(X), \scenter(Y)\}$ since there is an edge in $G[X]$
and an edge in~$G[Y]$ crossing each other.  Hence, $Y'=Y_0\cup
\{\scenter(X)\}$ forms a star.  Now, consider any vertex~$u\in X_0$.
The edge~$\{\scenter(X), u\}$ crosses the edge~$\{\scenter(Y),
\rightm(Y)\}$ and, therefore, $G$~contains the edge~$\{\rightm(Y),
u\}$.  Thus, $X'=X_0\cup \{\rightm(Y)\}$ forms a star.  For an illustration, see
\cref{fig:configIV} (right).  Also, $X'$ and $Y'$ are non-crossing
($X'$ is completely to the right of~$Y'$).
 
We now compare the widths of~$G[X']$ and~$G[Y']$ to the widths of the
original stars~$G[X]$ and $G[Y]$.  Obviously, $\width(X')=\width(X)$.
Moreover, since $\dist(X,Y)=s-1$, it follows that $\width(Y') =
\width(Y) =s$.  Hence, the \starpartition{s}
$\Psol'=(\Psol\setminus\{X,Y\})\cup\{X',Y'\}$ has the same width as
$\Psol$.

Since the widths have not changed, we have to show that
$\edgecrossingnum(\Psol')<\edgecrossingnum(\Psol)$.  We introduce the
same notations as in Configuration~III: Let $B_X$ (respectively $B_Y$,
$B_{X'}$, and $B_{Y'}$) be the set of branches of the corresponding star,
that is, the set of edges in $G[X]$ (respectively in $G[Y]$, $G[X']$, and~$G[Y']$),
and let $R_{X,Y}$ (respectively $R_{X',Y'}$) be the set of edges in~$G[Z]$
for any $Z\in \Psol\setminus\{X,Y\}$ (respectively for any $Z\in
\Psol'\setminus\{X',Y'\}$). Note that by definition of $\Psol'$,
$R_{X',Y'}=R_{X,Y}$, and we thus simply denote this set by~$R$.
Furthermore, let $\cross_{b,b}$ (respectively $\cross_{b',b'}$) be the number of
crossings between branches of $B_X$ and $B_Y$ (respectively between branches of $B_{X'}$ and
$B_{Y'}$), let $\cross_{b,r}$ (respectively $\cross_{b',r}$) be the number of
crossings between a branch in~$B_X\cup B_Y$ (respectively in $B_{X'}\cup
B_{Y'}$) and an edge in $R$, and let $\cross_{r,r}$ be the number of
crossings between two edges in $R$.  Note that
$\edgecrossingnum(\Psol) = \cross_{b,b}+\cross_{b,r}+\cross_{r,r}$ and
that $\edgecrossingnum(\Psol') =
\cross_{b',b'}+\cross_{b',r}+\cross_{r,r}$.  Then, it is easy to see that $\cross_{b',b'}=0$ ($X'$ and $Y'$ form
non-crossing stars), and $\cross_{b,b}>0$.

We now show that $\cross_{b',r}\leq \cross_{b,r}$.  First recall
that no edge in $R$~has an end point between $\scenter(X)$ and
$\rightm(Y)$.  We consider the branches in $B_{X'}\cup B_{Y'}$ and,
for each, give a unique edge in $B_{X}\cup B_{Y}$ crossing the same
edges of $R$.  For any leaf $x$ of $X'$, any $r\in R$ crossing
$\{\scenter(X'),x\}$ must also cross $\{\scenter(X),x\}$.  For the
leftmost branch of $Y'$, any~$r\in R$ crossing
$\{\scenter(X),\scenter(Y)\}$ must also cross
$\{\rightm(Y),\scenter(Y)\}$.  For any other branch
$b=\{\scenter(Y),y\}$ of $Y'$, any~$r\in R$ must also cross the same
branch $b$ of $Y'$.  Overall, we indeed have~$\cross_{b',r}\leq
\cross_{b,r}$, which implies
$\edgecrossingnum(\Psol')<\edgecrossingnum(\Psol)$.
   
Altogether, we have shown that a minimum-score \starpartition{s} $\Psol{}$
containing pairs of stars in Configurations~I to~IV leads to a
contradiction, since, in this case, we could find an \starpartition{s}
of lower score, which is a contradiction.
\end{proof}

\noindent As a consequence of \cref{lem:only_noncrossing+interleaving}, we
obtain the following corollary.

\begin{corollary}
  \label[corollary]{cor:at-most-one-crossing}
  Let $\Psol$ be an \starpartition{s} of a bipartite permutation graph
  $G$ with minimum score.  Then, for each star~$X \in \Psol$, there is
  at most one~$Y \in \Psol$ such that $X$ and $Y$ are interleaving,
  and for all $Z\in \Psol\setminus\{X,Y\}$, $X$ and $Z$ are
  non-crossing.
\end{corollary}

\begin{proof}
  Since $\Psol$ has minimum score, for any $Y\in \Psol\setminus\{X\}$,
  $G[X]$ and $G[Y]$ are either interleaving or non-crossing.
 
  Any star interleaving with $G[X]$ contains $\scenter(x)$ in its
  scope.  If there exist at least two such stars in $\Psol$, then
  their scopes intersect and they are in Configuration~I, which is
  impossible by \cref{lem:only_noncrossing+interleaving}.
\end{proof}

\noindent We now informally describe a dynamic programming algorithm for
deciding whether a bipartite graph~$G=(U,W,E)$ allows for a
\emph{good} \starpartition{s}.  It builds up a solution following the
strong ordering of the graph from left to right. A partial solution
can be extended in three ways only: either (i) a star is added with
the center in $U$, or (ii) a star is added with the center in $W$, or
(iii) two interleaving stars are added. The algorithm can thus
compute, for any given number of centers in $U$ and in $W$, whether it
is possible to partition the leftmost vertices of~$U$ and~$W$ in one
of the three ways (i)--(iii).  This algorithm leads to the following
result.

\begin{theorem} \label{thm:bipPerm} \probStarPart can be solved in
  $O(n^2)$~time on bipartite permutation
  graphs. %
\end{theorem}

\begin{proof}
  Let $(G,s)$ denote a \probStarPart instance, where $G=(U,W,E)$ is a
  bipartite permutation graph.  Furthermore, let
  $U=\{u_1,u_2,\dots,u_{k_U}\}$ and~$W=\{w_1,w_2,\dots,w_{k_W}\}$ such
  that $u_i \strongorder u_j$ (respectively\ $w_i \strongorder w_j$) implies
  $i<j$ for some fixed strong ordering $\strongorder$.  We describe a
  dynamic programming algorithm that finds a good
  \starpartition{s}~$\Psol$.  The idea is to use the fact that a star
  from~$\Psol$ is either interleaving with exactly one other star
  from~$\Psol$ or it does not cross any other star from~$\Psol$
  (see \cref{lem:only_noncrossing+interleaving} and
  \cref{cor:at-most-one-crossing}).  In both cases, the part of the
  graph that lies entirely to the left of the star (of the two
  interleaving stars respectively) with respect to the strong ordering must have an
  \starpartition{s} on its own.  This is clearly also true for the
  part of the graph that lies entirely to the right, but we do not need
  this for the proof.
  
  \looseness=-1 Informally, an entry $T(x,y)$ of our binary dynamic programming
  table~$T$ is true if and only if $x$~stars with centers from~$U$ and
  $y$~stars with centers from~$W$
  can ``consecutively cover'' the correspondingly large part of the
  graph from the left side of the strong ordering.  Formally, the
  binary dynamic programming table~$T$ is defined as
  \[T(x,y)=
  \begin{cases}
    1&\text{if $G[\{u_1,u_2,\dots,u_{x+s \cdot y},w_1,w_2,\dots,w_{y+s \cdot x}\}]$}\\
    &\text{\quad has an \starpartition{s},}\\
    0&\text{otherwise.}
  \end{cases}
  \]
  Initialize the table~$T$ by:\allowdisplaybreaks
  \begin{align*}
    T(0,1)&=
    \begin{cases}
      1&\text{if $G[\{u_1,u_2,\dots,u_{s},w_1\}]$ contains an
        \sstar},\\
      0&\text{otherwise,}
    \end{cases}
\\
    T(1,0)&=
    \begin{cases}
      1&\text{if $G[\{u_1,w_1,w_2,\dots,w_{s}\}]$ contains an \sstar},\\
      0&\text{otherwise, and}
    \end{cases}
\\
    T(1,1)&=
    \begin{cases}
      1&\text{if $G[\{u_1,u_2,\dots,u_{s+1},w_1,w_2,\dots,w_{s+1}\}]$}\\
      &\text{\quad contains disjoint \sstars},\\
      0&\text{otherwise.}
    \end{cases}
  \end{align*}
  Update the table~$T$ for all $1<x\le k_U$ and $1<y\le k_W$ by
  \begin{align*}
    T&(x,y)=\\&
    \begin{cases}
      1&
      \begin{minipage}[t]{0.9\linewidth}if one of the
          following holds:
        \begin{enumerate}[(a)]
        \item $T(x,y-1)=1$ and $G[\{u_{x+s \cdot (y-1)+1}, u_{x+s
            \cdot (y-1)+2}, \dots, \allowbreak u_{x+s \cdot (y-1)+s},\allowbreak w_{(y-1)+s \cdot x+1}\}]$ contains an \sstar,
    \item 
      $T(x-1,y)=1$ and $G[\{u_{(x-1)+s \cdot y+1},\allowbreak w_{y+s \cdot
         (x-1)+1},\allowbreak w_{y+s \cdot (x-1)+2}, \allowbreak \dots,\allowbreak w_{y+s \cdot
         (x-1)+s}\}]$ contains an \sstar, 
     \item $T(x-1,y-1)=1$ and $G[\{u_{(x-1)+s \cdot (y-1)+1},\allowbreak %
       \dots, \allowbreak u_{(x-1)+s \cdot (y-1)+s+1},\allowbreak w_{(y-1)+s
         \cdot (x-1)+1}, %
       \dots, w_{(y-1)+s \cdot
          (x-1)+s+1}\}]$ contains disjoint \sstars.
        \end{enumerate}
    \end{minipage}
\\
      0&\text{otherwise.}
    \end{cases}
  \end{align*}
  Concerning the running time, 
  first, a strong ordering of the vertices can be computed in linear time~\cite{SBS87}.
  Second, the table in the dynamic program has $O(k^2)$~entries and
  initialization as well as updating works in $O(s^2)$ time. Hence,
  the total running time is $O(k^2\cdot s^2)=O(n^2)$.

  Concerning the correctness of the algorithm, we show that
  $T(k_U,k_V)$ is true if and only if there is a good
  \starpartition{s} and, hence, if and only if there is an
  \starpartition{s}. To this end, consider an
  \starpartition{s}~$\Psol'$ for $G'\coloneqq G[\{u_1,u_2,\dots,u_{x+s \cdot
    y}$, $w_1,w_2,\dots,w_{y+s \cdot x}\}]$ with minimum score.  Now
  there are three simple cases:

  \paragraph{Case~(a).} The rightmost vertex of $G'$ in~$W$ is a
  center of a non-crossing star in~$\Psol'$ and, hence,
  $G[\{u_1,u_2,\dots,u_{x+ s\cdot (y-1)},w_1,w_2,\dots,w_{y-1+s \cdot
    x}\}]$ has an \starpartition{s}.

  \paragraph{Case~(b).} The rightmost vertex of $G'$ in~$U$ is a
  center of a non-crossing star in~$\Psol'$ and, hence,
  $G[\{u_1,u_2,\dots,u_{x-1+s \cdot y},w_1,w_2,\dots,w_{y+s \cdot
    (x-1)}\}]$ has an \starpartition{s}.

  \paragraph{Case~(c).} The rightmost vertex of $G'$ in~$U$ and $G'$s
  rightmost vertex in~$W$ are leaves of two interleaving stars.  Due
  to \cref{cor:at-most-one-crossing}, none of the other stars
  from~$\Psol'$ is crossing these two stars. It follows that
  $G[\{u_1,u_2,\dots,u_{x-1+s \cdot (y-1)},\allowbreak
  w_1,w_2,\dots,\allowbreak w_{y-1+s \cdot (x-1}\}]$ has an \starpartition{s}.

  Note that the rightmost vertex of $G'$ in~$U$ can only be a leaf of
  a non-crossing star in~$\Psol'$ if the rightmost vertex of $G'$
  in~$W$ is the center and vice versa.  Otherwise, the corresponding
  star is clearly not non-crossing.  Hence, these cases are already
  covered by~$(a)$ and~$(b)$.  Furthermore, neither the rightmost
  vertex of $G'$ in~$U$ nor in~$W$ can be a center of an interleaving
  star from~$\Psol'$, because both are rightmost with respect to the
  strong ordering and, thus, interleaving is impossible.  Thus we
  considered all cases and the update process is correct.  
\end{proof}

\section{Split graphs}\label{sec:splitgraphs}
A \emph{split graph} is a graph whose vertices can be partitioned into
a clique (that is, a complete subgraph) and an independent set (that is, a subgraph with only isolated vertices). 
Remarkably, split graphs are the only
graph class where we could show that \probPthreePart{} is solvable in
polynomial time, but that \probStarPart{} for $s\geq 3$ is NP-hard.

More precisely, we solve \probPthreePart{} on split graphs by reducing
it to finding a restricted form of \emph{factor} in an auxiliary
graph; herein, a \emph{factor} of a graph $G$ is a spanning subgraph
of~$G$ (that is, a subgraph containing all vertices). This graph
factor problem then can be solved in polynomial time \cite{Corn88}.
Alternatively, we can also solve the problem by reducing it to finding perfect matchings~(\cref{thm:splittract}).

Let $G = (C \cup I, E_C\cup E )$ be a split graph where $(C,E_C)$ is
a clique, $I$ induces an independent set, and $B = (C \cup I, E)$
forms a bipartite graph over~$C$ and~$I$.  Note that if $|C|+|I|$ is
not a multiple of $3$, or if $|I|>2|C|$, then $G$ trivially has no
$P_3$-partition.  We thus assume that $|C|+|I|$ (and hence $2|C|-|I|$)
is a multiple of $3$, and that $2|C|-|I| \ge 0$.

First, we show how a $P_3$-partition of a split graph is related to a
specific factor of the bipartite graph $B$.  Assume that~$G$ admits a
partition into $P_3$s and let $P$ denote the set of edges in the
partition.  There are three types of $P_3$s:
\begin{enumerate}[(i)]
\item\label{3clique-vertices} a $P_3$ consisting of three clique
  vertices,
\item\label{2clique-vertices+1independent-vertex} a $P_3$ consisting
  of two clique vertices and one independent set vertex, and
\item\label{1clique-vertex+2independent-vertices} a $P_3$ consisting
  of one clique vertex and two independent set vertices.
\end{enumerate}
Note that, for each~$P_3$, we can assume that the edges are selected so
that each independent set vertex is incident with \emph{at most} one edge in~$P$.
In particular, this implies that the two clique vertices in Type~(\ref{2clique-vertices+1independent-vertex}) are adjacent in the corresponding~$P_3$.
This leads to the following definition.

\begin{definition}\label{def:feasible factor}
\looseness=-1  A factor~$F$ of the bipartite graph~$B$ is \emph{feasible} if, in
  $F$, every independent set vertex has degree one, every clique vertex
  has degree zero, one or two, and 
  there are at least as many degree-zero
  clique vertices as degree-one clique vertices.
\end{definition}

It turns out that \cref{def:feasible factor} is necessary and sufficient for obtaining a $P_3$-partition.

\begin{lemma}\label[lemma]{lem:P3-partition==factor+contraint}
  A split graph $G$ admits a partition into $P_3$s if and only if there exists a
  feasible factor of its bipartite graph $B$.
\end{lemma}

\begin{proof}
  Assume that there is a partition of~$G$ into $P_3$s with edge set~$P$ 
  such that each vertex in~$I$ is incident with exactly one edge
  in $P$.  Let~$P_E \coloneqq P \cap E$ be the subset of edges of the
  partition which connect vertices from $C$ with vertices from~$I$,
  and let $F\coloneqq(C \cup I, P_E)$ be the corresponding factor of $B=(C\cup I, E)$.
  
  \looseness=-1 Each independent set vertex in~$I$ has degree one in~$F$ (since it is adjacent to
  exactly one edge in~$P$, which is also in $P_E$).  Each clique vertex~$v$
  in~$C$ belongs to a $P_3$ from $P$.  
  Depending on the type of this
  $P_3$, in $F$, vertex~$v$ can have degree zero (Type~(\ref{3clique-vertices}) or
  (\ref{2clique-vertices+1independent-vertex}), note that we assume each independent set vertex to be incident with at most one edge in $P$), 
  degree one~(Type~(\ref{2clique-vertices+1independent-vertex})),
  or degree two~(Type~(\ref{1clique-vertex+2independent-vertices})).
  Let~$n_{(\text{\ref{3clique-vertices}})}$
  (respectively $n_{(\text{\ref{2clique-vertices+1independent-vertex})}}$) denote
  the number of~$P_3$s of Type~(\ref{3clique-vertices})
  (respectively (\ref{2clique-vertices+1independent-vertex})).  
  It remains to show that 
  the number of degree-zero clique vertices is equal to or greater than 
  the number of degree-one clique vertices:
  The number of degree-zero vertices
  is~$n_{(\text{\ref{2clique-vertices+1independent-vertex}})}
  +3n_{(\ref{3clique-vertices})}$, and the number of degree-one
  vertices is $n_{(\text{\ref{2clique-vertices+1independent-vertex}})}$,
  hence the difference is positive.  Thus, $F$ is feasible.
  
  Conversely, let $F=(C\cup I, P_E)$ be a feasible factor of $B$.
  Then we partition~$G$ into~$P_3s$ as follows. 
  For each degree-two clique vertex~$v$ in $C$, 
  add $\{v,x,y\}$ to $P$ where $x$ and $y$ are
  the neighbors of $v$ in $F$.
  For each degree-one clique vertex~$v$ in~$C$, 
  add~$\{v,x,y\}$ to~$P$, 
  where $x$~is $v$'s neighbor in~$F$ 
  and~$y$~is an arbitrary degree-zero clique vertex (there are enough such vertices).  
  The number of remaining degree-zero vertices in~$C$ is thus a multiple of~$3$: 
  these vertices are simply grouped up in arbitrary triples
  Add these triples to $P$. 
  Overall, due to the degree constraints, $P$ is a~$P_3$-partition of~$G$. 
\end{proof}

\noindent 

\citet{Corn88} shows that finding a feasible factor in $B$ can be solved in polynomial time 
by reducing it to finding disjoint edges and triangles in a corresponding auxiliary graph.
Nevertheless, we show in the following how to reduce the problem to finding perfect matchings.
To this end, we formulate a nice property that a feasible factor in $B$ must fulfill:

\begin{property}\label{prop:feasible_factor+no_zero>=no_one->degree_constraints}
  Let $F=(C\cup I, P\cap E)$ be a feasible factor of the bipartite graph~$B=(C\cup I, E)$.  
  Let $q$ and $r$ be two non-negative integers such that $r\in \{0,1\}$ and $(2|C|-|I|)/3=2q+r$.
  Then, in $F$, the number~$n_1$ of degree-one vertices in~$C$ is~$2i+r$ for some~$i$, $0\leq i
  \leq q$. In particular, $n_1\le (2|C|-|I|)/3$.
\end{property}

\begin{proof}
  Let $n_0$, $n_1$, and $n_2$ be the number of degree-zero,
  degree-one, and degree-two clique vertices in~$F$.  
  Then,
  $n_0+n_1+n_2=|C|$ (all clique vertices have degree~0, 1 or 2),
  $n_1+2n_2=|I|$ (vertices in~$I$ have degree 1).
  Rearranging and resolving variable~$n_2$ yields
  \begin{align}\label{eq:n1formula}
    3n_1 = 2|C|-|I|- 2(n_0-n_1)\text{.}
  \end{align}
  As mentioned, $2|C|-|I|$ is a multiple of three (as well as $|C|+|I|$).
  Note that $n_0-n_1$, which is positive because $F$ is feasible,
  is a multiple of three because it equals 
  the number of clique vertices in $P_3$s of Type~\eqref{3clique-vertices}.
  Let $j$ be an integer with $(n_0-n_1)/3=j$. 
  Then, together with \eqref{eq:n1formula}, we obtain that
  \begin{align*}
    \nonumber n_1 &= \frac{2|C|-|I|}{3} - 2\cdot \frac{n_0-n_1}{3}\\
    &= 2q+r - 2j\text{.}
  \end{align*} 
  The last statement is satisfied since $2q+r=(2|C|-|I|)/3$ and $j\ge 0$.
\end{proof}

To be able use a perfect matching algorithm to solve our problem,
we first reduce it to a restricted variant of the graph factor problem: 
We add an additional vertex~$z$ to the bipartite graph~$B$,
and connect it to all vertices in~$C$. We call this graph~$B'$.  Now, the
following lemma states that~$B'$ can be used to find a
feasible factor for~$B$. %

\begin{lemma}\label[lemma]{lem:B=B'}
  The bipartite graph $B=(C\cup I, E)$ admits a feasible factor if and only if
  graph~$B'=(C\cup I \cup \{z\}, E\cup \{\{z,c\}\mid c\in C\})$ 
  has a factor satisfying the following degree constraints:
  \begin{inparaenum}[(1)]
  \item Every vertex in~$I$ has degree one,
  \item every vertex in~$C$ has degree zero or two, and
  \item the added vertex~$z$ has degree~$2i+r$,
  where $r\in \{0,1\}$ such that there is an integer~$q$
  with $(2|C|-|I|)/3=2q+r$
  and $i \in \{0,1,\ldots,q\}$.
\end{inparaenum}
\end{lemma}

\begin{proof}
  Assume that $B$ admits a feasible factor~$F=(C\cup I, P_E)$.  Then
  $F'=(C\cup I\cup\{z\}, P_E)$ is a factor of $B'$. For each
  degree-one vertex~$v\in C$, we add edge~$\{v, z\}$ to factor~$F'$.
  By \cref{prop:feasible_factor+no_zero>=no_one->degree_constraints},
  we thus add $2i+r$ edges, with $0\leq i \leq q$.  It is easy to
  verify that the degree constraints stated in the lemma are
  satisfied.

  Conversely, let $F'$ be a factor for graph~$B'$ where every
  independent set vertex has degree one, every clique vertex has
  degree zero or two, and vertex~$z$ has degree~$2i+r$ with $i \in \{0,1,\ldots,q\}$.
  If we delete from $F'$ all edges incident to vertex~$z$,
  then we obtain a feasible factor~$F$ for $B$ where the number $n_1$
  of degree-one clique vertices is the original degree of $z$,
  that is, 
  \begin{align}
    \label{eq:n1cond}
    n_1=2i+r\le (2|C|-|I|)/3\text{.}
  \end{align}

  \looseness=-1 Since each independent set vertex still has degree one, the number
  of degree-two clique vertices is $(|I|-n_1)/2$, and the number of
  degree-zero clique vertices is
  $n_0=|C|-n_1-(|I|-n_1)/2=(2|C|-|I|-n_1)/2$. The difference between
  the number of degree-zero and the number of degree-one vertices is
  $n_0-n_1=(2|C|-|I|-3n_1)/2$, which is non-negative by using~\eqref{eq:n1cond}.
\end{proof}

\noindent\looseness=-1 \cref{fig:splitgraph-factor} (Left) depicts an example factor
for the graph~$B'$ fulfilling the degree constraints of
\cref{lem:B=B'}.
  
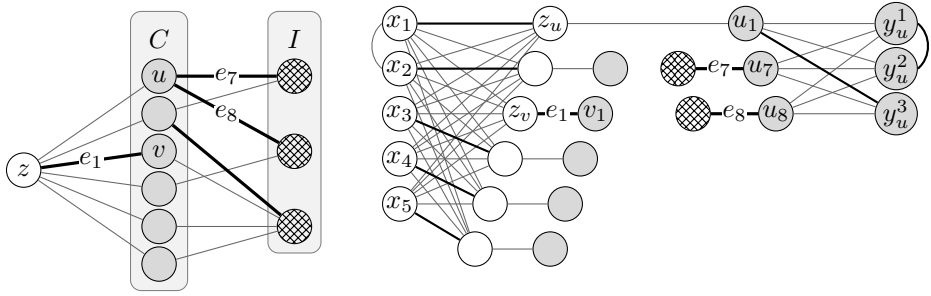
\begin{figure}[t]
  \pgfdeclarepatternformonly{soft crosshatch}{\pgfqpoint{-1pt}{-1pt}}{\pgfqpoint{4pt}{4pt}}{\pgfqpoint{3pt}{3pt}}%
  {
    \pgfsetstrokeopacity{0.3}
    \pgfsetlinewidth{0.4pt}
    \pgfpathmoveto{\pgfqpoint{3.1pt}{0pt}}
    \pgfpathlineto{\pgfqpoint{0pt}{3.1pt}}
    \pgfpathmoveto{\pgfqpoint{0pt}{0pt}}
    \pgfpathlineto{\pgfqpoint{3.1pt}{3.1pt}}
    \pgfusepath{stroke}
  }
\centering \def\layersep{1.8}
  \def\levelsep{.5} 
  \tikzstyle{labeledvertex}=[vertex, inner sep=.5pt, minimum size=13pt,fill=black!15]
  \tikzstyle{labeledindvertex}=[vertex, inner sep=.5pt, minimum size=13pt,fill=black!15,pattern=soft crosshatch]
  \tikzstyle{zvertex}=[vertex, inner sep=.5pt, minimum size=13pt,fill=white]
  \tikzstyle{thickedge}=[draw=black, line width=1.2pt] 
  \tikzstyle{halfthickedge}=[draw=black, line width=.8pt] 
  \begin{tikzpicture}[draw=black!60, node distance=\layersep]

    \begin{scope}[xshift=-7.6cm]
    \node[zvertex] (z) at (0-\layersep,-3.5*\levelsep)
    {$z$};

    \node (c-0) at (0,0) {$C$}; 
    \foreach \y/\l in {1/u,2/{},3/{v},4/{},5/{},6/{}} {
      \node[labeledvertex] (c-\y) at (0,-\y*\levelsep) {$\l$};
    }

    \node (i-0) at (\layersep,0) {$I$}; 
    \foreach \name / \l in {1/{},2/{},3/{}}
    {
      \node[labeledindvertex] (i-\name) at (\layersep,-\name*\levelsep*2+\levelsep) {$\l$};
    }

    \foreach \c in {1,...,6} \draw (z) -- (c-\c);

    \draw (c-2) -- (i-1); 
    \draw (c-4) -- (i-2); 
    \foreach \c in {3,5,6} \draw (c-\c) -- (i-3); 

    \draw[thickedge] (c-1) -- node[midway,fill=white,inner sep=0pt] {$e_7$} (i-1);
    \draw[thickedge] (c-1) -- node[midway,fill=white,inner sep=0pt] {$e_8$} (i-2); 
    \draw[thickedge] (c-2) -- (i-3);
     
   \draw[thickedge] (z) -- node[midway,fill=white,inner sep=0pt] {$e_{1}$} (c-3);
    \begin{pgfonlayer}{background}
      \node (C) [draw=black!50,rounded corners,fill=gray!10,fit=(c-0) (c-6)] {}; 
      \node (I) [draw=black!50,rounded
      corners,fill=gray!10,fit=(i-0) (i-3)] {};
    \end{pgfonlayer}
    \end{scope}

    \def\layersep{2.2}
    \def\levelsep{.6} 
    \def\baselevel{.8}
    \foreach \l/\name in {1/{z_u},2/{},3/{z_v},4/{},5/{},6/{}}
    {
      \node[zvertex] (z\l) at (0-\layersep-\l*0.2,-\l*\levelsep+\baselevel)
      {$\name$};
    }

    \foreach \l/\pos/\name in {1/2/{},2/3/{v_1},3/4/{},4/5/{},5/6/{}}
    {
      \node[zvertex] (a\l) at (0-2*\layersep,-\l*\levelsep+\baselevel)
      {$x_{\l}$};
      \node[labeledvertex] (e\pos) at (0-\layersep/2-\pos*.2-.1,-\pos*\levelsep+\baselevel)
      {$\name$};
      \draw (z\pos) -- (e\pos);

      \foreach \j in {1,...,6} 
      \draw (z\j) -- (a\l);
    }

    \draw (a1.west) edge[bend right=50] (a2.west);

    \foreach \l/\name in {1/1,2/7,3/8} {
      \node[labeledvertex] (u\l) at (\l*0.2,-\l*\levelsep+\baselevel)
      {$u_{\name}$};
    }
    \foreach \l in {1,...,3} {
      \node[labeledvertex] (b\l) at (\layersep,-\l*\levelsep+\baselevel)
      {$y^{\l}_{u}$};
      \foreach \j in {1,...,3}{
         \draw (u\j) -- (b\l);
      }
    }

    \foreach \i/\j in {2/7,3/8} {
      \node[labeledindvertex] (ee\i) at (\i*.2-\layersep/2,-\i*\levelsep+\baselevel) {};
      \draw[thickedge] (ee\i) -- node[midway,fill=white,inner sep=0pt] {$e_{\j}$} (u\i); 
    }

    \draw (b1.east) edge[bend left=50] (b2.east);

    \draw[] (z1) -- (u1);
    \draw[thickedge] (z3) -- node[midway,fill=white,inner sep=0pt] {$e_1$} (e3);

    \draw[halfthickedge] (u1) -- (b3);
    \draw[halfthickedge] (b1.east) edge[bend left=50] (b2.east);
    
    \foreach \i/\j in {1/1,2/2,3/4,4/5,5/6} {
      \draw[halfthickedge] (a\i) -- (z\j);
    }
  \end{tikzpicture}
  \caption{Left: An example of a factor for $B'$ fulfilling the degree
    constraints as required in \cref{lem:B=B'}. 
    The black thick bold edges reflects the constraints.
    Right: The gadget for vertex~$z$ and the gadget for the clique vertex~$u$ used to construct graph~$B^*$ (according to \cref{cons:graph_for_matching}). 
    The black thick bold edges (labeled) correspond to the ones marked on the right.
  The black bold edges are additional matching edges.}
  \label{fig:splitgraph-factor}
\end{figure}

Using a gadget introduced by \citet{Corn88}, we can even reduce \probPthreePart{}
to the perfect matching problem. We construct a graph~$B^*$ from $B'$ in which we are searching
for a perfect matching.  The idea is to replace every clique
vertex~$v$ by a gadget which can simulate the constraint that $v$ has
degree zero or two, and to replace vertex~$z$ by a gadget
to simulate its degree constraint.

\begin{construction}\label{cons:graph_for_matching}
  \normalfont
  Let $m=|E|$ (number of edges between $C$ and $I$ in the original
  split graph) and $d_u$ be the degree of a clique vertex~$u$ in $C$.
  Note that due to the edges going to vertex~$z$, we have $\sum_{u\in
    C} d_u=m+|C|$. It holds that $|I|=O(m)$ and $|C|=O(m)$.
  
  To construct the vertex set of $B^*$,
  first add a copy of $I$ to $B^*$.
  For each clique vertex~$u\in C$,
  add $d_u$ vertices~$y^{1}_{u},y^{2}_{u},\ldots,y^{{d_u}}_u$ to $B^*$;
  denote this set as~$Y(u)$.
  For each edge~$e_j\in E$ that is incident with $u$,
  add a vertex~$u_j$ to $B^*$; denote this set as $V(u)$.
  Note that $|Y(u)|=|V(u)|$ ($Y(u)$ and $V(u)$ are used to form a complete bipartite subgraph).
  For each edge~$\{z,u\}$ in $B'$, add a vertex~$z_{u}$ to $B^*$; denote this set as $V(z)$.
  Finally, add a set~$X$ of $|C|-r$ copies of vertex~$z$, named as $x_1,x_2,\ldots, x_{|C|-r}$ to $B^*$ ($V(z)$ and $X$ are used to form a complete bipartite subgraph).
  
  This completes the construction for the vertex set, 
  which consists of $|I|+\sum_{u\in C}2d_u+2|C|-r=O(m)$ vertices.

  Now we are ready to add edges to~$B^*$.
  For each edge~$e_j=\{u,v\}$ in $B'$ which connects 
  a clique vertex~$u$ and an independent set vertex~$v$,
  add an edge~$\{u_j, v\}$ to~$B^*$.
  Analogously, for each edge~$e_j=\{u,z\}$ in $B'$ which connects
  $z$ and a clique vertex~$u$, 
  add an edge~$\{u_j,z_u\}$ to~$B^*$.
  This is used to model the original edges of $B'$.
  Now, to model the degree constraints,
  for each clique vertex~$u\in C$,
  add to~$B^*$ an edge between every vertex from $V(u)$ and every vertex from $Y(u)$,
  and an edge between $y^{1}_u$ and $y^{2}_u$.
  Add to~$B^*$ an edge between every vertex from $V(z)$ and every vertex from $X$.
  Finally, for each integer~$i$ with $1\le i \le q$,
  add to~$B^*$ edge~$\{x_{2i-1},x_{2i}\}$.

  The overall number of edges in~$B^*$ is $(m+|C|)+\sum_{u\in
    C}(d_{u})^2+|C|(|C|-r)+ |C|+q\le (m+|C|)^2 + |C|^2 +O(m)=O(m^2)$.
  This finishes the construction.
\end{construction}

\cref{fig:splitgraph-factor} (Right), shows the gadget constructed for a clique vertex and the gadget for vertex~$z$. Now, we show how the constructed graph~$B^*$ can be used to find a
factor for~$B'$ satisfying the specific degree constraints as specified in
\cref{lem:B=B'}.

\begin{lemma}\label[lemma]{lem:B'has_a_factor==B*has_a_matching}
  Graph~$B^*$ constructed according to \cref{cons:graph_for_matching} admits a perfect matching if and only if graph~$B'$ admits a factor $F'$ satisfying the condition that  \begin{inparaenum}[(1)]
  \item Every vertex in~$I$ has degree one,
  \item every vertex in~$C$ has degree zero or two, and
  \item the added vertex~$z$ has degree~$2i+r$,
  where $r\in \{0,1\}$ such that there is an integer~$q$
  with $(2|C|-|I|)/3=2q+r$
  and $i \in \{0,1,\ldots,q\}$.
  \end{inparaenum}
\end{lemma}

\begin{proof}
  Let $M$ be a perfect matching for $B^*$.  We construct a
  factor~$F'=(C\cup I \cup\{z\}, P_E')$ for $B'$.  For each
  edge~$\{u_j,v\} \in M$ which connects an independent set vertex~$v$,
  add to $P_E'$ edge~$\{u,v\}$.  
  For each edge~$\{z_u, u_j\}\in M$ which connects vertices in~$V(z)$ and~$V(u)$,
  $u\in C$, add to $P_E'$ edge~$\{z,u\}$.

  We show that $F'$ is a factor for $B^*$ satisfying the properties stated in the lemma.
  Obviously, every independent set vertex~$u\in I$ has degree one.  
  Consider a clique vertex~$u\in C$.  
  By the construction of graph~$B^*$, 
  in order to match all vertices in $Y(u)$, either (i) every vertex in $Y(u)$ has
  to be matched to a vertex in $V(u)$ or (ii) $y^{1}_u$ and $y^2_{u}$ are
  matched together while every vertex in $Y(u)\setminus \{y^{1}_u, y^2_{u}\}$
  is matched to exactly one vertex in $V(u)$.  This implies that either
  no vertex or exactly two vertices in $V(u)$ are matched to some
  vertices which are not from~$Y(u)$.  Thus, $v$ has either degree
  zero or degree two in $F'$.

  Analogously, by the construction of graph~$B^*$, in order to match
  all vertices in $X$ which has size $|C|-r$, exactly $i$ pairs of
  vertices in $X$ can be left without being matched to any vertex
  in $V(z)$ where $0\le i \le q$ (note that only the first $2q$ vertices are connected by a path). 
  This implies that exactly $|C|-(|C|-r-2i) = 2i+r$ vertices from $V(z)$ are matched to vertices
  that are not from~$X$.  Thus, $z$ has degree $2i+r$. 

  Conversely, assume that $B'$ admits a factor~$F'$ satisfying the above
  properties.  We show that the following construction yields a perfect
  matching $M$ for $B^*$.
  
  For each edge~$e_j=\{u,v\}$ in~$F'$ that connects a clique vertex~$u$ and an independent set vertex~$v$, 
  add to $M$ edge~$\{u_j, v\}$.
  For each edge~$e_j=\{z,u\}$ in~$F'$ that connects vertex~$z$ with a
  clique vertex~$u$, add to $M$ edge~$\{z_u,u_j\}$.
  
  For each clique vertex~$u\in C$, let $R(u)\subseteq V(u)$ be the set
  of vertices which are not yet matched by~$M$.  We need to match
  all vertices in $R(u)$. %
  Depending on whether $u$ has degree zero or two in~$F'$, 
  $|V(u)|-|R(u)|$ is either zero or two, and
  $|Y(u)|-|R(u)|=|V(u)|-|R(u)|$.  
  Moreover, $B^*[V(u)\cup Y(u)|$ contains a complete bipartite graph for $V(u)$ and $Y(u)$.
  If $|R(u)|=|V(u)|$ (which means that $u$ has degree zero in $F'$), 
  then add to $M$ edges connecting exactly one vertex of $R(u)$ and one vertex of $Y(u)$;
  otherwise, $|R(u)|=|V(u)|-2$: add edge~$\{y^{1}_u,y^{2}_u\}$ to~$M$, 
  and edges connecting exactly one vertex of $R(u)$ and one vertex of $Y(u)\setminus \{y^{1}_u,y^{2}_u\}$.
  Analogously, let
  $R(z)\subseteq V(z)$ be the set of vertices in $V(z)$ which are not yet
  matched by $M$. By assumption, 
  there is an integer~$i$, $0\le i \le q$ such that $|V(z)|-|R(z)|=2i+r$.
  Since $|X|=|V(z)|-r$, we have $|X|-|R(z)|=2i$.  
  Moreover, $B^*[V(z)\cup X|$ contains a complete bipartite graph for $V(z)$ and $X$.
  Thus, for each $1\le k \le i$, add edge $\{x_{2k-1},x_{2k}\}$ to $M$, and
  add edges connecting exactly one vertex of $R(z)$ and one vertex of $X$
  to match vertices of $R_z$ and of~$V'_z\setminus \{x_{1},x_{2}, \ldots, x_{2i}\}$.
  It is easy to verify that $M$ is indeed a
  perfect matching.
\end{proof}

\noindent We now have gathered all ingredients to show \cref{thm:splittract}.

\begin{theorem}\label{thm:splittract}
  \probStarPart{} on split graphs is solvable in $O(m^{2.5})$~time for
  $s=2$.
\end{theorem}

\begin{proof}%
  Let $G=(C\cup V, C_E\cup E)$ be a split graph with $m$ being the
  number of edges in $E$.  
  Let $B'=(C\cup V \cup \{z\}, E\cup \{\{z,v\} \mid v\in C\})$ be a bipartite graph over $C$ and
  $I\cup\{z\}$.  Let $B^*$ be computed from $B'$ using \cref{cons:graph_for_matching}.  By
  \cref{lem:P3-partition==factor+contraint,lem:B=B',lem:B'has_a_factor==B*has_a_matching},
  $G$ admits a $P_3$-partition if and only if $B^*$ admits a perfect
  matching.  Since deciding whether a graph with $s$ vertices and $t$
  edges has a perfect matching can be done in $O(t\sqrt{s})$ time~\cite[Theorem~16.4]{Sch03}
  and since $B^*$ has $O(m)$ vertices and $O(m^2)$ edges, deciding
  whether $G$ has a $P_3$-partition can be done in $O(m^{2.5})$ time.
\end{proof}

\noindent In contrast, we can show that \probStarPart{} is NP-hard for each~$s\ge3$ by a reduction from \probExCover{}.

\begin{theorem}\label{thm:split-s>=3-NP-hard}
  \probStarPart{} on split graphs is NP-hard for $s\ge3$.
\end{theorem}

\begin{proof}
  We show that it is \classNP-hard to find an \starpartition{s} of a
  split graph via reduction from \probExCover \cite{GJ79} (illustrated
  in \cref{fig:splitgraph-reduction}).
  
  \decprob{\probExCover} {A finite set $U$ and a collection
    $\mathcal{S}$ of size-$s$ subsets of $U$.}  {Is there a
    subcollection $\mathcal{S}' \subseteq \mathcal{S}$ that partitions
    $U$ (each element of~$U$ is contained in exactly one subset in
    $\mathcal{S}'$)?}  

  \noindent Given $(U,\mathcal{S})$ with $|U|=sn$ and
  $|\mathcal{S}|=m \ge n$ where $n,m \in \mathbb{N}$, we construct a
  split graph~$G=~(C \cup I,E)$ as follows: The vertex set consists of
  a clique~$C$ and an independent set~$I$.  The clique~$C$ contains a
  vertex for each subset in~$\mathcal{S}$, the independent set~$I$
  contains a vertex for each element of~$U$.  For
  each~$S\in\mathcal{S}$ the corresponding vertex in~$C$ is adjacent
  to the~$s$ vertices in~$I$ that correspond to the elements
  of~$S$. Moreover, let~$q,r \in \mathbb{N}$ such that~$m-n = (s-1)q +
  r$. We add~$q$ dummy vertices to both~$C$ and~$I$ and connect every
  dummy vertex in~$C$ with all other vertices in~$C$ and uniquely with
  one of the dummy vertices in~$I$. Finally, we add one more dummy
  vertex to~$C$ and another~$s-r$ dummy vertices to~$I$. This last
  dummy in~$C$ is connected to all other vertices in~$C$ and to each
  of the~$s-r$ dummies in~$I$. Note that each dummy vertex in~$I$ has
  degree one. The above construction can be carried out in polynomial
  time.

  Now, let $\mathcal{S}' \subseteq \mathcal{S}$ be a partition of~$U$.
  Then we can partition~$G$ into stars of size~$s$ in the following
  way: For each~$S \in \mathcal{S}'$, we choose the star containing
  the vertex from~$C$ corresponding to~$S$ and the vertices in~$I$
  corresponding to the elements of~$S$. Moreover, each of the dummy
  vertices in~$I$ is put together with its neighboring dummy in~$C$
  and filled up to a star of size~$s$ with the remaining non-dummy
  vertices in~$C$ corresponding to the subsets in~$\mathcal{S}
  \setminus \mathcal{S}'$. Indeed, the values of~$q$ and~$r$ are
  chosen in a way that guarantees that this is
  possible. Since~$\mathcal{S}'$ partitions~$U$, we get a valid
  \starpartition{s} of~$G$.

  Conversely, in any \starpartition{s} of~$G$, all dummy vertices
  in~$I$ are grouped together into an $s$-star with their one dummy
  neighbor in~$C$ and the respective number of other non-dummy
  vertices from~$C$.  The values of~$q$ and~$r$ are such that there
  are exactly~$n$ non-dummy vertices left in~$C$ together with
  the~$s\cdot n$ vertices in~$I$ corresponding to~$U$.  It follows
  that each remaining non-dummy vertex in~$C$ forms a star with
  its~$s$ neighbors in~$I$, which yields a partition of~$U$.
\end{proof}

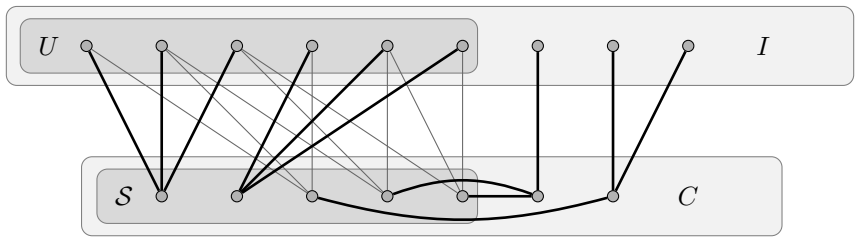
\begin{figure}[t]
  \centering \def\sep{2}
  \begin{tikzpicture}[draw=black!60]
    \tikzstyle{thickedge}=[draw=black, line width=1pt]
    \node (i-0) at (10,\sep) {$I$}; \node (u-0) at (.5,\sep) {$U$};
    \foreach \x in {1,2,...,9} \node[vertex] (i-\x) at (\x,\sep) {};

    \node (c-0) at (9,0) {$C$}; \node (s-0) at (1.5,0)
    {$\mathcal{S}$}; \foreach \x in {1,2,...,7} \node[vertex] (c-\x)
    at (\x+1,0) {};

    \draw (c-3) -- (i-1); \draw (c-3) -- (i-2); \draw (c-3) -- (i-4);
    \draw (c-4) -- (i-2); \draw (c-4) -- (i-3); \draw (c-4) -- (i-5);
    \draw (c-5) -- (i-3); \draw (c-5) -- (i-5); \draw (c-5) -- (i-6);
    \draw[thickedge] (c-6) -- (i-7); \draw[thickedge] (c-6)
    to[out=160, in=20] (c-4); \draw[thickedge] (c-6) -- (c-5);
    \draw[thickedge] (c-7) -- (i-8); \draw[thickedge] (c-7) -- (i-9);
    \draw[thickedge] (c-7) to[out=195, in=345] (c-3); \draw[thickedge]
    (c-1) -- (i-1); \draw[thickedge] (c-1) -- (i-2); \draw[thickedge]
    (c-1) -- (i-3); \draw[thickedge] (c-2) -- (i-4); \draw[thickedge]
    (c-2) -- (i-5); \draw[thickedge] (c-2) -- (i-6);

    \begin{pgfonlayer}{background}
      \node (C) [inner ysep=8pt, inner xsep=28pt, draw=black!50,rounded corners,fill=gray!10,fit=(c-0)
      (c-1)] {}; 
      \node (I) [inner ysep=8pt, inner xsep=28pt, draw=black!50,rounded
      corners,fill=gray!10,fit=(i-0) (i-1)] {}; 
      \node (U)
      [draw=black!50,rounded corners,fill=gray!30,fit=(u-0) (i-1) (i-6)]
      {}; 
      \node (S) [draw=black!50,rounded corners,fill=gray!30,fit=(s-0)
      (c-1) (c-5)] {};
    \end{pgfonlayer}
  \end{tikzpicture}
 \caption{Reduction from \probExCover to \probStarPart with $s=3$.
   Dummy vertices are in the light gray area and non-dummy vertices are in
   the dark gray area.  The \starpartition{3} for the constructed graph~$G$ is
   indicated by thick edges.}
 \label{fig:splitgraph-reduction}
\end{figure}

\section{Grid graphs}\label{sec:gridgraphs}

In this section, we show that \probPthreePart{} is NP-hard even on
grid graphs with maximum degree three, thus strengthening a result of
\citet{MZ05} and \citet{MT07}, who showed that \probPthreePart{} is
\classNP-complete on planar bipartite graphs of maximum degree three.

A \emph{grid graph} is a graph with a vertex set $V\subseteq \mathbb
N\times\mathbb N$ and edge set $\{\{u,v\}\mid u=(i,j)\in
V,v=(k,\ell)\in V, |i-k|+|j-\ell|= 1\}$. That is, its vertices can
be given integer coordinates such that every pair of vertices is
joined by an edge if and only if their coordinates differ by 1
in exactly one dimension.

\looseness=-1 To show NP-hardness of \probPthreePart{} on grid graphs, we exploit
the above mentioned result of \citet{MZ05} and \citet{MT07} and find a
suitable embedding of planar graphs into grid graphs while maintaining
the property of a graph having a $P_3$-partition. This allows us to
prove the following.

\begin{theorem}\label{NPhardgrid}
  \probPthreePart is \classNP-hard on grid graphs of maximum degree
  three.
\end{theorem}

\begin{figure}[t]
 \tikzstyle{vertex} = [align=center, circle, inner sep=1.5pt, fill=black]
  \tikzstyle{edge} = [color=black,opacity=.2,line
  cap=round, line join=round, line width=25pt]
  \centering \subfloat[Case 1: Vertex~$v$ and vertex~$w$ are covered by the same~$P_3$
    (which implies that one vertex is an endpoint and the other vertex is an internal
    point of the $P_3$).]{
    \begin{tikzpicture}[baseline=(v.south)]
       \tikzstyle{edge} = [color=black,opacity=.2,line
      cap=round, line join=round, line width=25pt]

      \node[vertex,label={above:$v$}] (v) {}; \node[vertex,right=of
      v,label={above:$w$}] (w) {};

      \node[below left=1cm of v] (a) {}; \node[above left=1cm of v]
      (b) {}; \node[left=1cm of v] (e) {}; \node[below right=1cm of w]
      (c) {}; \node[above right=1cm of w] (d) {}; \node[right=1cm of
      w] (f) {};

      \draw[thick] (v)--(w); \draw[dashed] (v)--(a); \draw[dashed]
      (v)--(b); \draw[dashed] (v)--(e); \draw[dashed] (w)--(c);
      \draw[dashed] (w)--(d); \draw[dashed] (w)--(f);

      \begin{pgfonlayer}{background}
        \draw[edge] (a)--(v.center)--(w.center);
      \end{pgfonlayer}
    \end{tikzpicture}$\iff$
    \begin{tikzpicture}[baseline=(v.south)]
      \node[vertex,label={above:$v$}] (v) {}; \node[vertex,right=of v]
      (n1) {}; \node[vertex,right=of n1] (n2) {};
      \node[vertex,right=of n2] (n3) {}; \node[vertex,right=of
      n3,label={above:$w$}] (w) {};

      \node[below left=1cm of v] (a) {}; \node[above left=1cm of v]
      (b) {}; \node[left=1cm of v] (e) {}; \node[below right=1cm of w]
      (c) {}; \node[above right=1cm of w] (d) {}; \node[right=1cm of
      w] (f) {};

      \draw[thick] (v)--(w); \draw[dashed] (v)--(a); \draw[dashed]
      (v)--(b); \draw[dashed] (v)--(e); \draw[dashed] (w)--(c);
      \draw[dashed] (w)--(d); \draw[dashed] (w)--(f);

      \begin{pgfonlayer}{background}
        \draw[edge] (a)--(v.center)--(n1.center); \draw[edge]
        (n2.center)--(w.center);
      \end{pgfonlayer}
    \end{tikzpicture}
  }

  \subfloat[Case 2: Vertex~$v$ and vertex~$w$ are covered by different $P_3$s
     (vertex~$v$ and/or vertex~$w$ can also be internal vertices).]{
    \begin{tikzpicture}[baseline=(v.south)]
      \node[vertex,label={above:$v$}] (v) {}; \node[vertex,right=of
      v,label={above:$w$}] (w) {};

      \node[below left=1cm of v] (a) {}; \node[above left=1cm of v]
      (b) {}; \node[left=1cm of v] (e) {}; \node[below right=1cm of w]
      (c) {}; \node[above right=1cm of w] (d) {}; \node[right=1cm of
      w] (f) {};

      \draw[thick] (v)--(w); \draw[dashed] (v)--(a); \draw[dashed]
      (v)--(b); \draw[dashed] (v)--(e); \draw[dashed] (w)--(c);
      \draw[dashed] (w)--(d); \draw[dashed] (w)--(f);

      \begin{pgfonlayer}{background}
        \draw[edge] (a)--(v.center); \draw[edge] (w.center)--(c);
      \end{pgfonlayer}
    \end{tikzpicture}$\iff$
    \begin{tikzpicture}[baseline=(v.south)]
      \node[vertex,label={above:$v$}] (v) {}; \node[vertex,right=of v]
      (n1) {}; \node[vertex,right=of n1] (n2) {};
      \node[vertex,right=of n2] (n3) {}; \node[vertex,right=of
      n3,label={above:$w$}] (w) {};

      \node[below left=1cm of v] (a) {}; \node[above left=1cm of v]
      (b) {}; \node[left=1cm of v] (e) {}; \node[below right=1cm of w]
      (c) {}; \node[above right=1cm of w] (d) {}; \node[right=1cm of
      w] (f) {};

      \draw[thick] (v)--(w); \draw[dashed] (v)--(a); \draw[dashed]
      (v)--(b); \draw[dashed] (v)--(e); \draw[dashed] (w)--(c);
      \draw[dashed] (w)--(d); \draw[dashed] (w)--(f);

      \begin{pgfonlayer}{background}
        \draw[edge] (a)--(v.center); \draw[edge]
        (n1.center)--(n3.center); \draw[edge] (w.center)--(c);
      \end{pgfonlayer}
    \end{tikzpicture}
  }
  \caption{All possibilities of two vertices~$v$ and~$w$ participating
    in a $P_3$-partition if they are joined by an edge or a path on
    three other degree-two vertices. Edges participating in the
    same~$P_3$ are grouped together in a gray background.}
  \label{fig:obsubdivide}
\end{figure}
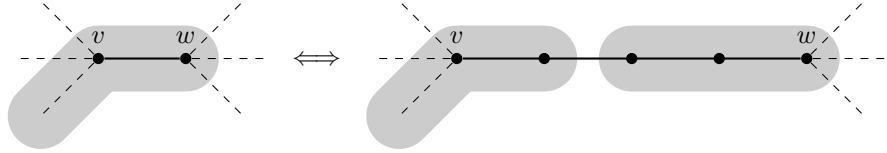
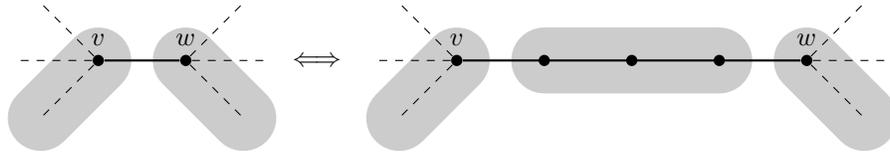
\noindent Towards proving \cref{NPhardgrid}, the following observation helps us
embed planar graphs into grid graphs: it allows us to replace edges
by paths on $3i$~new vertices for any $i\in\mathbb N$.

\begin{observation}\label[observation]{obsubdivide}
  Let $G$~be a graph, $e=\{v,w\}$ be an edge of~$G$, and $G'$ be the
  graph obtained by removing the edge~$e$ from~$G$ and by connecting
  $v$ and~$w$ using a path on three new vertices. Then, $G$ has a
  $P_3$-partition if and only if~$G'$ has one.
\end{observation}

\noindent Note that the correctness of \cref{obsubdivide} is proven by
\cref{fig:obsubdivide}, which enumerates all possible cases.

We can now prove \cref{NPhardgrid} by showing that
$G$~has a $P_3$-partition if and only~$G'$ has, where $G'$~is the
graph obtained from a planar graph~$G$ of maximum degree three using
the following construction.

\begin{construction}\label{cons:gridgraphs}
  \normalfont
  Let $G$~be a planar $n$-vertex graph of maximum degree three. Using a
  \pt{} algorithm of \citet{RT86} we obtain a crossing-free
  \emph{rectilinear embedding} of~$G$ into the plane such that:
  \begin{enumerate}
  \item Each vertex is represented by a horizontal line.
  \item Each edge is represented by a vertical line.
  \item All lines end at integer coordinates with integers in $O(n)$.
  \item If two vertices are joined by an edge, then the vertical line
    representing this edge ends on the horizontal lines
    representing the vertices.
  \end{enumerate}
  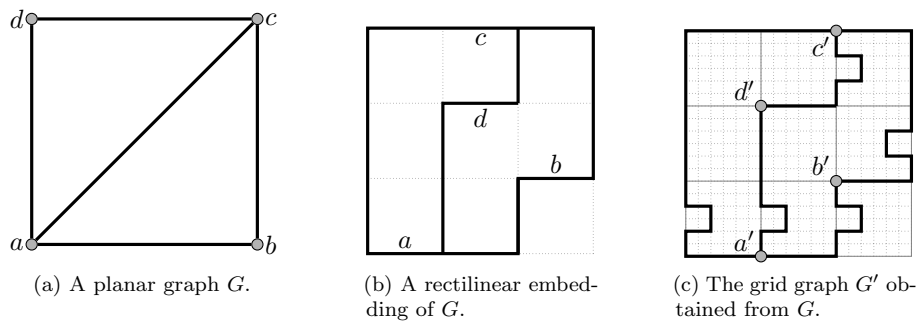
\begin{figure}[t]
    \centering \subfloat[A planar graph~$G$.]{\begin{tikzpicture}[label distance=-1mm] 

        \node[vertex, label={left:$a$}] (a) at (0,0) {}; \node[vertex,
        label={right:$b$}] (b) at (3, 0) {}; \node[vertex,
        label={right:$c$}] (c) at (3,3) {}; \node[vertex,
        label={left:$d$}] (d) at (0,3) {};

        \draw[very thick] (a) -- (b) -- (c) -- (d) -- (a); 
        \draw[very thick] (a) -- (c);
      \end{tikzpicture}\label{subfig:k4}}
    \hfill{} \subfloat[A rectilinear embedding
    of~$G$.]{\begin{tikzpicture}[very thick,label distance=-2mm]

        \draw[help lines, densely dotted] (0,0) grid (3,3);
        \draw (2,1) -- (2,0) -- (0,0) -- (0,3) -- (3,3) -- (3,1) --
        cycle; \draw (2,0) -- (2,1); \draw (2,2) -- (2,3); \draw (1,0)
        -- (1,2) -- (2,2);

        \node[label={above:$a$}] (a) at (0.5,0) {};
        \node[label={above:$b$}] (a) at (2.5,1) {};
        \node[label={below:$d$}] (a) at (1.5,2) {};
        \node[label={below:$c$}] (a) at (1.5,3) {};
      \end{tikzpicture}\label{subfig:rectili}}
    \hfill{} \subfloat[The grid graph~$G'$ obtained from~$G$.]{
      \begin{tikzpicture}[label distance=-1.5mm]
        \tikzstyle{edge} = [color=black,opacity=.2,line
        cap=round, line join=round, line width=5pt, very thick]
        \draw[help lines, densely dotted] (0,0)
        grid[step=1/6] (3,3); \draw[help lines] (0,0) grid (3,3);
        \draw[very thick] (2,1) -- (2,4/6) -- (2+2/6,4/6) -- (2+2/6,2/6) --
        (2,2/6) -- (2,0) -- (0,0) -- (0,2/6) -- (2/6,2/6) -- (2/6,4/6)
        -- (0,4/6) -- (0,3) -- (3,3) -- (3,1+4/6) -- (3-2/6,1+4/6) --
        (3-2/6,1+2/6) -- (3,1+2/6) -- (3,1) -- cycle;

        \begin{pgfonlayer}{background}
        
        \draw[edge] (0,0)--(2,0);
        \draw[edge] (2,1)--(3,1);
        \draw[edge] (1,2)--(2,2);
        \draw[edge] (0,3)--(3,3);
        
        \end{pgfonlayer}
        \draw[very thick] %
        (2,2) -- (2,2+2/6) -- (2+2/6,2+2/6) -- (2+2/6,2+4/6) --
        (2,2+4/6) -- (2,3); \draw[very thick] (1,0) -- (1,2/6) -- (1+2/6,2/6) --
        (1+2/6,4/6) -- (1,4/6) -- (1,2) -- (2,2);

        \node[vertex,label={above left:$a'$}] at (1,0) {};
        \node[vertex,label={below left:$c'$}] at (2,3) {};
        \node[vertex,label={above left:$d'$}] at (1,2) {};
        \node[vertex,label={above left:$b'$}] at (2,1) {};
      \end{tikzpicture}\label{subfig:grid}}
    \caption{Various embeddings of a planar graph. In the rectilinear
      embedding in \cref{subfig:rectili}, horizontal lines
      represent vertices of~$G$, while vertical lines represent its
      edges. In \cref{subfig:grid}, every intersection of a line
      with a grid point is a vertex, but only the vertices
      corresponding to vertices in \cref{subfig:k4} are
      shown. The horizontal lines of the rectilinear embedding
      are now replaced by paths highlighted in gray.}\label{fig:coolembedding}
  \end{figure}
  \cref{subfig:rectili} illustrates such an embedding. Without loss
  of generality, every end point of a line lies on another line. Now,
  in polynomial time, we obtain a grid graph~$G'$ from the rectilinear
  embedding, as follows:
  \begin{enumerate}
  \item We multiply all coordinates by six (see
    \cref{subfig:grid}).
    
  \item Every point in the grid touched by a horizontal line that
    represents a vertex~$v$ of~$G$ becomes a vertex in~$G'$. The
    horizontal path resulting from this horizontal line we denote
    by~$P(v)$ (indicated by a gray background in \cref{subfig:grid}).    

  \item For each vertical line, all its grid points become vertices
    in~$G'$, except for the third point from the bottom horizontal
    line that we bypass by adding a bend of
    five vertices to the vertical line (see \cref{subfig:grid}).
    
  \item With each vertex~$v$ in~$G$, we associate the vertex~$v'$
    of~$G'$ that lies on $P(v)$ and has degree three. There is at most
    one such vertex. If no such vertex exists, then we arbitrarily
    associate with~$v$ one of the end points of $P(v)$.
  \end{enumerate}
\end{construction}

\begin{proof}[Proof of \cref{NPhardgrid}]
  Since \cref{cons:gridgraphs} runs in polynomial time, it remains to
  prove that a graph~$G$ has a $P_3$-partition if and only if the
  graph~$G'$ obtained by \cref{cons:gridgraphs} has. By
  \cref{obsubdivide}, it is sufficient to verify that every
  edge~$e\coloneqq \{u,v\}$ in~$G$ is replaced by a path~$p$
  between~$u'$~and~$v'$ in~$G'$ whose number of inner vertices is
  divisible by three. %
  To this end, we partition the path~$p$ into two parts: one part
  consists of the subpaths~$p_u,p_v$ of~$p$ that lie on $P(u)$
  and~$P(v)$, respectively. Note that each of $p_u$ and $p_v$ might
  consist only of one vertex, as seen for the path from~$d'$ to~$a'$
  in \cref{subfig:grid}. The other part is a path~$p_e$ that
  connects $p_u$ to~$p_v$. We consider $p_e$ not to contain the
  vertices of~$p_u$ or~$p_v$. Hence, $p_e$ contains no vertices of any
  horizontal paths.

  The number of \emph{inner} vertices of~$p$ shared with the
  horizontal paths~$p_u$ and~$p_v$ is divisible by three (it is
  possibly zero) since all coordinates that start or end paths are
  divisible by three (in fact, by six). Herein, note that we do not
  count the vertices~$u'$ and~$v'$ lying on~$p_u$ or~$p_v$,
  respectively.

  Moreover, the number of vertices on~$p_e$ is also divisible by
  three: the number of vertices on a strictly vertical path~$p_s$
  connecting $p_u$ with $p_v$ would leave a remainder of two when
  divided by three (as the two vertices on $p_u$ and~$p_v$ are not
  considered to be part of~$p_s$). However, our added bend of five new
  vertices makes $p_e$ by four vertices longer compared
  to~$p_s$. Hence, the number of vertices on $p_e$ is also divisible
  by three. It follows that the total number of inner vertices of~$p$
  is divisible by three.
\end{proof}

\section{Chordal graphs}\label{sec:chordgraphs}
A graph is \emph{chordal} if every induced subgraph containing a cycle
of length at least four also contains a \emph{triangle}, that is, a
cycle of length three. We show that \probPthreePart restricted to
chordal graphs is \classNP-hard (in contrast to the \pt{}
solvability on split graphs which form
a subclass of chordal graphs) by reduction from
\probThreeDMatching.
More precisely, we use the construction that \citet{DF85}
provided to show that \probPthreePart{} is
NP-complete and observe that we can triangulate the resulting
graph while maintaining the correctness of the reduction.

\decprob{\probThreeDMatching~(3DM)} {Pairwise disjoint sets $\TDMR,
  \TDMB, \TDMY$ with $|\TDMR|=|\TDMB|=|\TDMY|=q$ and a set of
  triples~$\TDMT\subseteq \TDMR \times \TDMB \times \TDMY$.}  {Does
  there exist a \emph{perfect 3-dimensional matching} $M\subseteq T$,
  that is, $|M|=q$ and each element of $R\cup B \cup Y$ occurs in
  exactly one triple of $M$?}

\noindent \citet{DF85} introduced \cref{cons:TDM-to-P3Planar} described below and illustrated in
\cref{fig:gadget_c}).  Using it as a reduction from the \classNP-complete
restriction of 3DM to planar graphs~\cite{DF86},
they proved that \probPthreePart restricted to bipartite planar graphs is
\classNP-complete.
  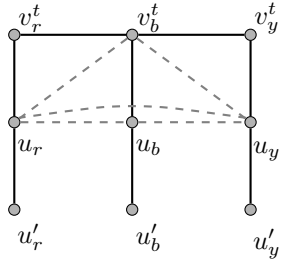
\begin{figure}[t]
    \centering
    \begin{tikzpicture}
      \node[] (n_vr) {$v^t_r$}; \node[right = of n_vr]
      (n_vb) {$v^t_b$}; \node[right = of n_vb] (n_vy) {$v^t_y$};
      
      \foreach \x in {r,b,y} { \node[vertex] at ($(n_v\x)+(-.22,
        -.22)$) (v\x) {}; \node[vertex, below = of v\x] (u\x) {};

        \node[below = 7.8ex of n_v\x] (n_u\x) {$u_{\x}$};
        \node[vertex, below = of u\x] (u'\x) {}; \node[below = 4.4ex
        of n_u\x] (n_u'\x) {$u'_{\x}$};

        \draw[-, thick] (v\x) -- (u\x); \draw[-, thick] (u\x) --
        (u'\x); }

      \draw[-, thick] (vr) -- (vb); \draw[-, thick] (vb) -- (vy);
      
      \draw[-, dashed, thick, gray] (ur) -- (ub) -- (uy); \draw[-,
      dashed, thick, gray] (vb) -- (ur); \draw[-, dashed, thick, gray]
      (vb) -- (uy);

      \draw[-, dashed, thick, gray] plot[smooth] coordinates
      {(uy.north west) ($(ub)+(.4, 0.2)$) %
        ($(ub)+(-.4,0.2)$) (ur.north east)};

    \end{tikzpicture}
    \caption{Gadget for a triple~$t=(r, b, y) \in \TDMT$ based on \cref{cons:TDM-to-P3Planar} (solid edges). 
    The reduction in the proof of \cref{thm:P3-hard-chordal}
    introduces the dashed edges.}
  \label{fig:gadget_c}
\end{figure}  
\begin{construction}\label{cons:TDM-to-P3Planar}
  \normalfont
  Let $\TDMInst$ with $|\TDMR|=|\TDMB|=|\TDMY|=q$ be an instance of
  3DM.  Construct a graph~$G=(V, E)$ as follows: For each
  element~$a\in \TDMR\cup \TDMB \cup \TDMY$, create two vertices~$u_a,
  u'_a$ and connect them by an edge~$\{u_a, u'_a\}$. We call~$u_a$
  an \emph{element-vertex} and $u'_a$~a \emph{pendant-vertex}. For
  each triple~$t=(r, b, y) \in \TDMT$, create three vertices~$v^t_{r},
  v^t_{b}, v^t_{y}$. We call these three vertices
  \emph{triple-vertices}. Make triple-vertex~$v^t_b$ adjacent to both
  $v^t_r$ and $v^t_y$. Also make triple-vertex~$v^t_r$ (respectively~$v^t_b$
  and~$v^t_y$) adjacent to element-vertex~$u_r$ (respectively~$u_b$
  and~$u_y$).
  Formally,
  \begin{align*}
    V&=\{u_a, u'_a\mid a\in \TDMR\cup \TDMB \cup \TDMY\} \cup \{v^t_r,
    v^t_b, v^t_y \mid t = (r, b, y) \in \TDMT\},\text{ and}\\
    E &= \{\{u_a, u'_a\} \mid a\in \TDMR\cup \TDMB \cup \TDMY\} \cup\\
    &\{\{u_r, v^t_r\}, \{u_b, v^t_b\}, \{u_y, v^t_y\} , \{v^t_r, v^t_b\},
    \{v^t_b, v^t_y\} \mid t =(r, b, y) \in \TDMT\}.
  \end{align*}
\end{construction}
  
\begin{theorem}\label{thm:P3-hard-chordal}
  \probPthreePart restricted to chordal graphs is \classNP-hard.
\end{theorem}

\begin{proof}
  We extend \cref{cons:TDM-to-P3Planar} to show the \classNP-hardness
  of \probPthreePart restricted to chordal graphs.  Make any two
  element-vertices adjacent to each other such that the graph induced
  by all element-vertices is complete. Furthermore, for each
  triple~$t=(r, b, y)\in \TDMT$, add two edges~$\{v^{t}_b, u_r\}$ and
  $\{v^{t}_b, u_y\}$ to the graph, as illustrated in
  \cref{fig:gadget_c}.  Let $G$ be the resulting graph.

  We first show that $G$ is chordal.  Consider any size-$\ell$ set~$C$
  of vertices with~$\ell \ge 4$ such that the subgraph $G_C$ induced
  by~$C$ contains a simple cycle of length~$\ell$.  Since the
  pendant-vertices all have degree one, $C$ cannot contain any
  pendant-vertex.  If~$C$~contains a degree-two triple-vertex~$v_r^t$
  (respectively~$v_y^t$) for some $t=(r,b,y)\in \TDMT$, then~$C$ contains
  both its neighbors:~$v_b^t$~and~$u_r$ (respectively $v_b^t$ and~$u_y$)
  which are connected, that is, forming a triangle.  Otherwise, if~$C$
  contains a degree-five triple-vertex~$v_b^t$ but does not
  contain~$v_r^t$ nor~$v_y^t$, then, in the cycle, $v_b^t$ lies
  between two element-vertices.  Since two element-vertices are always
  connected, $G_C$~contains a triangle.  Finally, if~$C$ does not
  contain any triple-vertex, then it is included in the set of
  element-vertices, which form a complete graph. Hence, $G_C$ contains
  a triangle.

  Second, we show that $\TDMInst$ has a perfect 3-dimensional matching
  if and only if $G$ can be partitioned into $P_3$s.

  For the ``only if'' part, suppose that $\TDMSol\subseteq T$ is a perfect
  3-dimensional matching for $\TDMInst$. Then, the $P_3$s in
  \[\Bigl(\smashoperator{\bigcup_{t=(r,b,y)\notin \TDMSol}}{\{\{v^t_{r}, v^t_{b},
    v^t_{y}\}\}}\Bigr) \cup \Bigl(\smashoperator{\bigcup_{t=(r,b,y) \in \TDMSol}}
  \{\{u'_r, u_r, v^t_r\}, \{u'_b, u_b, v^t_b\}, \{u'_y, u_y,
  v^t_y\}\}\Bigr)\] indeed partition the graph~$G$.

  For the ``if'' part, suppose that $G$ has a partition~$P$ into
  $P_3$s.  We first enumerate the possible centers of the $P_3$s.
  Since each pendant-vertex is only adjacent to its element-vertex,
  every element-vertex is the center of a $P_3$ that contains a
  pendant-vertex. We call such a $P_3$ an \emph{element-$P_3$}.  For
  each triple~$t=(r, b, y)\in \TDMT$, neither $v^t_r$ nor $v^t_y$ can
  be the center of a $P_3$ since they are adjacent to only one vertex
  which is not already a center (namely $v^t_b$).  Thus, any $P_3$
  which is not an element-$P_3$ must have vertex~$v^t_b$ as a center
  for some $t=(r,b,y)\in \TDMT$. We call such a $P_3$ a \emph{triple-$P_3$
    corresponding to $t$}.

  Now, consider a triple $t=(r,b,y)\in \TDMT$.  If there exists a
  triple-$P_3$ corresponding to~$t$ (that is, with center $v^t_b$), then
  its two leaves can only be $v^t_r$ and $v^t_y$ (since $u^t_r$,
  $u^t_b$, and $u^t_y$ are centers).
  Otherwise, each of the three triple-vertices $v^t_r$,
  $v^t_b$, and $v^t_y$ must be a leaf of an element-$P_3$ centered on
  $u_r$, $u_b$, and $u_y$.  Indeed, there is only one way to match the
  three triple-vertices to these three element-vertices, that is, by
  using the edges $\{v^t_r,u_r\}$, $\{v^t_b,u_b\}$ and $\{v^t_y,u_y\}$. As
  a consequence, the leaves of the element-$P_3$ centered on~$u_a$
  are~$u'_a$ and $v^t_a$ for some triple $t$ containing~$a$.

  It remains to show that the triples with \emph{no} corresponding triple-$P_3$
  in $P$ form a perfect 3-dimensional matching~$\TDMSol$ for
  $\TDMInst$.  Note that for each element $a\in \TDMR \cup \TDMB \cup
  \TDMY$, the element-$P_3$ centered in~$u_a$ uses a
  triple-vertex~$v^t_a$ for some triple $t$ containing $a$, which
  means that no triple-$P_3$ in $P$ corresponds to~$t$. Hence, $t\in
  \TDMSol$ and element~$a$ is matched by~$t$.  Now, it remains to show
  that every element is matched at most once. Suppose for the sake of
  contradiction that there is an element~$a\in \TDMR \cup \TDMB \cup
  \TDMY$ which is matched at least twice. To this end, let $v^t_a$~be
  a triple-vertex that together with~$u_a$ and $u'_{a}$ forms an
  element-$P_3$. Thus, $t\in \TDMSol$. Furthermore, let $t'$ be another
  triple in~$\TDMSol$ that matches~$a$. Since $t'$ has no corresponding
  triple-$P_3$ in $P$, there is an element-$P_3$ containing~$v^{t'}_a$.
  But then, $v^{t'}_a$ must form an element-$P_3$ together with $u_a$ and $u'_{a}$,
  which is a contradiction.
\end{proof}

\section{Conclusion}
We close with three open questions for future research.  What is the
complexity of \probStarPart for~$s\geq 2$ on permutation graphs?  What
is the complexity of \probStarPart for~$s\geq 3$ on interval graphs?
Are there other important graph classes (not necessarily perfect ones)
where \probStarPart is \pt{} solvable?

\paragraph{Acknowledgments.}
Ren\'e van Bevern was supported by the Russian Foundation for Basic Research (RFBR), project~16-31-60007 mol\textunderscore{}a\textunderscore{}dk, at Novosibirsk State University, and by the German Research Foundation (DFG), project DAPA (NI 369/12), at TU Berlin. 
Robert Bredereck was supported by the DFG, project PAWS (NI 369/10).
Laurent Bulteau and Gerhard J.\ Woeginger were supported by the Alexander von Humboldt Foundation, Bonn, Germany, while visiting TU~Berlin.
Jiehua Chen was supported by the Studienstiftung des Deutschen Volkes.
Vincent Froese was supported by the DFG, project DAMM (NI 369/13).
This work started at the yearly research retreat of the group ``Algorithms and Computational Complexity TU~Berlin'' in March 2013, 
held in Bad Schandau, Germany.

\bibliographystyle{abbrvnat}
\bibliography{star_jgt}

\end{document}